%% file: main.tex
\documentclass[pra,onecolumn,floatfix,superscriptaddress,longbibliography, nofootinbib]{revtex4-2}

\usepackage[utf8]{inputenc}
\usepackage[left]{lineno}
\usepackage[english]{babel}
\usepackage[colorlinks=true,citecolor=teal,linkcolor=teal,urlcolor=teal]{hyperref}

\usepackage{
    algorithm,
    algpseudocode,
    amsmath,
    amssymb,
    amsthm,
    array,
    bbm,
    csquotes,
    dsfont,
    graphicx,
    hyperref,
    lipsum,
    mathtools,
    multirow,
    nicefrac,
    qcircuit,
    slashed,
    xcolor,
    xspace,
    times,
    braket,
    physics
}  

\usepackage{chngcntr}
\usepackage[capitalize]{cleveref}

\newtheorem{theorem}{Theorem}

\newtheorem{lemma}[theorem]{Lemma}

\newtheorem{corollary}[theorem]{Corollary}

\newtheorem{observation}{Observation}

\crefname{observation}{observation}{observations}
\Crefname{observation}{Observation}{Observations}

\newenvironment{example}[1][Example]{\begin{trivlist}
\item[\hskip \labelsep {\bfseries #1}]}{\end{trivlist}}

\renewcommand{\vec}[1]{\boldsymbol{#1}}

\input{commands.tex}

\graphicspath{ {images/} }

\begin{document}

\title{A PAC-Bayesian approach to generalization for quantum models}

\newcommand{\FU}{Dahlem Center for Complex Quantum Systems,
Freie Universität 
Berlin, 
14195 Berlin, Germany}
\newcommand{\hhi}{Fraunhofer Heinrich Hertz Institute, 10587 Berlin, Germany}
\newcommand{\hzb}{Helmholtz-Zentrum Berlin f{\"u}r Materialien und Energie, 14109 Berlin, Germany}
\newcommand{\ubc}{Department of Physical Chemistry, University 
of the Basque Country UPV/EHU, 48080 Bilbao, Spain}
\newcommand{\ehu}{EHU Quantum Center, University of 
the Basque Country UPV/EHU, 48080 Bilbao, Spain}

\newcommand{\tec}{TECNALIA, Basque Research and Technology Alliance (BRTA), 48160 Derio, Spain}
\newcommand{\war}{Department of Computer Science, University of Warwick, UK}

\author{Pablo Rodriguez-Grasa}
\altaffiliation{These authors contributed equally to this work.}
\affiliation{\ubc}
\affiliation{\ehu}
\affiliation{\tec}

\author{Matthias C. Caro}
\altaffiliation{These authors contributed equally to this work.}
\affiliation{\war}

\author{Jens Eisert}
\affiliation{\FU}
\affiliation{\hhi}
\affiliation{\hzb}

\author{Elies Gil-Fuster}
\affiliation{\FU}
\affiliation{\hhi}

\author{Franz J. Schreiber}
\affiliation{\FU}

\author{Carlos Bravo-Prieto}
\affiliation{\FU}

\begin{abstract}

\noindent Generalization is a central concept in machine learning theory, yet for quantum models, it is predominantly analyzed through uniform bounds that depend on a model's overall capacity rather than the specific function learned. These capacity-based uniform bounds are often too loose and entirely insensitive to the actual training and learning process. Previous theoretical guarantees have failed to provide non-uniform, data-dependent bounds that reflect the specific properties of the learned solution rather than the worst-case behavior of the entire hypothesis class. To address this limitation, we derive the first PAC-Bayesian generalization bounds for a broad class of quantum models by analyzing layered circuits composed of general quantum channels, which include dissipative operations such as mid-circuit measurements and feedforward. Through a channel perturbation analysis, we establish non-uniform bounds that depend on the norms of learned parameter matrices; we extend these results to symmetry-constrained equivariant quantum models; and we validate our theoretical framework with numerical experiments. This work provides actionable model design insights and establishes a foundational tool for a more nuanced understanding of generalization in quantum machine learning.
\end{abstract}

\maketitle

\section{Introduction}\label{sec:introduction}

Quantum devices hold the promise of addressing computational problems that lie beyond the reach of classical computers. While the emerging field of \emph{quantum machine learning} (QML) offers potential improvements over classical algorithms \cite{biamonte2017quantum,RevModPhys.91.045002,PACLearning,TemmeML,OurDiffusionModels,MindTheGaps}, the practical capabilities and potential advantages of variationally trained, layered quantum models \cite{McClean_2016,Variational} remain poorly understood. To better assess their potential for algorithmic advantage, recent efforts have increasingly focused on the expressivity, trainability, and classical simulatability of these models \cite{Dequantization,PhysRevLett.131.100803, masot2025prospects}.

A concept that is increasingly recognized as a key factor in understanding the performance of quantum models is that of \emph{generalization}. However, most rigorous generalization guarantees for quantum models currently mirror classical capacity-based analyses. These \emph{uniform} bounds control the worst-case generalization gap across an entire hypothesis class in terms of measures such as the number of independent parameters, covering numbers, pseudo-dimension, or variants of Rademacher complexity adapted to quantum settings~\cite{caro2022generalization,huang2021power,caro2020pseudo,Banchi2021,Caro2021encodingdependent,abbas2021power,bu2021rademachercomplexitynoisyquantum,caro2023out,Du2022,hur2024understanding,wu2026generalization}. Because they depend solely on overall model capacity, they fail to account for the specific function selected during the training process. As seen in classical deep learning, this insensitivity is well known to produce pessimistic, often vacuous bounds in overparameterized regimes where models interpolate training data yet still generalize well~\cite{zhang2017_rethinking, Zahn2021_still_rethinking}. Recent work suggests that quantum models can exhibit analogous phenomena, for example perfectly fitting random labels in quantum phase recognition tasks~\cite{gil2024understanding}. Together, these observations motivate the search for non-uniform, data-dependent bounds for QML that can reflect properties of the learned solution rather than worst-case behavior of the entire model class. Crucially, these non-uniform guarantees offer a path toward principled model design by identifying specific architectural features that actively promote generalization.

In classical machine learning, PAC-Bayesian approaches have provided a flexible and remarkably successful path to non-uniform generalization guarantees: PAC-Bayesian inequalities tie a predictor's true risk to its empirical risk plus a complexity term determined by the \emph{Kullback-Leibler} (KL) divergence between a distribution over the learned parameters (the posterior) and a data-independent reference distribution (the prior). By tailoring the posterior to concentrate around the trained solution one can obtain bounds that depend on norms or spectral properties of the learned parameters and correlate well with empirical generalization~\cite{mcallester1999pac, McAllester2003, neyshabur2015norm, bartlett2017spectrally, perturbation_bound, jiang2019fantasticgeneralizationmeasures}. Motivated by this, in this work we derive the first PAC-Bayesian generalization bounds for a broad class of QML models. 
The quantum models we consider are structured as \emph{layered} quantum circuits, where each layer is defined by a parameterized quantum channel. Our approach emphasizes expressive models incorporating dissipation and feedforward dynamics, such as recently introduced dynamic quantum circuits~\cite{deshpande2024dynamicparameterizedquantumcircuits}. Moreover, this formalism also lends itself to principled incorporation of inductive biases such as symmetry constraints. Choosing priors that are supported on symmetry-constrained subspaces (i.e., priors that respect equivariance) reduces the effective KL penalty and yields provably tighter bounds for equivariant quantum models, thus quantifying how symmetry reduces effective complexity in a data-dependent manner.

Taken together, our results provide a new, broadly applicable analytical toolkit for QML (see Fig.~\ref{fig:framework}): a framework that allows a unified treatment of a broad class of 
quantum channels under a 
PAC-Bayesian lens, yields operationally interpretable, solution-dependent complexity terms, gives concrete guidance for model design, and more importantly, provides a more nuanced understanding of generalization in quantum machine learning.

\begin{figure*}[t]
        \centering
          \includegraphics[width=.8 \textwidth]{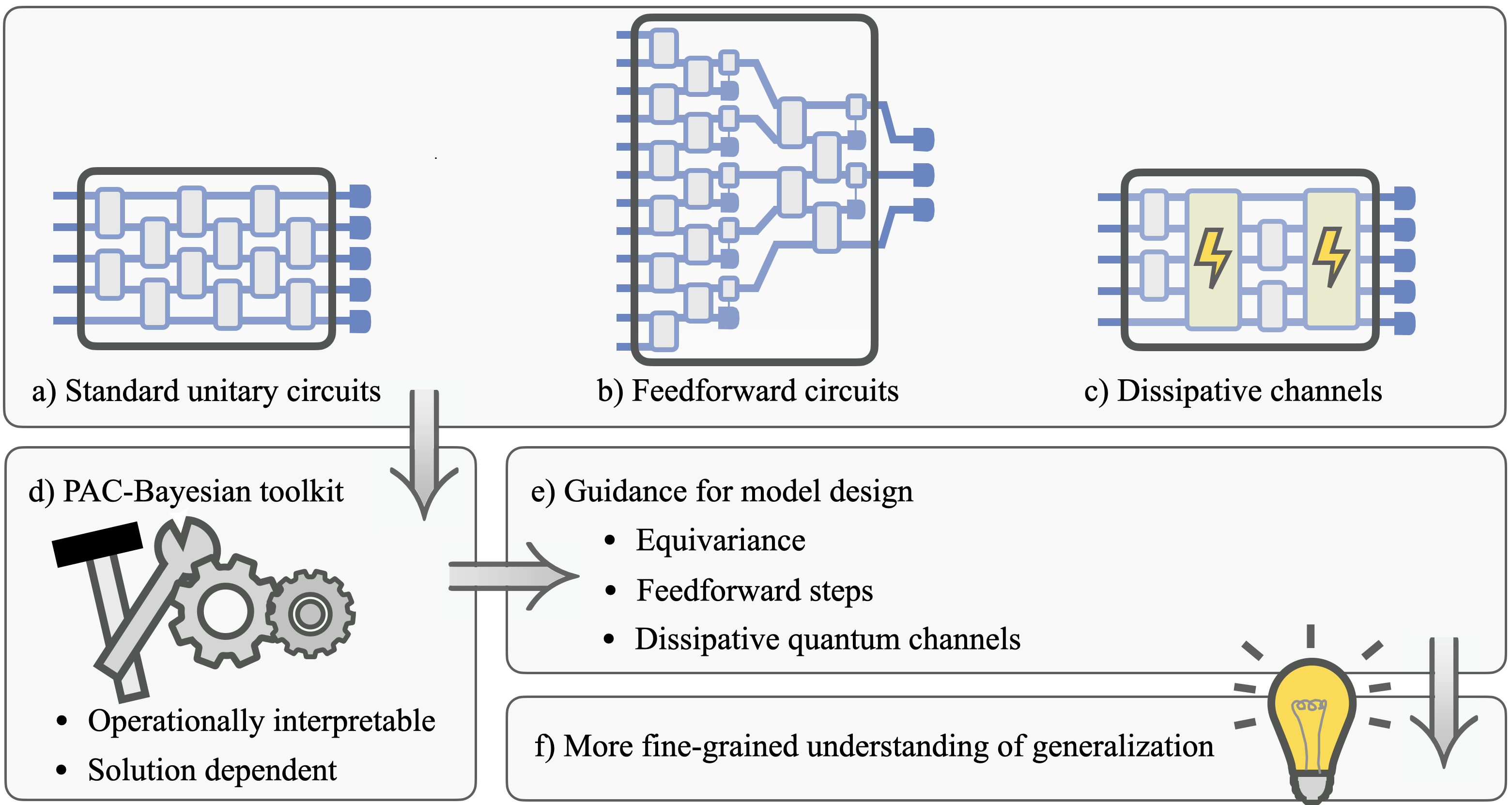}
        \caption{
            \textbf{Visualization of the proposed PAC-Bayesian framework for QML.} The framework unifies a broad spectrum of quantum architectures, including \textbf{(a)} standard unitary circuits, \textbf{(b)} mid-circuit measurement-and-feedforward architectures, and \textbf{(c)} dissipative channels. \textbf{(d)} Under the PAC-Bayesian lens, this methodology produces operationally interpretable and solution-dependent bounds that serve two main purposes: \textbf{(e)} providing theoretical guidance for model design, and \textbf{(f)} establishing a more nuanced understanding of generalization in quantum models. }
         
        \label{fig:framework}
    \end{figure*}

\section{Preliminaries}\label{s:pre}

In this section, we introduce the foundational concepts and notation used throughout this work. We begin by defining the supervised learning task and loss function. Next, we present the mathematical formalisms used to represent layered quantum models, specifically detailing the process matrix, Pauli transfer matrix, and equivariant channel representations. Finally, we review the classical PAC-Bayesian framework that underpins the derivation of our generalization bounds for quantum models.

\subsection{Learning setup}
Consider a supervised multi-class classification task with input domain $\calX$ and output classes labeled $y \in \{1,\dots,K\}$. A quantum model aims to solve this task by mapping an input $\vec{x} \in \calX$ to a probability distribution over the classes. First, the classical input $\vec{x}$ is encoded into an initial quantum state $\rho(\vec{x})$ in a Hilbert space $\calH_{\mathrm{in}}^{(1)}$. The model consists of $L$ layers, where each layer is characterized by a set of parameters. We denote the vector of all learned parameters across these layers as $\vec{w}$. The model then maps the input state $\rho(\vec{x})$ to a final output state $\rho_{\mathrm{\mathrm{out}}}(\vec{x})$ in $\calH_{\mathrm{\mathrm{out}}}^{(L)}$ via the composition of these parameterized layers, which we will formalize in the subsequent section. Finally, a measurement is performed on $\rho_{\mathrm{\mathrm{out}}}(\vec{x})$ to produce the output vector $f_{\vec{w}}(\vec{x}) \in \mathbb{R}^K$. Here, we assume without loss of generality that $\calH_{\mathrm{\mathrm{out}}}^{(L)}$ is $K$-dimensional, and we model our measurement as a standard computational basis measurement described by the \emph{positive operator-valued measure} (POVM) $\{\Pi_k = \ket{k}\bra{k}\}_{k=1}^K$. The $k$-th entry of the output vector represents the probability of obtaining outcome $k$ given by
\begin{equation}
    f_{\vec{w}}(\vec{x})[k] = \mathrm{Tr}(\Pi_k \rho_{\mathrm{\mathrm{out}}}(\vec{x}))\,.
\end{equation}
By construction, $f_{\vec{w}}(\vec{x})[k] \geq 0$ and $\sum_{k=1}^K f_{\vec{w}}(\vec{x})[k] = 1$.

We consider a classification task where the goal is to predict the correct label $y \in \{1,\dots,K\}$ for a given input $\vec{x}$. For a distribution $\calD$ over labeled examples $(\vec{x}, y)$ and a margin parameter $\gamma > 0$, the expected margin loss is defined as
\begin{equation}\label{eq:margin_loss}
    L_{\gamma}(f_{\vec{w}})=\bbP_{(\vec{x}, y)\sim \calD}\left[f_{\vec{w}}(\vec{x})[y]\leq \gamma + \max_{k\neq y}f_{\vec{w}}(\vec{x})[k]\right]\,.
\end{equation}
This loss measures the probability that the probability assigned to the correct class $y$ does not exceed the probability of the most likely incorrect class by at least the margin $\gamma$. The empirical estimate of this loss is denoted as $\hL_{\gamma}(f_{\vec{w}})$, where we drop the explicit dependence on the samples for notational simplicity. Setting $\gamma = 0$ recovers the standard classification loss, for which we write $L_0$ and $\hL_0$ as the expected risk and empirical risk, respectively. By definition, both $L_\gamma$ and $\hL_\gamma$ take values in the interval $[0, 1]$.

\subsection{Representations of quantum channels}

As introduced previously, we consider a quantum model as a sequence of $L$ layers processing an initial state $\rho(\vec{x})$. Unlike standard approaches that often restrict layers to unitary evolutions, we treat each layer $j$ as a general quantum channel $\phi_j$, which maps density operators on an input Hilbert space $\calH_{\mathrm{in}}^{(j)}$ to density operators on an output Hilbert space $\calH_{\mathrm{out}}^{(j)}$ (where $\calH_{\mathrm{out}}^{(j)} = \calH_{\mathrm{in}}^{(j+1)}$ for $j < L$). This framework captures a significantly broader class of architectures than standard parameterized quantum circuits (PQCs) analysis, including those with dissipative dynamics, mid-circuit measurements, and feedforward operations.
Crucially, we employ these formalisms not as prescriptions for the physical implementation of the circuit, but as mathematical representations of the \emph{effective map} implemented by each layer. 

The overall quantum model maps the input state $\rho(\vec{x})$ to the final output state $\rho_{\mathrm{out}}(\vec{x})$ via composition,
\begin{equation}
    \rho_{\mathrm{out}}(\vec{x}) = \left(\phi_L \circ \dots \circ \phi_1\right)\rho(\vec{x})\,.
\end{equation}
To analyze these channels, we employ two distinct formalisms: the \emph{process matrix} (PM) and the \emph{Pauli transfer matrix} (PTM). In these formulations, we characterize the action of $\phi_j$ via a weight matrix $W_j$ that models the deviation from the maximally depolarizing channel $\phi_\mathrm{DEP}(\rho)=\mathrm{Tr}(\rho)\mathbb{I}/d$, where $d$ is the dimension of the Hilbert space. This establishes a physical baseline for our analysis: a channel with $W_j=0$ yields the maximally mixed state regardless of the input, representing a model with no capacity to distinguish data. Consequently, our PAC-Bayes bounds will naturally penalizes models based on how far they drift from this baseline. The set of weight matrices $\mathrm{vec}(\{W_j\}_{j=1}^L)$ collectively parameterizes the quantum model, and we denote the vector of all these parameters as $\vec{w}$.

\subsubsection{Process matrix formalism}\label{sss:pm_formalism}
The PM framework, while requiring equal input and output dimensions ($d^{(j)}_\mathrm{in}=d^{(j)}_\mathrm{out}=2^n$, where $n$ is the number of qubits), offers a compact representation. We denote the action of a layer $\phi_j=\phi_\mathrm{PM}^{W_j}$ as
\begin{equation}
    \phi_\mathrm{PM}^{W_j}(\rho) = \sum_{A,B} \chi(A,B) \sigma_A \rho \sigma_B\,,
\end{equation}
where $\sigma_A,\sigma_B$ denote Pauli matrices, and 
\begin{equation}
\chi(A,B)=\left(\frac{\delta_{A,B}}{4^n} + W_j(A,B)\right)\,,
\end{equation}
is the \emph{process matrix} of the channel $\phi_\mathrm{PM}^{W_j}$. At $W_j=0$, we correctly recover the maximally depolarizing channel $\phi_\mathrm{DEP}=\phi_\mathrm{PM}^0$. To quantify the complexity of the learned models in our subsequent bounds, we will use specific norms of the parameter matrices. Concretely, we will rely on the $0$-norm (i.e., the sparsity) $\xi_j :=\norm{W_j}_0$, the Frobenius norm $\norm{W_j}_F$, and the entry-wise $1$-norm $\norm{W_j}_{1,1}$.

\subsubsection{Pauli transfer matrix formalism}\label{sss:ptm_formalism}
The PTM framework, which was recently employed for quantum learning in \cite{caro2212learning}, allows an independent choice of input and output dimensions. Let $A \in \{0,1,2,3\}^{n_\mathrm{\mathrm{out}}}$ and $B \in \{0,1,2,3\}^{n_\mathrm{in}}$. The action of layer $\phi_j=\phi_{\mathrm{PTM}}^{W_j}$ is given by
\begin{equation}
    \phi_{\mathrm{PTM}}^{W_j}(\rho) = \frac{1}{\sqrt{d^{(j)}_\mathrm{in} d^{(j)}_\mathrm{\mathrm{out}}}} \sum_{A,B} R(A,B)\, \mathrm{Tr}(\sigma_B \rho) \sigma_A\,,
\end{equation}
where
\begin{equation}
R(A,B)= \left(\delta_{A,0} \delta_{B,0} \sqrt{\frac{d^{(j)}_\mathrm{in}}{d^{(j)}_\mathrm{\mathrm{out}}}} + W_j(A,B)\right)\,,
\end{equation}
is the \emph{Pauli transfer matrix} of the channel $\phi^{W_j}_{\text{PTM}}$. Note that, even though the same notation is used, the weight matrices $W_j$ in the PM and PTM frameworks are different. However, the specific context will always clarify which formalism is in use. Analogous to the PM representation, our generalization bounds for this framework will depend on the sparsity $\xi_j :=\norm{W_j}_0$, the Frobenius norm $\norm{W_j}_F$, and the entry-wise $1$-norm $\norm{W_j}_{1,1}$ of these PTM weight matrices.

\subsubsection{Equivariant channel formalism}\label{sss:equivariant_formalism}
Finally, when a task exhibits symmetries described by a compact group $G$, we explicitly model the circuit structure to respect the underlying group action. An $L$-layer equivariant quantum model is defined by a sequence of unitary representations $(R^{\mathrm{in}}, R^{(1)},\dots, R^{(L)} = R^{\mathrm{out}})$ acting on the Hilbert spaces connecting the layers. The choice of intermediate representations $R^{(j)}$ dictates the internal symmetry spaces of the model, acting as an architectural hyperparameter that determines how symmetric features are processed, as illustrated in Fig.~\ref{fig:structure_EQNN}(a).

Each layer $\phi_j$ is a quantum channel mapping states from layer $j-1$ to layer $j$. Crucially, $\phi_j$ is required to be $(G, R^{(j-1)}, R^{(j)})$-equivariant, meaning its operation commutes with the group action: 
\begin{equation}
\phi_j \circ \mathrm{Ad}_{R^{(j-1)}}(g) = \mathrm{Ad}_{R^{(j)}}(g) \circ \phi_j \quad \text{for all } g \in G\,.
\end{equation}
This condition imposes constraints on the parameters of the channel. Specifically, the channel's Choi operator $J^{\phi_j}$ must lie in the commutant algebra of the product representation $R^{(j-1)\ast} \otimes R^{(j)}$.

To parameterize this efficiently, we employ the isotypic decomposition~\cite{ragone2022representation}, which decomposes the joint input-output space into a direct sum of \emph{irreducible representations} (irreps) of $G$ as $\calH_{\mathrm{in}} \otimes \calH_{\mathrm{out}} \simeq \bigoplus_\lambda \calH_\lambda \otimes \calH_{m_{\lambda}}$. Here, $\lambda$ indexes the unique irreps contained in the tensor product $R^{(j-1)\ast} \otimes R^{(j)}$, $d_\lambda$ is the dimension of the irrep space, and $m_\lambda$ is the dimension of the multiplicity space.
Because the Choi operator lies in the commutant algebra, Schur's Lemma guarantees that in the basis aligned with this isotypic decomposition, $J^{\phi_j}$ takes a strictly block-diagonal form, as seen in Fig.~\ref{fig:structure_EQNN}(b). We therefore define the parameterization for layer $j$ directly in this irrep basis. Without symmetry constraints, representing the channel would require a fully parameterized matrix across the entire joint space. By enforcing equivariance, however, the layer is completely defined by a smaller set of parameter matrices $\{W_{j,\lambda}\}_\lambda$ that act exclusively on the multiplicity spaces. The Choi operator takes the specific form
\begin{equation} \label{eq:equivariant_param} 
J^{\phi_j} \simeq \bigoplus_{\lambda} \left( \mathbb{I}_{d_{\lambda}} \otimes (J_{0,\lambda} + W_{j,\lambda}) \right)\,. 
\end{equation}
Here, $J_{0,\lambda} = \mathbb{I}_{m_\lambda}/{d^{(j)}_{\mathrm{out}}}$ represents the maximally depolarizing channel restricted to this symmetry sector, serving as a fixed, equivariant baseline. The matrices $W_{j,\lambda}$ (of size $m_\lambda \times m_\lambda$) capture the learned deviations from this baseline within each symmetry sector. This representation ensures that the model is strictly equivariant while significantly reducing the number of free parameters compared to a general quantum channel. For our theoretical guarantees, we will evaluate the complexity of these models using symmetry-adapted measures, specifically the number of equivariant parameters $\Xi_j = \sum_\lambda m^2_{j,\lambda}$, as well as the equivariant trace and Frobenius norms defined over the irrep blocks $W_{j,\lambda}$.

\begin{figure*}[t]
        \centering
        \includegraphics[width=.9\textwidth]{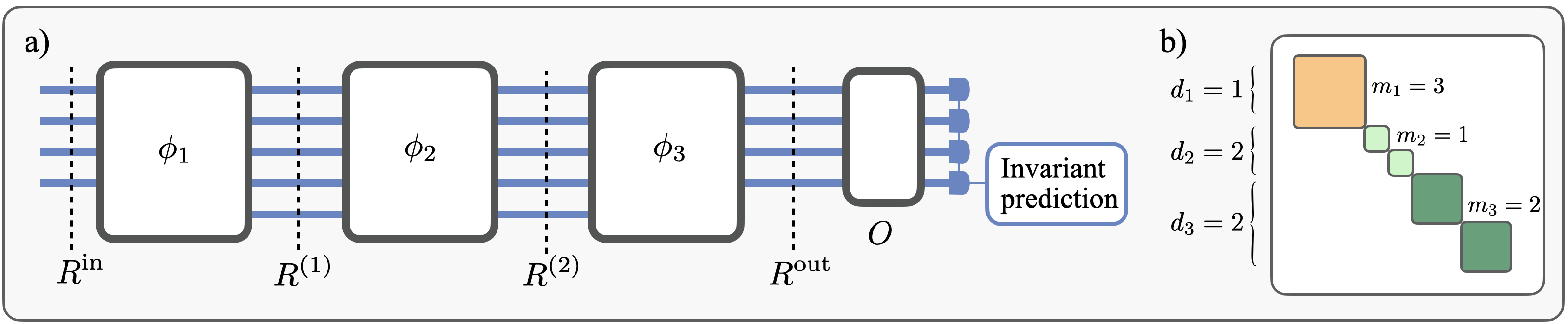}
        \caption{
            \textbf{Architecture and parameterization of equivariant quantum models.} \textbf{(a)}  Schematic of a multi-layer equivariant quantum model. The model consists of a sequence of quantum channels $\phi_j$ and unitary representations $R^{(j)}$ that define the symmetry constraints at each layer. The implementation concludes with an observable $O$ yielding an invariant prediction. \textbf{(b)} The block-diagonal structure of a channel's Choi operator arising from the isotypic decomposition. By virtue of Schur's Lemma, the parameters decouple into independent blocks acting on the multiplicity spaces. Here, the colored squares represent the learnable parameter matrices $W_{j,\lambda}$ of size $m_\lambda \times m_\lambda$, while the repetition count $d_\lambda$ corresponds to the 
            dimension of the associated 
            irreducible representation.
        } 
        \label{fig:structure_EQNN}
    \end{figure*}

\subsection{PAC-Bayesian framework}\label{ss:framework}
The \emph{PAC-Bayesian framework}~\cite{mcallester1998some, mcallester1999pac} provides generalization guarantees, often for randomized predictors. Instead of analyzing a single hypothesis $f_{\vec{w}}$ learned from data, this approach considers a distribution over predictors, termed the posterior distribution $Q$, which typically depends on the training data. The framework bounds the risk of hypotheses drawn from $Q$ by relating it to their empirical performance and to the 
\emph{Kullback-Leibler} (KL) divergence between $Q$ and a 
data-independent prior distribution $P$.

A common strategy involves defining the posterior $Q$ as a distribution centered around the learned parameters $\vec{w}$, often by adding random noise. Let $\vec{w}$ represent the parameters learned from the training data. Consider a random perturbation $\vec{u}$, and let $Q$ be the distribution of the pertubed parameters $\vec{w} + \vec{u}$. A foundational PAC-Bayes result bounds the expected risk under $Q$.

\begin{lemma}[PAC-Bayes bound for randomized predictors~\cite{McAllester2003}]\label{lemma:pac-bayes_rand}
        Let $f_{\vec{w}}(\vec{x}):\calX \to \bbR^{K}$ be any predictor (not
        necessarily a neural network) with parameters $\vec{w}$, and let $P$ be
        any distribution on the parameters that is independent of the
        training data. Then, for any $\delta >0$, with probability $\geq 1-\delta$ over the training set of size $N$, for any $\vec{w}$,
        and any random perturbation $\vec{u}$, we have
\begin{align}
    \mathbb{E}_{\vec{u}} [L_0 (f_{\vec{w}+\vec{u}})]
    &\leq \mathbb{E}_{\vec{u}} [\hL_0 (f_{\vec{w}+\vec{u}})] + 2\sqrt{\frac{2\left(D_{\mathrm{KL}}\left(\vec{w}+\vec{u}\|P\right)+\ln\left(\frac{2N}{\delta}\right)\right)}{N-1}}\,.
\end{align}
\end{lemma}

To derive a bound on the risk $L_0(f_{\vec{w}})$ of the specific, non-randomized predictor $f_{\vec{w}}$, we can relate $L_0(f_{\vec{w}})$ to the empirical margin loss $\hL_{\gamma}(f_{\vec{w}})$ via the expected risk of the randomized predictor $f_{\vec{w}+\vec{u}}$. The following lemma provides such a connection, contingent on the perturbation $\vec{u}$ being sufficiently small in terms of its effect on the output.
\begin{lemma}[PAC-Bayes margin bound~\cite{perturbation_bound}]\label{lemma:pac-bayes_margin}
        Let $f_{\vec{w}}(\vec{x}):\calX \to \bbR^{K}$ be any predictor (not
  necessarily a neural network) with parameters $\vec{w}$, and $P$ be
  any distribution on the parameters that is independent of the
  training data. Then, for any $\gamma, \delta >0$, with probability
  $\geq 1-\delta$ over the training set of size $N$, for any $\vec{w}$,
  and any random perturbation $\vec{u}$ s.t. $\bbP_{\vec{u}}
  \left[\max_{\vec{x} \in
      \calX} \lVert{f_{\vec{w}+\vec{u}}(\vec{x})-f_{\vec{w}}(\vec{x})} \rVert_\infty <
      \frac{\gamma}{4} \right]\geq \frac{1}{2}$, we have
\begin{equation}
L_0(f_{\vec{w}}) \leq \hL_{\gamma}(f_{\vec{w}})+ 4\sqrt{\frac{D_{\mathrm{KL}} \left(\vec{w}+\vec{u}\|P\right)+\ln\left(\frac{6N}{\delta}\right)}{N-1}}\,.
\end{equation}
\end{lemma}
The lemma, variations of which exist in the literature (see, e.g., Refs.~\cite{LS2003, McAllester2003, behboodi2022pac, perturbation_bound}), provides a powerful tool that applies broadly, independently of whether the predictor is a neural network or some other parameterized model. It will form the foundation for deriving specific generalization guarantees for quantum models using PAC-Bayes theory.

\section{Results}

In this section, we present our main theoretical and empirical findings. We first establish general PAC-Bayesian generalization bounds for quantum models parameterized as quantum channels, providing a rigorous guarantee that links the true risk to the empirical loss and the spectral properties of the learned parameters. We then extend this framework to the setting of geometric QML, deriving a specialized bound for equivariant models that explicitly quantifies the benefits of symmetry through group-theoretic properties. Finally, we corroborate our theoretical insights through a set of numerical experiments on a standard quantum learning task.

\subsection{PAC-Bayes generalization bound for quantum models}\label{s:pac-bayes}

Here, we derive generalization bounds for quantum models using the PAC-Bayesian framework. Our approach is inspired by the perturbation-based method from Ref.~\cite{perturbation_bound}, originally developed for classical feedforward networks with ReLU activations using weight norms. Next, we use the PAC-Bayes framework to analytically derive margin-based bounds in terms of weight norms of the quantum model. As established in Lemma~\ref{lemma:pac-bayes_margin}, the PAC-Bayes framework requires bounding the change in the model's output under perturbations to its parameters. To achieve this, we quantify how the output of the quantum model changes under small perturbations, see \Cref{ss:perturbation} and Appendix~\ref{a:perturbation-bounds}. By 
leveraging our perturbation bounds alongside Lemma~\ref{lemma:pac-bayes_margin}, we establish PAC-Bayes generalization bounds for quantum models. Below, we present \emph{informal} versions of our main theorems using simplified notation. To highlight the core complexity terms that drive generalization---which are the primary focus of our analysis---we abuse the $\widetilde{\mathcal{O}}$ notation to suppress a further (typically non-dominant) term. We refer to \Cref{a:generalization-bounds} for the formal, mathematically precise statements of all theorems.

\begin{theorem}[PAC-Bayes generalization bounds for quantum models -- PM framework (Informal)]\label{th:pac-bayes-PM}
      For any $L > 0$, let $f_{\vec{w}}^{\mathrm{PM}}: \mathcal{X} \to \mathbb{R}^K$ be an $L$-layer quantum model represented via PM framework $\phi_j = \phi_{\mathrm{PM}}^{W_j}$ with $\xi_j\coloneqq \norm{W_j}_0$, $\xi_{\max}=\max_j \xi_j$, and $\beta_{\rm PM}=\sum_{j=1}^L\left(\prod_{\ell=j+1}^L\norm{W_\ell}_{1,1}\right)$. Then, for any $\gamma > 0$, with high probability over the training set of size $N$, for any  vector of learned parameters $\vec{w}=\mathrm{vec}(\{W_j\}_{j=1}^L)$, we have
    \begin{equation}
        L_0(f_{\vec{w}}^{\mathrm{PM}}) \leq \hat{L}_{\gamma}(f_{\vec{w}}^{\mathrm{PM}})+ \widetilde{\mathcal{O}}\left(\sqrt{\frac{\beta_{\rm PM}^2\:\xi_{\max} \ln\left(\sum_{j=1}^L2^{\xi_j}\right)\:\sum_{j=1}^L\norm{W_j}_F^2}{\gamma^2 N}} \right)\,.
    \end{equation}
\end{theorem}

\begin{theorem}[PAC-Bayes generalization bound for quantum models -- PTM framework (Informal)]\label{th:pac-bayes-PTM}
      For any $L > 0$, let $f_{\vec{w}}^{\mathrm{PTM}}: \mathcal{X} \to \mathbb{R}^K$ be an $L$-layer quantum model represented via PTM framework $\phi_j = \phi_{\mathrm{PTM}}^{W_j}$ with $\xi_j\coloneqq \norm{W_j}_0$, $\xi_{\max}=\max_j \xi_j$, and $\beta_{\rm PTM} = \sqrt{d_\mathrm{out}^{(L)}}\sum_{j=1}^L\frac{1}{\sqrt{d_\mathrm{in}^{(j)}}}\prod_{\ell=j+1}^L \lVert W_\ell\rVert_{1,1}$. Then, for any $\gamma > 0$, with high probability over the training set of size $N$, for any  vector of learned parameters $\vec{w}=\mathrm{vec}(\{W_j\}_{j=1}^L)$, we have 
    \begin{equation}\label{eq:formal_PM}
        L_0(f_{\vec{w}}^{\mathrm{PTM}}) \leq \hat{L}_{\gamma}(f_{\vec{w}}^{\mathrm{PTM}})+ \widetilde{\mathcal{O}}\left(\sqrt{\frac{\beta_{\rm PTM}^2\:\xi_{\max} \ln\left(\sum_{j=1}^L2^{\xi_j}\right)\:\sum_{j=1}^L\norm{W_j}_F^2}{\gamma^2 N}} \right)\,.
    \end{equation}
\end{theorem}
\begin{proof}[Proof sketch]
    The generalization bounds in \Cref{th:pac-bayes-PM,th:pac-bayes-PTM} are derived via the PAC-Bayesian framework in Lemma~\ref{lemma:pac-bayes_margin}, using a Gaussian prior $P=\mathcal{N}(0,\sigma^2\mathbb{I})$ and a posterior $Q = \mathcal{N}(\vec{w},\sigma^2\mathbb{I})$ that introduces a random perturbation around the trained parameters $\vec{w}$. 
    To apply \Cref{lemma:pac-bayes_margin}, we need to show that the perturbed model output deviates from the original by at most $\gamma/4$ with high probability. \Cref{lemma:pert_bound_PM,lemma:pert_bound_PTM} in \Cref{sec:methods} provide sensitivity bounds for such perturbations, expressed in terms of $\norm{W_j}_{1,1}$ and the norms of the perturbation matrices. As the perturbations are Gaussian, their norms can, with high probability, be controlled using standard tail bounds for random matrices. This allows us to determine a suitable variance $\sigma^2$ to ensure that the output deviation remains within the required margin.

     With the perturbation scale $\sigma$ determined, we compute the KL divergence between the posterior and prior. A technical challenge arises from dependence of $\sigma$ on the learned parameters, which means that $P$ is not a data-independent prior as required in \Cref{lemma:pac-bayes_margin}. To resolve this, we use a covering net argument over discretized values of $\norm{W_j}_{1,1}$; controlling the covering net size is made feasible by structural constraints on $W_j$ that arise from \emph{complete positivity} and \emph{trace preservation} of the associated quantum channels. A union bound over this net ensures that the PAC-Bayes bound holds for any set of learned parameters while using data-independent priors. Substituting the KL bound into  Lemma~\ref{lemma:pac-bayes_margin}  and accounting for the union bound from the covering net completes the proof. Full proof details are provided in Appendix~\ref{a:generalization-bounds}.
\end{proof}

\Cref{th:pac-bayes-PM,th:pac-bayes-PTM} established PAC-Bayes generalization bounds for general quantum models in two different representations, linking the true risk to the empirical margin loss and a non-uniform complexity measure of the model after training.
To better interpret the operational meaning of these bounds, the terms can be understood as follows:
\begin{itemize}
    \item The $\beta_{\text{P(T)M}}$ term effectively measures how much a local parameter perturbation in a given layer is amplified as it propagates through the subsequent layers.
    \item The Frobenius norm term $\sum_{j=1}^L\norm{W_j}_F^2$ acts a geometric distance evaluating how far the learned model is from the maximally depolarizing channel. Operationally, it measures how much the optimization algorithm had to pull the model away from a completely uninformative, constant-valued physical baseline in order to fit the training data.
    \item The sparsity term $\xi_{\max} \ln\left(\sum_{j=1}^L2^{\xi_j}\right)$ represents the number of active parameters across layers, controlling the effective volume of the hypothesis space. It ensures that the random perturbations remain bounded within the required margin.
\end{itemize}
Crucially, this highlights the difference between our results and uniform bounds: these terms do not merely depend on general static features of the model (e.g., the total number of trainable gates or encoding operations), but actually on where the model ends up \emph{after} training. This data-dependent complexity is captured explicitly through the learned parameter norms, linking the generalization gap directly to the specific solution found by the optimization algorithm. 
\Cref{th:pac-bayes-PM,th:pac-bayes-PTM} are phrased for quantum models in which an initial data encoding is followed by trainable layers.
However, these results immediately extend to data re-uploading quantum models~\cite{perez2020data}, as we show next.

\begin{corollary}[PAC-Bayes generalization bounds for data re-uploading quantum models]\label{cor:reuploading} 
The PAC-Bayes generalization bounds established in \Cref{th:pac-bayes-PM,th:pac-bayes-PTM} also apply to data re-uploading quantum models of the form
\begin{equation}
   \rho_{\mathrm{\mathrm{out}}}(\vec{x}) = \left(\phi_L \circ \mathcal{E}_L^{\vec{x}} \circ \phi_{L-1} \circ \mathcal{E}_{L-1}^{\vec{x}} \circ \cdots \circ \phi_1 \circ \mathcal{E}_1^{\vec{x}}\right)(\rho_0)\,,
\end{equation}
where each $\phi_j=\phi_{\mathrm{P(T)M}}^{W_j}$ is a parameter-dependent quantum channel, and each $\mathcal{E}_j^{\vec{x}}$ is an input-dependent quantum channel implementing the data re-uploading step.
\end{corollary}

Full proof details can be found in Appendix~\ref{a:data-re-uploading}. Allowing the data re-uploading steps as well as the trainable layers to be arbitrary (potentially non-unitary) quantum channels, the class of quantum models covered by \Cref{cor:reuploading} now captures all models commonly studied in variational quantum machine learning.

To illustrate how our abstract mathematical framework applies to concrete QML architectures, we analyze two representative examples: \emph{quantum convolutional neural networks} (QCNNs)~\cite{cong2019quantum} and \emph{dynamic PQCs}~\cite{deshpande2024dynamicparameterizedquantumcircuits}. We specifically focus on these models because they naturally incorporate the non-unitary operations (such as mid-circuit measurements, feedforward, and partial traces) that our bounds are uniquely designed to handle. Furthermore, these architectures will serve as the foundation for our numerical experiments in~\Cref{s:numerics}, providing a direct link between our theoretical guarantees and empirically observed performance.

\begin{example}[Application to dynamic PQCs.]
Dynamic PQCs represent a promising class of models that utilize mid-circuit measurements and feedforward to enhance expressivity and mitigate barren plateaus \cite{deshpande2024dynamicparameterizedquantumcircuits}. We consider a dynamic PQC architecture (which we later evaluate numerically) with a single effective layer, used here for the sake of illustration and to match the setting of our numerical experiments. This effective layer consists of two consecutive dynamic operations per qubit, followed by a global unitary. Focusing on a single layer is sufficient to capture the key mechanism of interest, namely how local dissipative operations can drive each qubit channel toward a depolarizing regime, and the same reasoning extends straightforwardly to deeper circuits composed of multiple such layers. Since the dynamic operations act locally, the overall channel factorizes across qubits, implying that the Frobenius norm factorizes as $\|\chi\|_F^2=\prod_{i=1}^n \|\chi_i\|_F^2$. Moreover, $\|W\|_F^2=\|\chi\|_F^2-1/4^n$, where $\|\chi\|_F^2$ equals the purity of the corresponding normalized Choi state. This structure allows for an efficient evaluation of the norm, as it only requires computing single-qubit process matrices, with a cost linear in $  n  $. As detailed in \Cref{a:pac-bayes_PM}, with proper parameter choices, each effective single-qubit channel could implement a maximally depolarizing channel, minimizing the local purity to $\tr(\chi_i^2)=1/4$. In this regime, $\|\chi\|_F^2=(1/4)^n$, yielding $\|W\|_F^2=0$ and thus providing an explicit physically implementable example of a vanishing main term in our bound. Moreover, this configuration provides a natural reference point, since small deviations from the maximally depolarizing channel induce controlled perturbations of the main term.
\end{example}

\begin{example}[Application to QCNNs.]
    QCNNs are a widely studied architecture for quantum learning tasks, naturally featuring non-unitary pooling operations \cite{cong2019quantum}. We consider a QCNN in which the number of qubits is halved at each layer, $n_{\ell+1}=n_\ell/2$, up to a maximum depth $L=\log_2 n$ (taking $n$ a power of two). To derive concrete scaling behaviors for our bounds (see~\Cref{a:pac_bayes-PTM}), we introduce a modeling assumption for the parameter sparsity: we assume a sparsity model $\xi_\ell=\xi_L\,2^{\alpha(n_\ell-n_L)}$ which starts from a fixed sparsity $\xi_L$ at the final layer and increases exponentially as qubits are added toward earlier layers ($l\xrightarrow{}1$), where $\alpha>0$ is a constant. For the standard QCNN block---consisting of a global unitary acting on $n_\ell$ qubits followed by a partial trace that leaves $n_{\ell+1}$ qubits---the Frobenius norm satisfies $\|W_\ell\|_F^2=d_{\rm in}^{(\ell)}\,d_{\rm in}^{(\ell)}-d_{\rm in}^{(\ell)}/d_{\rm out}^{(\ell)}=2^{n_\ell/2}(2^{n_\ell}-1)$.
    Importantly, this quantity is independent of the learned parameters, although $\beta_{\text{PTM}}$ is not. By contrast, as in the dynamic PQCs discussed earlier, introducing dissipative blocks at each layer renders $\|W_\ell\|_F^2$ fully parameter-dependent throughout the circuit. In this hybrid architecture (which we also employ in our upcoming experiments), the Frobenius norm can in principle be driven arbitrarily close to zero, thereby enabling the small-norm regime and yielding tighter non-uniform generalization guarantees.
\end{example}

To contextualize our results, we can also compare our non-uniform PAC-Bayesian guarantees to the uniform generalization bounds typically found in the literature, such as those established in Refs.~\cite{Du2022,caro2022generalization}. The defining strength of our bounds is that they depend explicitly on the specific post-trained parameters of the model, rather than just the model's overall capacity. However, to facilitate a direct theoretical comparison, we can consider a worst-case scenario over these post-training parameters to derive a uniform version of our worst-case bounds. As detailed in~\Cref{a:generalization-bounds}, for models evaluated under the PM formalism, there is a distinctly favorable regime where this worst-case PAC-Bayesian bound analytically outperforms the standard uniform bound. While we were unable to identify a similarly favorable regime for the PTM formalism under the considered sparsity and dimensionality assumptions, it should be noted that this analysis strictly evaluates a conservative, uniform relaxation of our non-uniform bounds. The main advantage of our PAC-Bayesian approach emerges when the explicit parameter dependence is preserved, allowing us to circumvent the pessimistic bottlenecks inherent to parameter-independent uniform bounds.

Having established this general framework, we now proceed to specialize our analysis and derive a refined bound tailored to equivariant quantum models.

\subsection{Extension to equivariant quantum models}\label{s:equivariant}

While the general bounds established in~\Cref{th:pac-bayes-PM,th:pac-bayes-PTM} provide a robust and universally applicable foundation for any quantum channel, our framework's flexibility further allows for the direct incorporation of architectural inductive biases. A prominent example is equivariance, where quantum models are built to respect inherent symmetries in the data or task, often leading to improved parameter efficiency and generalization~\cite{ragone2022representation,theory_equivariant,Schatzki,tuysuz2024symmetry,PRXQuantum.4.010328}. To establish a more nuanced understanding of their generalization properties and potentially derive a tighter bound, we specialize our PAC-Bayesian approach to leverage the specific structure of these equivariant channels. Analogous to the general case, we analytically derive a margin-based bound, but here we express it in terms of the specific properties arising from the equivariant parameterization in the irrep basis. As established in \cref{lemma:pac-bayes_margin}, this requires bounding the change in the model's output under perturbations to its parameters. By quantifying how the output changes under small perturbations to its block-diagonal equivariant representation---utilizing specific properties of the irreducible representations (see \Cref{ss:perturbation} and Appendix~\ref{a:perturbation-bounds})---we establish a PAC-Bayes generalization bound tailored to the symmetry of the model. 

\begin{theorem}[PAC-Bayes generalization bound for equivariant quantum models (Informal)]\label{th:pac-bayes-equivariant}
For any depth $L > 0$, let $f_{\vec{w}}^{\mathrm{eq}}: \mathcal{X} \to \mathbb{R}^K$ be an $L$-layer equivariant quantum model parameterized by $w = \{W_{j,\lambda}\}_{j,\lambda}$. Let $\Xi_j = \sum_\lambda m_{j,\lambda}^2$ be the number of equivariant parameters per layer, where $\Xi_{\max} = \max_j \Xi_j$, $\beta_{\mathrm{eq}} = \sum_{j=1}^L \left( \prod_{\ell=j+1}^L \|W_l\|_{\mathrm{eq},1} \right)$ and let $d_{\max} = \max_\lambda d_\lambda$ be the maximum dimension of the irreducible representations.
Then, for any margin $\gamma > 0$, with high probability over the training set of size $N$, for any  vector of learned parameters $\vec{w}=\mathrm{vec}(\{W_j\}_{j=1}^L)$, we have
\begin{equation}
    L_0(f_{\vec{w}}^{\mathrm{eq}}) \leq \hat{L}_{\gamma}(f_{\vec{w}}^{\mathrm{eq}}) + \widetilde{\mathcal{O}}\left(\sqrt{\frac{d_{\max}^2 \, \Xi_{\max} \, \ln(\sum^L_{j=1} 2^{\Xi_j}) \, \beta_{\mathrm{eq}}^2 \sum_{j=1}^L \|W_j\|_{\mathrm{eq},F}^2}{\gamma^2 N}} \right)\,,
\end{equation}
where the layer-wise equivariant norms 
are defined in terms of the irrep blocks $W_{j,\lambda}$ as
\begin{equation}
    \|W_j\|_{\mathrm{eq},1} = \sum_{\lambda} d_\lambda \|W_{j,\lambda}\|_1 \quad \text{and} \quad \|W_j\|_{\mathrm{eq},F}^2 = \sum_{\lambda} d_\lambda \|W_{j,\lambda}\|_F^2 \,.
\end{equation}
\end{theorem}

Full proof details are provided in Appendix~\ref{a:generalization-bounds}. \Cref{th:pac-bayes-equivariant} establishes a PAC-Bayes generalization bound specifically for equivariant quantum models. This bound mirrors the form of the general case but replaces the general complexity measures with symmetry-adapted quantities: multiplicities $m_\lambda$, irrep dimensions $d_\lambda$, and the equivariant norms. This highlights that for highly symmetric problems, the generalization guarantee can be significantly tighter than what general bounds would suggest.

\subsection{Numerical experiments}\label{s:numerics}
To validate the theoretical framework established in the previous sections, we perform proof-of-concept numerical experiments designed to analyze the correlation between our analytically derived generalization bounds and the true performance of quantum models. Specifically, we investigate how the norms of the learned parameter matrices relate to the actual generalization gap in a classification task.

We focus on the classification of quantum phases of matter, which is a relevant task for the study of condensed-matter physics~\cite{carrasquilla2017machine, Sachdev}, and due to its significance, it frequently appears as a benchmark problem in the literature~\cite{carrasquilla2017machine, caro2022generalization, gil2024understanding, recio2024learning}. The task is to classify the ground states of a generalized cluster Hamiltonian for a $1$D spin chain into one of four phases of matter.
To define the classification task, we utilize a restricted training set of $N=8$ samples. While this size is sufficient to identify the phase diagram~\cite{caro2022generalization, gil2024understanding}, it intentionally creates a regime where models can exhibit a wide range of behaviors; from overfitting the training data to achieving robust generalization. The quantum models are evaluated on a separate, independent test set 
of $1000$ samples.

We employ both dynamic PQC and QCNN architectures, which naturally incorporate classical feedforward and dissipation dynamics analyzed in our framework.
Moreover, we design the experiments to also reveal how norms play a role in generalization. First, we consider a single-layer dynamic PQC, which allows us to fix $\beta_{\text{PM}}=1$ while varying the Frobenius norm term, i.e., the purity of the single-qubit channels. Second, in our QCNN architecture, we also incorporate a few dynamic operations similar to those found in the dynamic PQC. This design choice provides a wider range of complexity term values by allowing the modification of both $\beta_{\text{PTM}}$ and the Frobenius norms.
Further details regarding the implementation are provided in the numerical methods Section~\ref{ss:num_methods} and the publicly available code~\cite{github_repository}.

In Figure~\ref{fig:numerics_results}, we illustrate the results obtained from $1400$ independent training runs for the (a) dynamic PQC and (b) QCNN architectures. The empirical data confirms a positive trend across both architectures, linking the norms of the learned parameter matrices to the generalization gap. For the dynamic PQC evaluated under the PM formalism, we observe a positive linear correlation of $r=0.26$ between the theoretical complexity term and the actual generalization error. This relationship is more pronounced for the QCNN evaluated under the PTM formalism, which exhibits a stronger correlation of $r=0.46$. Crucially, a likely reason for the weaker correlation in Fig.~\ref{fig:numerics_results}(a) compared to Fig.~\ref{fig:numerics_results}(b) is a consequence of our experimental design: As the single-layer dynamic PQC fixes $\beta_{\text{PM}}$ to a constant, its theoretical complexity relies exclusively on the Frobenius norm. In contrast, the multi-layered QCNN allows for both the $\beta_{\text{PTM}}$ and Frobenius norm terms to actively vary. By allowing both contributions to influence the overall complexity measure, the bound seems to more comprehensively capture the model's generalization behavior, thereby yielding the stronger correlation observed.
Ultimately, the distributions confirm that models converging to solutions with smaller complexity terms, and therefore lower parameter norms, tend 
to exhibit smaller generalization gaps.

\begin{figure*}[t]
        \centering
        \includegraphics[width=.9\textwidth]{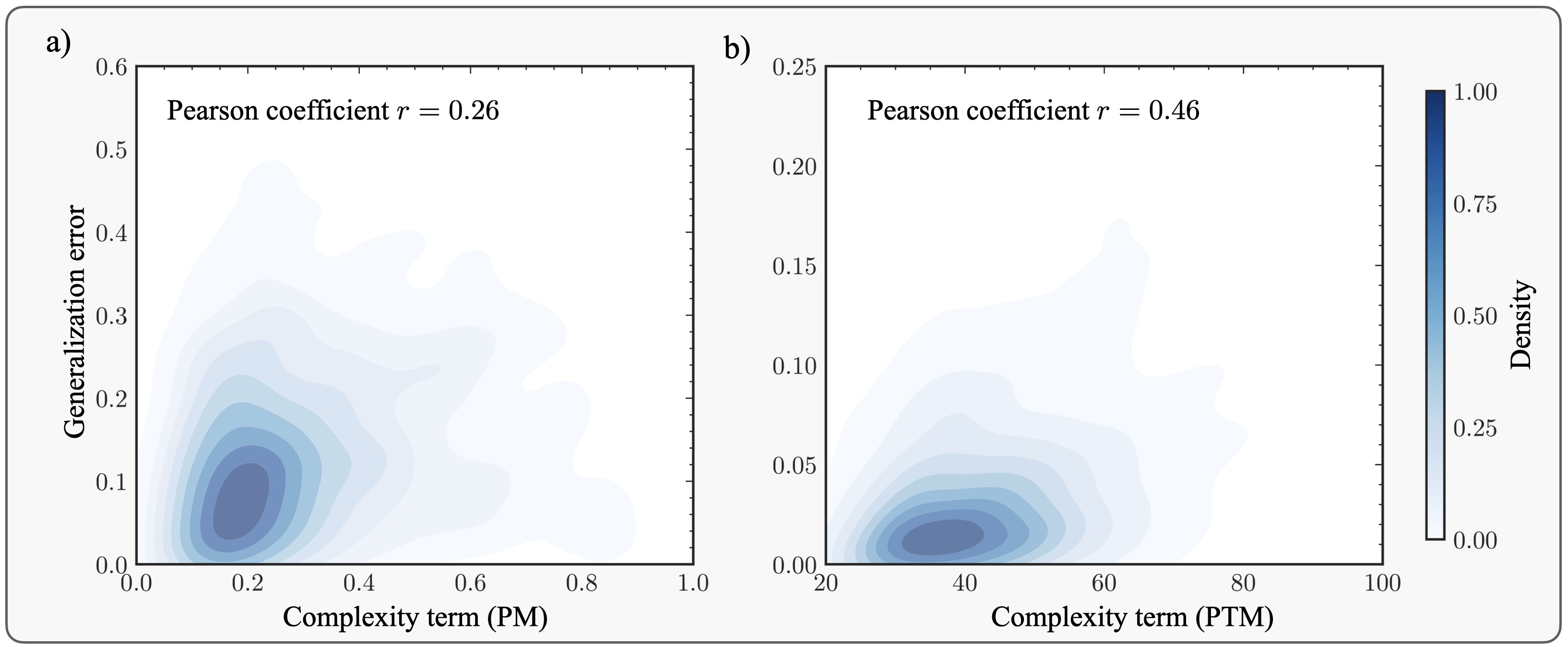}
        \caption{
            \textbf{Correlation between complexity term and generalization gap.} Generalization gap as a function of the complexity term $\beta_\text{P(T)M} ({ \sum_{j=1}^L \norm{W_j}_F^2 })^{1/2}$ derived in~\Cref{th:pac-bayes-PM,th:pac-bayes-PTM} for the \textbf{(a)} dynamic PQC and \textbf{(b)} QCNN. The density plots represent the outcomes of $1400$ individual and independent trained models, yielding a positive Pearson correlation coefficients of $r=0.26$ and $r=0.46$ for the dynamic PQC and QCNN architectures, respectively. Overall, we observe that models with smaller complexity term, and therefore lower norms, tend to exhibit smaller generalization gaps.
        } 
        \label{fig:numerics_results}
    \end{figure*}

\section{Discussion}\label{s:discussion}

Our work marks an advancement for understanding generalization in \emph{quantum machine learning}  (QML), particularly as QML models 
begin to 
exhibit phenomena like overparameterization~\cite{larocca2023theory, kempkes2025doubledescentquantummachine} or the 
capacity to fit 
random data~\cite{gil2024understanding}. 
Recent perspectives 
in classical \emph{machine learning} 
(ML) suggest that these behaviors are neither mysterious nor unique to 
neural networks. Instead, generalization can be rigorously characterized by analyzing soft inductive biases (algorithmic preferences for simpler solutions within flexible hypothesis spaces) using frameworks like PAC-Bayes~\cite{wilson2025deep}. Consistent with this view, the generalization bounds we derive provide new perspectives in the study of QML. Unlike standard
\emph{uniform} generalization bounds (based on quantities such as VC dimension or covering numbers), which describe worst-case behavior of the entire model class~\cite{zhang2017_rethinking, gil2024understanding, wilson2025deep}, our results account for the data distribution and the learned parameters. Consequently, they provide a more nuanced theoretical lens than standard uniform convergence approaches.

Central to our approach is the focus on QML models structured as layered quantum channels, represented via the \emph{Pauli transfer 
matrix} or the \emph{process matrix}. While this formulation deviates from standard models based purely on unitary layers, it aligns with a growing body of research utilizing non-unitary operations to implement expressive channel-based models. This direction is supported by emerging architectures such as dynamic quantum circuits~\cite{deshpande2024dynamicparameterizedquantumcircuits}, linear combination of unitaries~\cite{heredge2025nonunitary}, or instantaneously deep quantum neural networks~\cite{huang2025generative}. While recent research highlights that feedforward operations, mid-circuit measurements, and dissipation improve trainability by mitigating severe drawbacks like barren plateaus~\cite{mcclean2018barren, deshpande2024dynamicparameterizedquantumcircuits, heredge2025nonunitary, huang2025generative, zapusek2025scaling}, we establish that these same non-unitary mechanisms also yield favorable generalization properties.
Moreover, our framework demonstrates its flexibility by providing generalization guarantees specifically tailored to equivariant quantum models. By parameterizing channels in the basis of irreducible representations, we derive bounds that depend explicitly on the group structure, specifically its dimensions and multiplicities, rather than the full Hilbert space dimension. This yields a rigorous theoretical explanation for the benefits of geometric QML: equivariance acts as a \enquote{hard} inductive bias that restricts the exploration of the hypothesis space to symmetry-respecting channels. Mathematically, this is captured by replacing unstructured operator norms with symmetry-adapted norms that are significantly smaller in the presence of non-trivial symmetries. Thus, our results add a new way to formally quantify how physical symmetries can be leveraged in QML.

While the PAC-Bayesian framework is versatile, adapting it to quantum models requires a bespoke parameterization. As demonstrated in Appendix~\ref{a:kraus}, computing the PAC-Bayes complexity terms using standard mathematical representations of quantum channels, such as the Kraus or Stinespring parameterizations, leads to bottlenecks. The strict constraints of trace preservation and complete positivity force the relevant parameter norms to be blind to the actual function the model learned. Consequently, this direct application yields uniform bounds. This highlights why our specific representation is crucial. By framing the model's weights as deviations from the maximally depolarizing channel, we establish a natural, operationally meaningful baseline. This approach is what allows the PAC-Bayesian framework to successfully yield meaningful, non-uniform generalization guarantees. These results can be interpreted on two distinct levels: operationally at the channel level, and geometrically at the level of the loss landscape.

Operationally, our results convey a clear message: the closer the output channel is to the maximally depolarizing channel, the better generalization is to be expected. This acts as a specific realization of a \enquote{soft} inductive bias in the quantum context. While the depolarizing channel is usually linked to uncorrected hardware noise, this relationship is intuitively clear in a learning context. The maximally depolarizing channel acts a constant function: independently of the input, the output is always the maximally-mixed state. Since generalization captures the difference between training and test performance, a constant-valued channel achieves zero generalization gap, as it is unable to fit any data regardless of the source. This creates a fundamental tension: the closer the channel is to maximally depolarizing, the smaller the generalization gap, yet the channel cannot achieve good training performance unless it adapts to the data. Consequently, a fundamental trade-off emerges between expressivity (fitting the data) and generalization (approximating the maximally depolarizing channel). In this balancing act, our framework dictates that given two channels with equivalent training performance, the more dissipative channel is theoretically guaranteed to generalize better. This provides actionable guidance for QML model design: engineered dissipation---via e.g., mid-circuit measurements and feedforward---and explicit regularization terms should be actively leveraged to maximize learning performance.

Geometrically, the PAC-Bayes framework offers a useful picture regarding the loss landscape that mirrors insights from modern deep learning theory. Following standard de-randomization techniques~\cite{biggs2022margins}, the curvature of the loss landscape is captured by the complexity term in \Cref{th:pac-bayes-PM,th:pac-bayes-PTM,th:pac-bayes-equivariant}. Intuitively, this complexity term represents the largest curvature in any direction in the parameter space. This directly aligns with well-established observations in classical ML, where local minima with lower curvature (flat minima) correspond to low effective dimensionality and represent strong predictors of generalization~\cite{wilson2025deep}. Therefore, while this geometric picture does not strictly dictate how to optimize QML models, it provides a method to certify the quality of the output of a training algorithm. 

Looking ahead, our research opens several promising avenues. First, the versatility of the PAC-Bayesian framework encourages the exploration of alternative, perhaps more refined, representations of quantum channels that better align with empirical performance. Second, our results highlight that training strategies promoting robustness to parameter perturbations can lead to better generalization. This intriguingly suggests the potential of treating aspects of channel noise in near-term quantum devices or specific channel parameters as trainable hyperparameters to optimize generalization.
On a more general note, classical machine learning has benefited immensely from the development of a diverse array of PAC-Bayesian bounds, which are now recognized among the most powerful theoretical tools for understanding generalization. 
We contend that a similarly rich trajectory of developing varied non-uniform bounds is vital for advancing QML. Our work, providing the first PAC-Bayesian generalization bounds of its kind for QML models, serves as a foundational contribution towards this goal.

\section{Methods}\label{sec:methods}

\subsection{Perturbation bounds}\label{ss:perturbation}

At the center of our proofs are perturbation bounds that quantify how much the output of a quantum model changes under perturbations to its parameters. Analogous bounds for classical neural networks were crucial in establishing competitive PAC-Bayes bounds \cite{perturbation_bound, behboodi2022pac}.
The following lemmas provide such bounds, relating the change in the output to matrix norms of the PTM, PM, or equivariant parameter matrices and the perturbations at each layer.

\begin{lemma}[Perturbation bound for quantum models -- PM framework]\label{lemma:pert_bound_PM}
        For any $L>0$, let $f_{\vec{w}}^{\mathrm{PM}}: \calX \to \bbR^K$ be an $L$-layer quantum model parameterized by $\vec{w}=\mathrm{vec}(\{W_j\}_{j=1}^L)$. For any $\vec{x} \in \calX$, and any perturbation $\vec{u}=\mathrm{vec}\left(\{U_j\}_{j=1}^L\right)$ such that $\norm{U_j}_{1,1} \leq \frac{1}{L} \norm{W_j}_{1,1}$ holds for all $j$, the change in the output of the quantum model is bounded as
\begin{equation}
    \begin{split}
        \norm{f_{\vec{u}+\vec{w}}^{\mathrm{PM}}(\vec{x})-f_{\vec{w}}^{\mathrm{PM}}(\vec{x})}_\infty
        \leq e \sum_{j=1}^L\left[\norm{U_j}_{1,1}\left(\prod_{\ell=j+1}^L\norm{W_\ell}_{1,1}\right)\right]\, ,
    \end{split}
\end{equation}
where the norm $\norm{A}_{1,1}$ equals the sum of absolute values of all the entries of the matrix $A$.
\end{lemma}

\begin{lemma}[Perturbation bound for quantum models -- PTM framework]\label{lemma:pert_bound_PTM}
        For any $L>0$, let $f_{\vec{w}}^{\mathrm{PTM}}: \calX \to \bbR^K$ be an $L$-layer quantum model parameterized by $\vec{w}=\mathrm{vec}(\{W_j\}_{j=1}^L)$. For any $\vec{x} \in \calX$, and any perturbation $\vec{u}=\mathrm{vec}\left(\{U_j\}_{j=1}^L\right)$ such that $\norm{U_j}_{1,1} \leq \frac{1}{L} \norm{W_j}_{1,1}$ holds for all $j$, the change in the output of the quantum model is bounded as
\begin{equation}
    \begin{split}
        \norm{f_{\vec{u}+\vec{w}}^{\mathrm{PTM}}(\vec{x})-f_{\vec{w}}^{\mathrm{PTM}}(\vec{x})}_\infty
        \leq e\sqrt{d_{\mathrm{\mathrm{out}}}^{(L)}}\sum_{j=1}^L\left[\frac{1}{\sqrt{d_{\mathrm{in}}^{(j)}}}\norm{U_j}_{1,1}\left(\prod_{\ell=j+1}^L\norm{W_\ell}_{1,1}\right)\right]\, ,
    \end{split}
\end{equation}
where the norm $\norm{A}_{1,1}$ equals the sum of absolute values of all the entries of the matrix $A$.
\end{lemma}

\begin{lemma}[Perturbation bound for equivariant quantum models]\label{lemma:pert_bound_equivariant}
For any $L > 0$, let $f_{\vec{w}}^{\mathrm{eq}}: \calX \to \bbR^K$ be an $L$-layer equivariant quantum model parameterized by $\vec{w}$ representing the set of parameter matrices $\{W_{j,\lambda}\}_{j,\lambda}$. For any $\vec{x} \in \mathcal{X}$ and any perturbation $\vec{u}$ representing $\{U_{j,\lambda}\}_{j,\lambda}$ such that $\|U_{j,\lambda}\|_{1} \le \frac{1}{L} \|W_{j,\lambda}\|_{1}$ holds for all $j$ and $\lambda$, the change in the output of the quantum model is bounded as
\begin{equation}
    \norm{f_{\vec{w}+\vec{u}}^{\mathrm{eq}}(\vec{x}) - f_{\vec{w}}^{\mathrm{eq}}(\vec{x})}_\infty \le e \sum_{j=1}^{L} \left[ \norm{U_j}_{\mathrm{eq},1} \left( \prod_{\ell=j+1}^{L} \norm{W_\ell}_{\mathrm{eq},1} \right) \right]\,,
\end{equation}
where the equivariant 1-norm for a layer is defined as $\norm{V}_{\mathrm{eq},1} := \sum_{\lambda} d_\lambda \norm{V_{\lambda}}_1$, with $d_\lambda$ being the dimension of the irreducible representation $\lambda$ and $\norm{\cdot}_1$ being the trace norm.
\end{lemma}

These lemmas establish upper bounds on the sensitivity of the quantum model output to perturbations in its parameters.

\begin{proof}[Proof sketch]
    We sketch the proof for the PM and PTM parameterizations (\Cref{lemma:pert_bound_PM,lemma:pert_bound_PTM}); the proof for the equivariant case (\Cref{lemma:pert_bound_equivariant}) follows a similar recursive structure but requires specialized norms. Detailed proofs are provided in \Cref{a:perturbation-bounds}.
    Starting from the PM and PTM parameterization, it is easy to show that 
    \begin{equation}\label{eq:sketch-parameter-change}
        \norm{(\phi_{\mathrm{PTM}}^{W_{j+1}+U_{j+1}}-\phi_{\mathrm{PTM}}^{W_{j+1}})}_{1\to 1}\leq \sqrt{\frac{d_{\mathrm{\mathrm{out}}}^{(j+1)}}{d_{\mathrm{in}}^{(j+1)}}}\,\|U_{j+1}\|_{1,1}
    \end{equation}
    and similarly 
    \begin{equation}
        \norm{(\phi_{\mathrm{PM}}^{W_{j+1}+U_{j+1}}-\phi_{\mathrm{PM}}^{W_{j+1}})}_{1\to 1}\leq \|U_{j+1}\|_{1,1} \, .
    \end{equation}
    We can also use 
    the PM and PTM parameterizations to control the $1$-norm contraction coefficients of the associated channels as
    \begin{equation}
        \mathrm{CC}_{1,1}\left(\phi_{\mathrm{PTM}}^{W_{j+1}+U_{j+1}}\right)\leq \sqrt{\frac{d_{\mathrm{\mathrm{out}}}^{(j+1)}}{d_{\mathrm{in}}^{(j+1)}}}\,(\|W_{j+1}\|_{1,1}+\|U_{j+1}\|_{1,1})
    \end{equation}
    and
    \begin{equation}
        \mathrm{CC}_{1,1}\left(\phi_{\mathrm{PM}}^{W_{j+1}+U_{j+1}}\right)\leq \norm{W_{j+1}}_{1,1}+\norm{U_{j+1}}_{1,1} \, .
    \end{equation}
    By definition of the contraction coefficient, we now have 
    \begin{equation}\label{eq:sketch-input-change}
        \norm{\phi_{\mathrm{P(T)M}}^{W_{j+1}+U_{j+1}}(\rho)-\phi_{\mathrm{P(T)M}}^{W_{j+1}+U_{j+1}}(\sigma)}_1
        \leq \mathrm{CC}_{1,1}\left(\phi_{\mathrm{P(T)M}}^{W_{j+1}+U_{j+1}}\right)\norm{\rho-\sigma}_1 \, .
    \end{equation}
    Combining \Cref{eq:sketch-parameter-change,eq:sketch-input-change} via a triangle inequality, we obtain
    \begin{equation}
        \norm{\phi_{\mathrm{PTM}}^{W_{j+1}+U_{j+1}}(\rho)-\phi_{\mathrm{PTM}}^{W_{j+1}}(\rho)}_1
        \leq \sqrt{\frac{d_{\mathrm{\mathrm{out}}}^{(j+1)}}{d_{\mathrm{in}}^{(j+1)}}}\,(\|W_{j+1}\|_{1,1}+\|U_{j+1}\|_{1,1})\,\norm{\phi_{\mathrm{PTM}}^{W_{j}+U_{j}}(\rho)-\phi_{\mathrm{PTM}}^{W_{j}}(\rho)}_1+\sqrt{\frac{d_{\mathrm{\mathrm{out}}}^{(j+1)}}{d_{\mathrm{in}}^{(j+1)}}}\,\|U_{j+1}\|_{1,1}
    \end{equation}
    and
    \begin{equation}
        \norm{\phi_{\mathrm{PM}}^{W_{j+1}+U_{j+1}}(\rho)-\phi_{\mathrm{PM}}^{W_{j+1}}(\rho)}_1
        \leq (\norm{W_{j+1}}_{1,1}+\norm{U_{j+1}}_{1,1})\;\norm{\phi_{\mathrm{PM}}^{W_{j}+U_{j}}(\rho)-\phi_{\mathrm{PM}}^{W_{j}}(\rho)}_1+\norm{U_{j+1}}_{1,1} \, .
    \end{equation}
    We now solve these recurrence relations to get a closed-form bound on $\norm{\phi_{\mathrm{P(T)M}}^{W_{j}+U_{j}}(\rho)-\phi_{\mathrm{P(T)M}}^{W_{j}}(\rho)}_1$, which after using $\|U_\ell\|_{1,1}\leq \frac{1}{L}\|W_\ell\|_{1,1}$, $(1+1/L)^L\leq e$, $d_\mathrm{out}^{(j)}=d_\mathrm{in}^{(j+1)}$ and H\"older's inequality completes the proof. 
\end{proof}

\subsection{Numerical methods}\label{ss:num_methods}
In this section, we provide detailed information regarding the learning task, data generation, and the specific training procedures used to obtain the results presented in Section~\ref{s:numerics}. The code used to generate the numerical results can be found in Ref.~\cite{github_repository}.

\begin{example}[Hamiltonian and dataset generation.]
We consider the generalized cluster Hamiltonian, a standard model exhibiting symmetry-protected topological order. For a chain of $n$ qubits, the Hamiltonian is defined as
\begin{align} H = \sum_{i=1}^n \left(Z_i - J_1 X_i X_{i+1} - J_2 X_{i-1} Z_i X_{i+1}\right)\,, \label{eq:cluster_hamiltonian_methods} 
\end{align}
where $X_i$, $Z_i$ denote Pauli operators on site $i$, and we consider periodic boundary conditions. The ground-state phase diagram in the $(J_1, J_2)$ plane contains four distinct phases~\cite{verresen2017one}.

To generate the dataset, we sample coupling pairs uniformly in the domain $J_1,J_2 \in [-4, 4]$. For each sampled pair, we compute the ground state using exact diagonalization. To rigorously test generalization performance, we utilize a restricted training set of size $N=8$ and evaluate performance on an independent test set of size $1000$.
\end{example}

\begin{example}[Model architectures.]
We employ two distinct $6$-qubit architectures in our simulations: the dynamic PQC and the QCNN.

\begin{paragraph}{Dynamic PQC:}
The architecture begins with a six-qubit data register. This is immediately followed by two sequential blocks of dynamic, mid-circuit operations applied independently to each data qubit. These dynamic blocks operate by preparing an independent auxiliary qubit in a superposition state via a parameterized $R_X$ rotation, which is subsequently measured in the computational basis. The resulting classical bit triggers the conditional application of a parameterized single-qubit unitary operation on the corresponding data qubit. This mechanism acts as a probabilistic quantum channel, injecting controlled stochasticity into the network. Following the dynamic blocks, the architecture employs a series of two-qubit parameterized unitaries arranged in a staggered, brick-wall configuration across neighboring qubits. Finally, the expectation value is evaluated with respect to a target observable on a reduced two-qubit subsystem.
\end{paragraph}

\begin{paragraph}{QCNN:}
The proposed QCNN architecture consists of alternating convolutional and pooling layers. The convolutional layers apply parameterized two-qubit unitaries to neighboring qubit pairs to capture spatial correlations, while the pooling layers systematically trace out subsets of qubits to reduce the system's dimensionality. To induce a wider variance in model complexity and generalization error, we introduce dynamic, mid-circuit operations immediately following each convolutional block and prior to pooling. Specifically, these dynamic layers utilize auxiliary qubits to apply conditional single-qubit rotations based on mid-circuit measurement outcomes. Through this sequential process, the input state is progressively reduced from six qubits down to two, from which the final expectation values are measured to yield the prediction.
\end{paragraph}
\end{example}

\begin{example}[Training.] The models were trained using the Adam optimizer~\cite{kingma2014adam} with a learning rate of $0.005$ over $200$ epochs with random initialization. To ensure our analysis covered a wide variety of trained models and solution norms, we also included a complexity-based regularization term in the training objective. Specifically, we regularized the cost function by adding a term $\propto |\lambda| \norm{W}$, where $\lambda$ is sampled from a normal distribution centered at zero for the different training runs. The optimization aimed to maximize the fidelity between the predicted quantum states and the target labels corresponding to the phases of matter.
In total, we conducted $1400$ independent training runs for the dynamic PQC and $1400$ independent training runs for the QCNN. All simulations were performed using the {\tt PennyLane}~\cite{bergholm2018pennylane} library.
\end{example}

\section*{Code and data availability}
    The code and data generated during the current study are available in Ref.~\cite{github_repository}.

\acknowledgments
P.~R.-G. thanks Ruben Ibarrondo and Frederik vom Ende for fruitful discussions. The authors also want to thank Johannes Jakob
Meyer for useful comments on an earlier version of this manuscript. \mbox{P.~R.-G.} acknowledges support from the Basque Government through Grant No. IT1470-22, and the Elkartek project KUBIBIT - kuantikaren berrikuntzarako ibilbide teknologikoak (ELKARTEK25/79). J.~E., E.~G.-F., F.~J.~S., 
and C.~B.-P. acknowledge the 
BMFTR (MUNIQC-Atoms, 
HYBRID, QuSol, Hybrid++),
the DFG (CRC 
183 and SPP 2514), the Quantum Flagship (Millenion, PasQuanS2), the Munich Quantum Valley, Berlin Quantum and the European Research Council (DebuQC)
for financial support.
E.~G.-F. is a recipient of the 2023 Google PhD Fellowship.

\textit{Note added:} After presenting our work at QTML25, we became aware of ongoing independent work by Michele Grossi, Jogi Suda Neto, and Cenk Tüysüz on PAC-Bayesian bounds for (equivariant) unitary quantum neural networks, currently in preparation.

\bibliography{citations.bib}
\newpage

\vspace*{0.5cm}

\onecolumngrid
\appendix

\makeatletter
\renewcommand{\p@subsection}{}
\makeatother

\renewcommand{\thesubsection}{\thesection.\arabic{subsection}}
\crefname{subsection}{Appendix}{Appendices}
\Crefname{subsection}{Appendix}{Appendices}
\renewcommand{\thetheorem}{\thesection.\arabic{theorem}}
\renewcommand{\theHtheorem}{\thesection.\arabic{theorem}}
\renewcommand{\theHfigure}{\thesection.\arabic{figure}}
\setcounter{figure}{0}
\setcounter{theorem}{0}

\begin{center}
\large{Supplementary Material for ``A PAC-Bayesian approach to generalization for quantum models''
}
\end{center}
\counterwithin{figure}{section}

\section{Perturbation bounds}
\label{a:perturbation-bounds}
In this appendix, we provide the formal proofs for the perturbation bounds that quantify how the output of a quantum model deviates under fluctuations of its parameters. A central challenge in establishing non-uniform guarantees for quantum machine learning (QML) is ensuring that the model's output remains stable within a specified margin when the learned weights are perturbed.

To address this, we employ a unified proof strategy: for each mathematical representation---be it the Process Matrix (PM), Pauli Transfer Matrix (PTM), or the equivariant formalism---we relate the sensitivity of layer $j+1$ to the accumulated perturbation of the previous layer $j$. This recursive approach allows us to solve a general recurrence relation for the entire layered structure.
We begin by establishing a general template (\Cref{lemma:pert_bound_template}) that formalizes this recursive logic for arbitrary layered channels. We then specialize this template to each of the three representations discussed in the main text to yield the final sensitivity bounds.

\subsection{A template for perturbation bounds of layered parameterized channels}

\begin{lemma}[Perturbation bound template for layered parameterized channels]\label{lemma:pert_bound_template}
    Let $\lVert\,\cdot\,\rVert_\ast$ denote an arbitrary matrix norm, and let $L\in\bbN_{\geq0}$ denote the number of layers.
    Let $\calC^{\ww}=\bigcirc_{\ell=1}^L \calC_\ell^{W_\ell}$ be an arbitrary $L$-layer quantum channel paramterized by the matrices $\ww=(W_\ell)_{\ell=1}^L$, and for a fixed paramterized quantum state $\rho(\xx)$, define a function: $f_{\ww}:\calX\to\bbR^K$ as $[f_{\ww}(x)]_k\coloneqq\tr\{\Pi_k\calC^{\ww}[\rho(\xx)]\}$, for $k\in\{1,\ldots,K\}$, where $\Pi_k$ are the effects of a POVM.
    For each layer, let $\uu=(U_\ell)_{\ell=1}^L$ be perturbation matrices fulfilling $\lVert U_\ell\rVert_\ast\leq\frac{1}{L}\lVert W_\ell\rVert_\ast$.
    Let $\rho_j^{\ww}(x) \coloneqq \bigcirc_{\ell=1}^j \calC_\ell^{W_\ell}$ be the intermediate state after $j$ layers, and let $\Delta_j\coloneqq \lVert\rho_j^{\ww+\uu}(x)-\rho_j^{\ww}(x)\rVert_1$ quantify the effect of the perturbation after $j$ layers.
    Suppose that $\calC^{\ww}$ fulfills the following assumptions, for $\mu_j,\nu_j$ real-valued constants independent of $W_j$ and $U_j$:
    \begin{align}
        \left\lVert \calC_{j+1}^{W_{j+1}+U_{j+1}}\left(\rho^{\ww+\uu}_{j}(\xx)-\rho^{\ww}_{j}(\xx)\right)\right\rVert_1 &\leq \mu_{j+1}\left(\lVert W_{j+1}\rVert_\ast+\lVert U_{j+1}\rVert_\ast\right)\Delta_j, \tag{Assumption 1} \label{eq:assumption1}\\
        \left\lVert\left(\calC_{j+1}^{W_{j+1}+U_{j+1}}-\calC_{j+1}^{W_{j+1}}\right)\left(\rho_j^{\ww+\uu}(\xx)\right)\right\rVert_1 &\leq \nu_{j+1}\lVert U_{j+1}\rVert_\ast. \tag{Assumption 2} \label{eq:assumption2}
    \end{align}
    Then the effect of the perturbation on the output of the function fulfills the following bound:
    \begin{align}
        \lVert f_{\ww+\uu}(\xx)-f_{\ww}(\xx)\rVert_\infty \leq e\sum_{i=1}^L\left[\nu_i\lVert U_i\rVert_\ast\left(\prod_{\ell=i+1}^L\mu_\ell\lVert W_\ell\rVert_\ast\right)\right].
    \end{align}
\end{lemma}
\begin{proof}
    We prove the statement directly, as a chain of bounds.

    Let us call $\Pi_k=\lvert k\rangle\!\langle k\rvert$ the measurement outcome corresponding to the $k^\text{th}$ output of $f_{\ww}$.
    Then, the effect of the perturbation on the function output relates to the the effect of the perturbation on the output quantum state:
    \begin{align}
        \lVert f_{\ww+\uu}(\xx)-f_{\ww}(\xx)\rVert_\infty &=\max_{k}\left\lvert\tr\left[\Pi_k\left(\rho^{\ww+\uu}_L(x)-\rho_L^{\ww}(\xx)\right)\right]\right\rvert\\
        &\leq\max_k\lVert\Pi_k\rVert_\infty\left\lVert\rho_L^{\ww+\uu}(\xx)-\rho_L^{\ww}(\xx)\right\rVert_1\\
        &\leq\Delta_L.
    \end{align}
    It is therefore sufficient to upper bound $\Delta_L$.
    For a layer $j$, we relate $\Delta_{j+1}$ to $\Delta_j$ by adding and subtracting a new term in the definition and exploiting the triangular inequality:
    \begin{align}
        \Delta_{j+1} &\coloneqq\left\lVert\calC^{W_{j+1}+U_{j+1}}_{j+1}\left(\rho_j^{\ww+\uu}(\xx)\right) - \calC_{j+1}^{W_{j+1}}\left(\rho_j^{\ww}(\xx)\right)\right\rVert_1 \\
        &= \left\lVert\calC^{W_{j+1}+U_{j+1}}_{j+1}\left(\rho_j^{\ww+\uu}(\xx)\right) - \calC^{W_{j+1}+U_{j+1}}_{j+1}\left(\rho_j^{\ww}(\xx)\right) + \calC^{W_{j+1}+U_{j+1}}_{j+1}\left(\rho_j^{\ww}(\xx)\right) - \calC_{j+1}^{W_{j+1}}\left(\rho_j^{\ww}(\xx)\right)\right\rVert_1 \\
        &\leq \left\lVert\calC^{W_{j+1}+U_{j+1}}_{j+1}\left(\rho_j^{\ww+\uu}(\xx) - \rho_j^{\ww}(\xx)\right)\right\rVert_1 + \left\lVert\left(\calC^{W_{j+1}+U_{j+1}}_{j+1} - \calC_{j+1}^{W_{j+1}}\right)\left(\rho_j^{\ww}(\xx)\right)\right\rVert_1. \\
    \end{align}
    We recognize the two terms match those in the Assumptions of the statement of the Lemma, and hence we can directly upper-bound their sum accordingly:
    \begin{align}
        \Delta_{j+1} &\leq \nu_{j+1}\lVert U_{j+1}\rVert_\ast + \mu_{j+1}\left(\lVert W_{j+1}\rVert_\ast+\lVert U_{j+1}\rVert_\ast\right)\Delta_j.
    \end{align}
    We next extract the general form for this recursive formula, and unroll it up to its $L^\text{th}$ term:
    \begin{align}
        \Delta_L &\leq \sum_{j=1}^L\left[\nu_j \lVert U_j\rVert_\ast \left(\prod_{\ell=j+1}^L\mu_\ell \left(\lVert W_\ell\rVert_\ast+\lVert U_\ell\rVert_\ast\right)\right)\right] \\
        &\leq \sum_{j=1}^L\left[\nu_j \lVert U_j\rVert_\ast \left(\prod_{\ell=j+1}^L\mu_\ell \left(\lVert W_\ell\rVert_\ast+\frac{1}{L}\lVert W_\ell\rVert_\ast\right)\right)\right] \\
        &\leq \sum_{j=1}^L\left[\nu_j \lVert U_j\rVert_\ast \left(\prod_{\ell=j+1}^L\mu_\ell \lVert W_\ell\rVert_\ast\left(1+\frac{1}{L}\right)\right)\right] \\
        &\leq \sum_{j=1}^L\left[\nu_j \lVert U_j\rVert_\ast\left(1+\frac{1}{L}\right)^{L-j} \left(\prod_{\ell=j+1}^L\mu_\ell \lVert W_\ell\rVert_\ast\right)\right] \\
        &\leq \left(1+\frac{1}{L}\right)^{L}\sum_{j=1}^L\left[\nu_j \lVert U_j\rVert_\ast \left(\prod_{\ell=j+1}^L\mu_\ell \lVert W_\ell\rVert_\ast\right)\right] \\
        &\leq e\sum_{j=1}^L\left[\nu_j \lVert U_j\rVert_\ast \left(\prod_{\ell=j+1}^L\mu_\ell \lVert W_\ell\rVert_\ast\right)\right]. \\
    \end{align}
    This completes the proof.
\end{proof}

\subsection{Three perturbation bounds for different channel representations}

\subsubsection{Perturbation bound in the PM formalism}

We begin with our perturbation bound in the PM formalism. Recall that here, the quantum model takes the form 
\begin{equation}
    f_{\vec{w}}^{\mathrm{PM}}(\vec{x})[k] = \mathrm{Tr}(\Pi_k \rho_{\mathrm{\mathrm{out}}}^{\mathrm{PM}}(\vec{x}))\, ,
\end{equation}
where the output state is $\rho_{\mathrm{\mathrm{out}}}^{\mathrm{PM}}(\vec{x}) = \phi_L \circ \cdots \circ \phi_1\left(\rho(\vec{x})\right)$, with each $\phi_j = \phi_{\mathrm{PM}}^{W_j}$ given as a quantum channel characterized by the deviation matrix $W_j$ from the PM of the maximally depolarizing channel. We express our results in 
terms of the $\norm{.}_{1,1}$-norm,
which is a perfectly valid matrix norm, but not a unitarily invariant one.

\begin{lemma}[Perturbation bound for quantum models -- PM framework (restatement of \Cref{lemma:pert_bound_PM})]\label{a-lemma:pert_bound_PM}
        For any $L>0$, let $f_{\vec{w}}^{\mathrm{PM}}: \calX \to \bbR^K$ be an $L$-layer quantum model parameterized by $\vec{w}=\mathrm{vec}(\{W_j\}_{j=1}^L)$. For any $\vec{x} \in \calX$, and any perturbation $\vec{u}=\mathrm{vec}\left(\{U_j\}_{j=1}^L\right)$ such that $\norm{U_j}_{1,1} \leq \frac{1}{L} \norm{W_j}_{1,1}$ holds for all $j$, the change in the output of the quantum model is bounded as
\begin{equation}
    \begin{split}
        \norm{f_{\vec{w}+\vec{u}}^{\mathrm{PM}}(\vec{x})-f_{\vec{w}}^{\mathrm{PM}}(\vec{x})}_\infty
        \leq e \sum_{j=1}^L\left[\norm{U_j}_{1,1}\left(\prod_{\ell=j+1}^L\norm{W_\ell}_{1,1}\right)\right]\, ,
    \end{split}
\end{equation}
where the norm $\norm{A}_{1,1}$ equals the sum of absolute values of all the entries of the matrix $A$.
\end{lemma}
\begin{proof}
    We prove this result by invoking Lemma~\ref{lemma:pert_bound_template}, with $\phi_\text{PM}^{W_j}=\calC_j^{W_j}$.
    We must only show that the PM parameterization fulfills the assumptions, taking the matrix $1$-norm, and with $\mu_j=\nu_j=1$.
    
    First, recall the Pauli twirling identity 
    \begin{align}
        \frac{1}{4^n} \sum_A \sigma_A X \sigma_A = \frac{\Id}{2^n}\tr(X),
    \end{align}
    which holds for all operators $X$. In particular, if $X$ is traceless, as is the case for $\rho_j^{\vec{w}+\vec{u}}(\vec{x})-\rho^{\vec{w}}_j(\vec{x})$ (or any two density matrices), it follows that $(1/4^n) \sum_A \sigma_A X \sigma_A =0$.
    Therefore, we have 
    \begin{align}
        \norm{\phi_\text{PM}^{W_{j+1}+U_{j+1}}(X)}_1 &= \norm{\sum_{A,B} \left(\frac{\delta_{A,B}}{4^n} + W_{j+1}(A,B) + U_{j+1}(A,B) \right) \sigma_A X \sigma_B}_1 \\
        &= \norm{\frac{1}{4^n} \sum_A \sigma_A X\sigma_A + \sum_{A,B} \left(W_{j+1}(A,B) + U_{j+1}(A,B) \right) \sigma_A X \sigma_B}_1 \\
        &= \norm{\sum_{A,B} \left(W_{j+1}(A,B) + U_{j+1}(A,B) \right) \sigma_A X \sigma_B}_1 \\
        &\leq \sum_{A,B} \left| W_{j+1}(A,B) + U_{j+1}(A,B) \right| \norm{\sigma_A X \sigma_B }_1 \\
        &\leq \left( \norm{W_{j+1}}_{1,1} + \norm{U_{j+1}}_{1,1} \right) \norm{X}_1 \label{eq:app_PM_2}.
    \end{align}
    For the last line, we have used the unitary invariance of the 
    trace norm, i.e., 
    $\norm{\sigma_A X \sigma_B}_1 = \norm{X}_1$. Setting $X = \rho_j^{\vec{w}+\vec{u}}(\vec{x})-\rho^{\vec{w}}_j(\vec{x})$, we obtain \ref{eq:assumption1}:
    \begin{align}
        \norm{\phi_\text{PM}^{W_{j+1}+U_{j+1}}\left(\rho_j^{\vec{w}+\vec{u}}(\vec{x})-\rho^{\vec{w}}_j(\vec{x})\right)}_1 &\leq \left( \norm{W_{j+1}}_{1,1} + \norm{U_{j+1}}_{1,1} \right) \Delta_j.
    \end{align}
    We readily recover \ref{eq:assumption2} as well:
    \begin{align}
        \norm{\left(\phi_\text{PM}^{W_{j+1}+U_{j+1}}-\phi_\text{PM}^{W_{j+1}} \right) (\rho_j^{\vec{w}}(\vec{x}))}_1 &=
        \left\|\sum_{A,B} U_{j+1}(A,B)\,\sigma_A \rho_j^{\vec{w}}(\vec{x}) \sigma_B\right\|_1\\
        &\le \sum_{A,B} |U_{j+1}(A,B)|\,\|\sigma_A \rho_j^{\vec{w}}(\vec{x}) \sigma_B\|_1\\
        &= \sum_{A,B} |U_{j+1}(A,B)|\,\|\rho_j^{\vec{w}}(\vec{x})\|_1\\
        &= \|U_{j+1}\|_{1,1}. \label{eq:app_PM_3}
    \end{align}
    These, together with Lemma~\ref{lemma:pert_bound_template}, complete the proof.
\end{proof}

\subsubsection{Perturbation bound in the PTM formalism}

Next, we derive the perturbation bound in the PTM formalism. Here, we work with
\begin{equation}
    f_{\vec{w}}^{\mathrm{PTM}}(\vec{x})[k] = \mathrm{Tr}(\Pi_k \rho_{\mathrm{\mathrm{out}}}^{\mathrm{PTM}}(\vec{x}))\, ,
\end{equation}
where the output state is $\rho_{\mathrm{\mathrm{out}}}^{\mathrm{PM}}(\vec{x}) = \phi_L \circ \cdots \circ \phi_1\left(\rho(\vec{x})\right)$, with each $\phi_j = \phi_{\mathrm{PTM}}^{W_j}$ given as a quantum channel characterized by the deviation matrix $W_j$ from the PTM of a maximally depolarizing channel. 

\begin{lemma}[Perturbation bound for quantum models -- PTM framework (restatement of \Cref{lemma:pert_bound_PTM})]\label{a-lemma:pert_bound_PTM}
        For any $L>0$, let $f_{\vec{w}}^{\mathrm{PTM}}: \calX \to \bbR^K$ be an $L$-layer quantum model parameterized by $\vec{w}=\mathrm{vec}(\{W_j\}_{j=1}^L)$. For any $\vec{x} \in \calX$, and any perturbation $\vec{u}=\mathrm{vec}\left(\{U_j\}_{j=1}^L\right)$ such that $\norm{U_j}_{1,1} \leq \frac{1}{L} \norm{W_j}_{1,1}$ holds for all $j$, the change in the output of the quantum model is bounded as
\begin{equation}
    \begin{split}
        \norm{f_{\vec{u}+\vec{w}}^{\mathrm{PTM}}(\vec{x})-f_{\vec{w}}^{\mathrm{PTM}}(\vec{x})}_\infty
        \leq e\sqrt{d_{\mathrm{\mathrm{out}}}^{(L)}}\sum_{j=1}^L\left[\frac{1}{\sqrt{d_{\mathrm{in}}^{(j)}}}\norm{U_j}_{1,1}\left(\prod_{\ell=j+1}^L\norm{W_\ell}_{1,1}\right)\right]\, ,
    \end{split}
\end{equation}
where the norm $\norm{A}_{1,1}$ equals the sum of absolute values of all the entries of the matrix $A$.
\end{lemma}
\begin{proof}
    The proof mirrors the approach from Lemma~\ref{a-lemma:pert_bound_PM}, just replacing the PM for the PTM parameterization.
    We explicitly confirm the Assumptions, again taking the matrix $1$-norm, but this time with $\mu_j=\nu_j=\sqrt{d_\text{out}^{(j)}/d_\text{in}^{(j)}}$, where $d_\text{out}^{(j)}$ and $d_\text{in}^{(j)}$ correspond to the dimensions of the output and input states for the $j^\text{th}$-layer channel.
    Recall, as before, that we assume $d_\text{out}^{(j)}=d_\text{in}^{(j+1)}$.
    First, let us call $X\coloneqq \left(\rho_j^{\ww+\uu}(\xx)-\rho_j^{\ww}(\xx)\right)$, then
    \begin{align}
        \left\lVert\phi_{\mathrm{PTM}}^{W_{j+1}+U_{j+1}}(X)\right\rVert_1 &= \frac{1}{\sqrt{d_{\mathrm{in}}^{(j+1)}d_{\mathrm{out}}^{(j+1)}}} \left\lVert\sum_{A,B} \left(W_{j+1}(A,B)+U_{j+1}(A,B)\right)\tr\left[\sigma_B^{(j+1)}X\right]\sigma_A^{(j+1)}\right\rVert_1 \\
          \nonumber
        &\leq \frac{1}{\sqrt{d_{\mathrm{in}}^{(j+1)}d_{\mathrm{out}}^{(j+1)}}} \sum_{A,B}\lvert W_{j+1}(A,B)+U_{j+1}(A,B)\rvert\left\lvert\tr\left[\sigma_B^{(j+1)}X\right]\right\rvert \lVert\sigma_A^{(j+1)}\rVert_1 \\
          \nonumber
        &\leq \sqrt{\frac{d_{\mathrm{out}}^{(j+1)}}{d_{\mathrm{in}}^{(j+1)}}}\left(\lVert W_{j+1}\rVert_{1,1} + \lVert U_{j+1}\rVert_{1,1}\right) \lVert X\rVert_1 \\
          \nonumber
        &= \sqrt{\frac{d_{\mathrm{out}}^{(j+1)}}{d_{\mathrm{in}}^{(j+1)}}}\left(\lVert W_{j+1}\rVert_{1,1} + \lVert U_{j+1}\rVert_{1,1}\right) \Delta_j.
          \nonumber
    \end{align}
    In the second line we have used the sub-multiplicativity property and the triangle inequality, and in the third line we have used both the triangle and H\"older's inequalities.
    In the fourth line, we have used the definition of $\Delta_j$ in terms of $X$.
    Next, we have, for any $\rho$
    \begin{align}
        \left\lVert \left(\phi_{\mathrm{PTM}}^{W_{j+1}+U_{j+1}}-\phi_{\mathrm{PTM}}^{W_{j+1}}\right)\left(\rho\right)\right\rVert_1 &\coloneqq \frac{1}{\sqrt{d_{\mathrm{in}}^{(j+1)}d_{\mathrm{out}}^{(j+1)}}}\left\lVert\sum_{A,B} U_{j+1}(A,B) \tr\left[\sigma_B^{(j+1)}\rho\right]\sigma_A^{(j+1)}\right\rVert_1 \\
          \nonumber
        &\leq \frac{1}{\sqrt{d_{\mathrm{in}}^{(j+1)}d_{\mathrm{out}}^{(j+1)}}}\sum_{A,B}\lvert U_{j+1}(A,B)\rvert\left\lvert\tr\left[\sigma_B^{(j+1)}\rho\right]\right\rvert\lVert\sigma_A^{(j+1)}\rVert_1 \\
          \nonumber
        &\leq \sqrt{\frac{d_{\mathrm{out}}^{(j+1)}}{d_{\mathrm{in}}^{(j+1)}}}\lVert U_{j+1}\rVert_{1,1}\lVert\rho\rVert_1 \\
          \nonumber
        &\leq \sqrt{\frac{d_{\mathrm{out}}^{(j+1)}}{d_{\mathrm{in}}^{(j+1)}}}\lVert U_{j+1}\rVert_{1,1}.
          \nonumber
    \end{align}
    This bound is true in particular for $\rho=\rho_j^{\ww+\uu}(\xx)$ as appears in the bound.
    Invoking Lemma~\ref{lemma:pert_bound_template}, we obtain
    \begin{align}
        \lVert f_{\uu+\ww}^{\mathrm{PTM}}(\xx) - f_{\ww}^{\mathrm{PTM}}(\xx)\rVert_\infty &\leq e\sum_{j=1}^L\left[\sqrt{\frac{d_{\mathrm{out}}^{(j)}}{d_{\mathrm{in}}^{(j)}}}\lVert U_j\rVert_{1,1} \prod_{\ell=j+1}^L \left(\sqrt{\frac{d_{\mathrm{out}}^{(\ell)}}{d_{\mathrm{in}}^{(\ell)}}} \lVert W_\ell\rVert_{1,1}\right)\right].
    \end{align}
    By inserting the condition $d_\text{out}^{(j)}=d_\text{in}^{(j+1)}$, we can simplify several of the multiplicative factors:
    \begin{align}
        \lVert f_{\uu+\ww}^{\mathrm{PTM}}(\xx) - f_{\ww}^{\mathrm{PTM}}(\xx)\rVert_\infty &\leq e\sqrt{d_{\mathrm{\mathrm{out}}}^{(L)}}\sum_{j=1}^L\left[\frac{1}{\sqrt{d_{\mathrm{in}}^{(j)}}}\norm{U_j}_{1,1}\left(\prod_{\ell=j+1}^L\norm{W_\ell}_{1,1}\right)\right]\,,
    \end{align}
    thus completing the proof.
\end{proof}

\subsubsection{Perturbation bound for equivariant quantum models}

Finally, we turn our attention to the equivariant case. Here, the model output is defined as
\begin{equation}
    f_{\vec{w}}^{\mathrm{eq}}(\vec{x})[k] = \mathrm{Tr}\big(\Pi_k \rho_{\mathrm{\mathrm{out}}}^{\mathrm{eq}}(\vec{x})\big)\,,
\end{equation}
where the output state is given by the composition of layers $\rho_{\mathrm{\mathrm{out}}}^{\mathrm{eq}}(\vec{x}) = \phi_L \circ \cdots \circ \phi_1\left(\rho(\vec{x})\right)$. The global parameter vector $\vec{w}$ collects the parameter matrices for all layers.

Each layer $\phi_j$ is an equivariant quantum channel parameterized by a set of deviation matrices $W_j = \{ W_{j,\lambda} \}_{\lambda}$, indexed by the irreducible representations (irreps) $\lambda$ of the symmetry group. In the irrep basis, the Choi operator of $\phi_j$ is block-diagonal. Crucially, due to Schur's Lemma, each block decomposes as $J_{j,\lambda} \propto \mathbb{I}_{d_\lambda} \otimes (J_{0,j,\lambda} + W_{j,\lambda})$, meaning it acts as the identity on the $d_\lambda$-dimensional irrep space and is fully specified by the matrix $W_{j,\lambda}$ acting on the multiplicity space.

\begin{lemma}[Equivariant perturbation bound (restatement of \Cref{lemma:pert_bound_equivariant})]\label{a-lemma:pert_bound_equivariant}
For any $L > 0$, let $f_{\vec{w}}^{\mathrm{eq}}: \calX \to \bbR^K$ be an $L$-layer equivariant quantum model parameterized by $\vec{w}$ representing the set of parameter matrices $\{W_{j,\lambda}\}_{j,\lambda}$. For any $\vec{x} \in \mathcal{X}$ and any perturbation $\vec{u}$ representing $\{U_{j,\lambda}\}_{j,\lambda}$ such that $\|U_{j,\lambda}\|_{1} \le \frac{1}{L} \|W_{j,\lambda}\|_{1}$ holds for all $j$ and $\lambda$, the change in the output of the quantum model is bounded as
\begin{equation}
    \norm{f_{\vec{w}+\vec{u}}^{\mathrm{eq}}(\vec{x}) - f_{\vec{w}}^{\mathrm{eq}}(\vec{x})}_\infty \le e \sum_{j=1}^{L} \left[ \norm{U_j}_{\mathrm{eq},1} \left( \prod_{\ell=j+1}^{L} \norm{W_\ell}_{\mathrm{eq},1} \right) \right]\,,
\end{equation}
where the equivariant 1-norm for a layer is defined as $\norm{V}_{\mathrm{eq},1} := \sum_{\lambda} d_\lambda \norm{V_{\lambda}}_1$, with $d_\lambda$ being the dimension of the irreducible representation $\lambda$ and $\norm{\cdot}_1$ being the trace norm.
\end{lemma}

\begin{proof}
    We prove this result directly, using Lemma~\ref{lemma:pert_bound_template}.
    This time, we take the equivariant matrix norm defined in the lemma statement, with $\mu_j=\nu_j=1$.

    We first confirm \ref{eq:assumption1}:
    \begin{align}
        \norm{\phi_{j+1}^{\vec{w}+\vec{u}}(\rho_j^{\vec{w}+\vec{u}}(\xx)-\rho_j^{\vec{w}}(\xx))}_1 &\leq\left(\lVert\phi_{W_{j+1}}\rVert_\diamond+\lVert\phi_{U_{j+1}}\rVert_\diamond\right)\lVert\rho_j^{\vec{w}+\vec{u}}(\xx)-\rho_j^{\vec{w}}(\xx)\rVert_1.
    \end{align}
    For these equivariant channels, the diamond norm is upper bounded by the equivariant 1-norm of the parameters \cite{nechita2018almost}: $\norm{\phi_{V}}_\diamond \le \norm{J_{V}}_1 = \norm{V}_{\mathrm{eq},1}$, thus yielding
    \begin{align}
        \norm{\phi_{j+1}^{\vec{w}+\vec{u}}(\rho_j^{\vec{w}+\vec{u}}-\rho_j^{\vec{w}}))}_1 &\leq \left(\lVert\phi_{W_{j+1}}\rVert_{\mathrm{eq},1}+\lVert\phi_{U_{j+1}}\rVert_{\mathrm{eq},1}\right)\Delta_j.
    \end{align}
    Next, for \ref{eq:assumption2}, we notice that the effect of the perturbation at the channel level simplifies to the action of the perturbation channel $\phi_{U_{j+1}}$ on the state $\rho_j^{\vec{w}}$. Since $\norm{\rho_j^{\vec{w}}}_1 = 1$, we have:
    \begin{align}
        \norm{\phi_{U_{j+1}}(\rho_j^{\vec{w}})}_1 \le \norm{\phi_{U_{j+1}}}_\diamond \le \norm{U_{j+1}}_{\mathrm{eq},1}\,.
    \end{align}
    Having checked the assumptions, Lemma~\ref{lemma:pert_bound_template} implies the stated bound.
\end{proof}

\section{Technical preliminaries on channel parameterizations}
\label{a:preliminaries-param}
To derive our generalization error bounds, it is necessary to obtain upper bounds on the norms of the matrices that describe the different parameterizations. In this section, we also analyze the structural properties of our parameterizations, which enables us to derive interpretable generalization guarantees.

\subsection{PM parameterization}
\label{a:preliminaries-param-PM}
Before turning to explicit norm bounds on $W=\chi-\mathbb{I}/4^n$, it is useful to recall the operational meaning of the Frobenius norm of the PM and its relation to the purity of the Choi operator of the channel. For an $n$-qubit quantum channel $\phi$ with unnormalized Choi operator $J_\phi$, the latter admits the expansion
\begin{equation}
    J_\phi=\sum_{A,B}\chi(A,B)\,
    |\sigma_A\rangle\rangle\langle\langle\sigma_B|,
\end{equation}
where $\{\sigma_A\}$ denotes the $n$-qubit Pauli basis. Using the Hilbert--Schmidt inner product $\langle\langle\sigma_B|\sigma_A\rangle\rangle =\tr(\sigma_B\sigma_A)=2^n\delta_{A,B}$ and the fact that $\chi$ is Hermitian, one finds
\begin{equation}
    \tr(J_\phi^2)
    =2^{2n}\sum_{A,B}|\chi(A,B)|^2
    =2^{2n}\|\chi\|_F^2.
\end{equation}
Equivalently,
\begin{equation}
    \|\chi\|_F^2=\frac{1}{4^n}\tr(J_\phi^2).
\end{equation}
Since for channels with equal input and output dimension the purity of the unnormalized Choi operator satisfies
$1\leq \tr(J_\phi^2)\leq 4^n$, it follows that
\begin{equation}
    \frac{1}{4^n}\leq \|\chi\|_F^2\leq 1,
\end{equation}
with the lower bound attained by the maximally depolarizing channel and the upper bound by unitary channels. As a consequence, the Frobenius norm of the perturbation matrix $W=\chi-\mathbb{I}/4^n$ quantifies the deviation of the channel from the maximally mixed Choi state. We begin by stating two elementary but useful structural properties of $W$, which follow from complete positivity and trace preservation of the associated quantum channel.

\begin{observation}[Vanishing trace of $W$]\label{observation:trace_W}
The diagonal entries of $W$ satisfy
\begin{equation}
    \sum_A W(A,A)=0.
\end{equation}
\end{observation}

\begin{proof}
Using the trace-preserving condition of the channel, we have
\begin{equation}
    \mathbb{I}=\sum_{A,B}\chi(A,B)\sigma_B\sigma_A,
\end{equation}
where $\chi=\mathbb{I}/4^n+W$. Substituting this expression yields
\begin{equation}
    \mathbb{I}
    =\mathbb{I}+\sum_{A,B}W(A,B)\sigma_B\sigma_A,
\end{equation}
which implies
\begin{equation}
    0=\sum_{A,B}W(A,B)\sigma_B\sigma_A.
\end{equation}
Taking the trace against $\sigma_C$ and using orthogonality of the Pauli basis, we obtain for any fixed $C$
\begin{equation}
    0=\sum_A(-1)^{AC}W(A,AC),
\end{equation}
where $AC=0$ if $[\sigma_A,\sigma_C]=0$ and $AC=1$ otherwise. Choosing
$\sigma_C=\mathbb{I}^{\otimes n}$ yields the claim.
\end{proof}

\begin{observation}[Lower bound on diagonal entries]\label{observation:diag_lower_W}
For all $A$, the diagonal entries of $W$ satisfy
\begin{equation}
    W(A,A)\geq -\frac{1}{4^n}.
\end{equation}
\end{observation}

\begin{proof}
Complete positivity of the channel implies that the PM $\chi$ is positive
semidefinite, and therefore
\begin{equation}
    W=\chi-\frac{\mathbb{I}}{4^n}\geq -\frac{\mathbb{I}}{4^n}.
\end{equation}
Since $W$ is Hermitian, the Courant--Fischer--Weyl min--max principle implies that each diagonal entry is lower bounded by the minimal eigenvalue of $W$, which proves the claim.
\end{proof}

We now turn to explicit norm bounds. Throughout, we assume that $W$ has at most $\xi$ non-zero entries.

\begin{lemma}[Upper bound on the $1$-to-$1$ norm]\label{a-lemma:W_11_bound-PM}
Under the above assumptions, the $1$-to-$1$ norm of $W$ satisfies
\begin{equation}
    \norm{W}_{1,1}\leq \frac{2\xi^2+\xi-2}{4^n}.
\end{equation}
\end{lemma}

\begin{proof}
Since the PM $\chi$ is positive semidefinite, its entries satisfy
\begin{equation}
    |\chi(A,B)|\leq \sqrt{\chi(A,A)\chi(B,B)}.
\end{equation}
For $A\neq B$, we have $\chi(A,B)=W(A,B)$, and therefore
\begin{equation}
    |W(A,B)|\leq
    \sqrt{\left(\frac{1}{4^n}+W(A,A)\right)
    \left(\frac{1}{4^n}+W(B,B)\right)}.
\end{equation}
Separating diagonal and off-diagonal contributions and applying the
Cauchy--Schwarz inequality yields
\begin{equation}
    \norm{W}_{1,1}
    \leq \sum_A|W(A,A)|
    +\sum_{A\neq B}\abs{\frac{1}{4^n}+W(A,A)}.
\end{equation}
Using that $W$ has at most $\xi$ non-zero entries and combining this with \Cref{observation:trace_W,observation:diag_lower_W}, the quantity $\sum_A|W(A,A)|$ is maximized when $\xi-1$ diagonal entries equal $-1/4^n$ and one equals $(\xi-1)/4^n$, yielding
\begin{equation}
    \sum_A|W(A,A)|\leq \frac{2(\xi-1)}{4^n}.
\end{equation}
Collecting all terms concludes the proof.
\end{proof}

\begin{lemma}[Upper bound on the Frobenius norm]\label{a-lemma:W_frob_bound-PM}
Under the same 
assumptions, the Frobenius norm of $W$ satisfies
\begin{equation}
    \norm{W}_F^2\leq \frac{\xi^3}{4^{2n}}.
\end{equation}
\end{lemma}

\begin{proof}
For $A\neq B$, positivity of $\chi$ implies
\begin{equation}
    |W(A,B)|^2\leq
    \left(\frac{1}{4^n}+W(A,A)\right)
    \left(\frac{1}{4^n}+W(B,B)\right).
\end{equation}
Separating diagonal and off-diagonal terms and proceeding analogously to the proof of \Cref{a-lemma:W_11_bound-PM}, we obtain
\begin{equation}
    \norm{W}_F^2
    \leq \sum_A|W(A,A)|^2
    +\frac{\xi}{4^{2n}}
    +\xi\sum_A|W(A,A)|^2.
\end{equation}
Using again the constraints from \Cref{observation:trace_W,observation:diag_lower_W}, the sum of squared diagonal entries is maximized when $\xi-1$ entries equal $-1/4^n$ and one equals $(\xi-1)/4^n$, which gives
\begin{equation}
    \sum_A|W(A,A)|^2\leq \frac{\xi(\xi-1)}{4^{2n}}.
\end{equation}
Combining all terms yields the stated bound.
\end{proof}

\subsection{PTM parameterization}
\label{a:preliminaries-param-PTM}
We denote by $d_{\mathrm{in}}$ and $d_{\mathrm{out}}$ the input and output dimensions of the channel $\phi$, respectively, and recall that the PTM entries are given by
\begin{equation}
    R^\phi(A,B)=\frac{1}{\sqrt{d_{\mathrm{in}}\,d_{\mathrm{out}}}}
    \tr(\sigma_A\,\phi(\sigma_B)).
\end{equation}
Note that the PTM of a tensor-product quantum channel factorizes as the tensor product of the PTMs of the individual channels. That is, if $\phi=\bigotimes_{i=1}^k\phi_i$, then $R^\phi=\bigotimes_{i=1}^k R^{\phi_i}$.

We begin by deriving an entrywise bound on the PTM.
\begin{lemma}[Entrywise bound on the PTM]\label{a-lemma:PTM_entry_bound}
For any quantum channel $\phi$, the PTM entries satisfy
\begin{equation}
    |R^\phi(A,B)|\leq \sqrt{\frac{d_{\mathrm{in}}}{d_{\mathrm{out}}}}.
\end{equation}
If, in addition, $\phi$ is contractive in the $\infty$-norm (e.g., if $\phi$ is unital), then the PTM entries satisfy
\begin{equation}
    |R^\phi(A,B)|\leq \sqrt{\frac{d_{\mathrm{out}}}{d_{\mathrm{in}}}}.
\end{equation}
\end{lemma}
\begin{proof}
Both bounds follow directly from Hölder's inequality applied to $\tr(\sigma_A\,\phi(\sigma_B))$, together with the contractivity of quantum channels under the $1$-norm and, in the second case, the $\infty$-norm.
\end{proof}
Since $W=R^\phi-R^{\mathrm{DEP}}$ differs from $R^\phi$ only in the $(\mathbb{I},\mathbb{I})$ entry, the same bounds apply to all non-zero entries of $W$.
\begin{lemma}[Upper bound on the $1$-to-$1$ norm]\label{a-lemma:W_11_bound_PTM}
Assume that $W$ has at most $\xi$ non-zero entries. If $\phi$ is unital, then
\begin{equation}
    \norm{W}_{1,1}\leq \xi\,\sqrt{\frac{d_{\mathrm{out}}}{d_{\mathrm{in}}}}.
\end{equation}
\end{lemma}
\begin{proof}
The claim follows immediately by combining the entrywise bound from \Cref{a-lemma:PTM_entry_bound} with the sparsity assumption.
\end{proof}
We now turn to bounds that do not rely on unitality.
\begin{lemma}[Upper bound on the Frobenius norm]\label{a-lemma:W_F_bound_PTM}
The Frobenius norm of $W$ satisfies
\begin{equation}
    \norm{W}_F^2\leq d_{\mathrm{in}}^2-\frac{d_{\mathrm{in}}}{d_{\mathrm{out}}}.
\end{equation}
\end{lemma}
\begin{proof}
Let $J_\phi$ denote the unnormalized Choi matrix of the quantum channel $\phi$. Expressing $J_\phi$ in the Pauli basis, one finds that $\norm{R^\phi}_F=\norm{J_\phi}_F$. Since the eigenvalues of $J_\phi$ are non-negative and satisfy $\tr(J_\phi)=d_{\mathrm{in}}$, we obtain
\begin{equation}
    \frac{d_{\mathrm{in}}}{d_{\mathrm{out}}}\leq \norm{J_\phi}_F^2\leq d_{\mathrm{in}}^2.
\end{equation}
Using that $R^\phi=R^{\mathrm{DEP}}+W$ and that the supports of $R^{\mathrm{DEP}}$ and $W$ are disjoint, the result follows.
\end{proof}
Finally, combining the Frobenius norm bound with sparsity yields an alternative $1$-to-$1$ norm estimate.
\begin{lemma}[Sparsity-based $1$-to-$1$ bound]\label{a-lemma:W_11_bound_PTM_sparsity}
If $W$ has at most $\xi$ non-zero entries, then
\begin{equation}
    \norm{W}_{1,1}\leq \sqrt{\xi}\,\norm{W}_F
    \leq \sqrt{\xi\left(d_{\mathrm{in}}^2-\frac{d_{\mathrm{in}}}{d_{\mathrm{out}}}\right)}
    \leq \sqrt{\xi}\,d_{\mathrm{in}}.
\end{equation}
\end{lemma}
In the unital case, the bounds of \Cref{a-lemma:W_11_bound_PTM,a-lemma:W_11_bound_PTM_sparsity} are complementary: the former is tighter when $d_{\mathrm{out}}\ll d_{\mathrm{in}}$, while the latter applies more generally without structural assumptions on the channel.

\subsection{Equivariant parameterization}
\label{a:preliminaries-param-equivariant}

In this section, we derive upper bounds on the norms of the matrices that describe the equivariant parameterization. For an equivariant layer $\ell$, the parameters are constrained to the multiplicity spaces of the irreducible representations (irreps) $\lambda$ of the symmetry group $G$. The layer is defined by a set of parameter matrices $\{W_{l,\lambda}\}_\lambda$ of size $m_\lambda \times m_\lambda$.

We first establish an upper bound on the equivariant 1-norm in terms of the equivariant Frobenius norm and the number of parameters.

\begin{lemma}[Upper bound on the equivariant 1-norm]
Let $\Xi_\ell = \sum_\lambda m_\lambda^2$ be the total number of equivariant parameters at layer $\ell$, and let $|G|$ denote the order of the finite symmetry group. The equivariant 1-norm of the weight matrix satisfies
\begin{equation} \label{eq:norm_bound_param}
    \norm{W_{\ell}}_{\mathrm{eq},1} \leq \left( |G| \Xi_\ell \right)^{1/4} \norm{W_{\ell}}_{\mathrm{eq},F}.
\end{equation}
\end{lemma}
\begin{proof}
By definition, the layer-wise equivariant 1-norm is given by
\begin{equation}
    \norm{W_{\ell}}_{\mathrm{eq},1} = \sum_{\lambda} d_\lambda \norm{W_{\ell,\lambda}}_1.
\end{equation}
First, we relate the trace norm of each block to its Frobenius norm. For any matrix $A$ of size $n \times n$, $\norm{A}_1 \leq \sqrt{n} \norm{A}_F$. Since $W_{\ell,\lambda}$ is an $m_\lambda \times m_\lambda$ matrix acting on the multiplicity space, we have $\norm{W_{\ell,\lambda}}_1 \leq \sqrt{m_\lambda} \norm{W_{\ell,\lambda}}_F$. Substituting this back into the sum yields
\begin{equation}
    \norm{W_{\ell}}_{\mathrm{eq},1} \leq \sum_{\lambda} d_\lambda \sqrt{m_\lambda} \norm{W_{\ell,\lambda}}_F.
\end{equation}
We now apply the Cauchy-Schwarz inequality, $\sum_k a_k b_k \leq \sqrt{\sum_k a_k^2} \sqrt{\sum_k b_k^2}$, with $a_\lambda = \sqrt{d_\lambda m_\lambda}$ and $b_\lambda = \sqrt{d_\lambda} \norm{W_{\ell,\lambda}}_F$:
\begin{align}
    \norm{W_{\ell}}_{\mathrm{eq},1} &\leq \left( \sqrt{\sum_{\lambda} d_\lambda m_\lambda} \right) \left( \sqrt{\sum_{\lambda} d_\lambda \norm{W_{\ell,\lambda}}_F^2} \right) \nonumber \\
    &= \sqrt{D_\ell} \norm{W_{\ell}}_{\mathrm{eq},F},
\end{align}
where $D_\ell = \sum_{\lambda} d_\lambda m_\lambda = \text{dim}(\mathcal{H}_\ell)$ is the total dimension of the Hilbert space at layer $\ell$.

Next, we relate the Hilbert space dimension $D_\ell$ to the number of equivariant parameters $\Xi_\ell = \sum_\lambda m_\lambda^2$. Using the Cauchy-Schwarz inequality on the sum defining $D_\ell$:
\begin{equation}
    D_\ell = \sum_{\lambda} d_\lambda m_\lambda \leq \sqrt{\sum_{\lambda} d_\lambda^2} \sqrt{\sum_{\lambda} m_\lambda^2}.
\end{equation}
From representation theory, we know that for a finite group $G$, the sum of squared irrep dimensions equals the order of the group, $\sum_\lambda d_\lambda^2 = |G|$~\cite{serre1977linear}. Thus:
\begin{equation}
    D_\ell \leq \sqrt{|G|} \sqrt{\Xi_\ell}.
\end{equation}
Combining these results gives the stated bound in Eq.~\eqref{eq:norm_bound_param}.
\end{proof}

We can further simplify the bound by removing the dependence on the Frobenius norm. This requires an assumption regarding the spectral properties of the weight matrices.

\begin{lemma}[Spectral upper bound on the equivariant 1-norm]\label{a:lemma-eq1upper}
Let $\eta_\ell = \max_\lambda \norm{W_{\ell,\lambda}}_\infty$ be the maximum spectral norm across all irrep blocks in layer $\ell$. Under the same assumptions as the previous lemma, the equivariant 1-norm satisfies
\begin{equation}
    \norm{W_{\ell}}_{\mathrm{eq},1} \leq \eta_\ell \sqrt{|G| \Xi_\ell}.
\end{equation}
\end{lemma}
\begin{proof}
We begin with the definition of the equivariant Frobenius norm:
\begin{equation}
    \norm{W_{\ell}}_{\mathrm{eq},F} = \sqrt{\sum_{\lambda} d_\lambda \norm{W_{\ell,\lambda}}_F^2}.
\end{equation}
For any matrix $A$ of rank $r$, the Frobenius norm is bounded by $\norm{A}_F \leq \sqrt{r} \norm{A}_\infty$, where $\norm{A}_\infty$ is the spectral norm (the largest singular value). The rank of the block $W_{\ell,\lambda}$ is at most its dimension $m_\lambda$. Thus:
\begin{equation}
    \norm{W_{\ell,\lambda}}_F \leq \sqrt{m_\lambda} \norm{W_{l,\lambda}}_\infty.
\end{equation}
Substituting $\eta_\ell = \max_\lambda \norm{W_{\ell,\lambda}}_\infty$ into the sum yields:
\begin{equation}
    \norm{W_{\ell}}_{\mathrm{eq},F} \leq \sqrt{\sum_{\lambda} d_\lambda m_\lambda \eta_\ell^2} = \eta_\ell \sqrt{\sum_{\lambda} d_\lambda m_\lambda} = \eta_\ell \sqrt{D_\ell}.
\end{equation}
Recalling from the proof of the previous lemma that $D_\ell \leq \sqrt{|G|\Xi_\ell}$, we can bound the Frobenius norm as:
\begin{equation}
    \norm{W_{\ell}}_{\mathrm{eq},F} \leq \eta_\ell \left( |G| \Xi_\ell \right)^{1/4}.
\end{equation}
Now we substitute this back into the bound for the equivariant 1-norm from Eq.~\eqref{eq:norm_bound_param}:
\begin{align}
    \norm{W_{\ell}}_{\mathrm{eq},1} &\leq \left( |G| \Xi_\ell \right)^{1/4} \norm{W_{\ell}}_{\mathrm{eq},F} \nonumber \\
    &\leq \left( |G| \Xi_\ell \right)^{1/4} \cdot \eta_\ell \left( |G| \Xi_\ell \right)^{1/4} \nonumber \\
    &= \eta_\ell \sqrt{|G| \Xi_\ell},
\end{align}
which completes the proof.
\end{proof}

\section{Generalization bounds}
\label{a:generalization-bounds}

Here, we show how the perturbation bounds in~\Cref{a:perturbation-bounds} can be used to derive PAC-Bayesian generalization bounds. Inspired by Refs.~\cite{neyshabur2015norm, behboodi2022pac}, we obtain PAC-Bayes generalization bounds for quantum models in the PM, PTM, and equivariant formalisms. We further present explicit examples illustrating our bounds for specific quantum models, along with comparisons to uniform relaxations of our worst-case PAC-Bayes bounds. Notably, these uniform bounds constitute new derivations and are of independent interest.

After quantifying the stability of a quantum model with respect to parameter fluctuations, the derivation proceeds in three main steps.

\paragraph{Bounding the noise magnitude.} We assume the parameter perturbation follows a Gaussian distribution. To utilize the sensitivity bound established in the first step, we must determine the maximum magnitude this random noise vector can reach with high probability. Using concentration inequalities, we derive a tail bound which guarantees that, with high probability, the norm of the perturbation vector does not exceed a threshold determined by the noise variance and the effective number of parameters (or equivariant dimensions).

\paragraph{Enforcing the margin condition.} The PAC-Bayes margin bound requires that the random perturbation does not alter the model's prediction on the training data. Specifically, the change in the output must be smaller than a fraction of the margin parameter. By combining the perturbation bound with the probabilistic noise limit from step one, we obtain a stability constraint. This constraint dictates the maximum permissible noise variance: to ensure the model output remains stable, the noise variance must scale inversely with the model's complexity and the desired margin.

\paragraph{Computing the generalization bound.} The final component of the PAC-Bayes bound is the \emph{Kullback-Leibler} (KL) divergence between the posterior distribution (centered at the learned parameters) and the prior distribution (centered at zero). For Gaussian distributions, this divergence is proportional to the magnitude of the trained weights divided by the noise variance. We substitute the variance constraint derived in step three into this KL term. This substitution reveals the fundamental trade-off: higher model complexity forces us to use smaller noise variance to maintain stability, which in turn increases the KL divergence and results in a looser generalization bound. Finally, applying a covering number argument to handle the data dependence of the complexity term yields the final generalization guarantee.

\subsection{Noise magnitude, margin condition and union bound in the PM and PTM formalism}

We now proceed with deriving the first two steps in both the PM and PTM formalisms. We denote the effective number of parameters in layer $j$ by $\xi_j$, corresponding to the number of independent entries in $W_j$. 

\subsubsection{Tail bound for Gaussian perturbation}

We consider Gaussian perturbations $U_j(A,B) \sim \mathcal{N}(0,\sigma^2)$ for all entries of the layer matrices. Using the Chernoff bound for the $l_1$ norm, we have
\begin{equation}
    \mathbb{P}(\|U_j\|_{1,1} \ge t) \le 2^{\xi_j} e^{-t^2/(2 \xi_j \sigma^2)}.
\end{equation}
Applying a union bound over all $L$ layers and defining $\xi_{\max} = \max_j \xi_j$, we obtain a high-probability bound:
\begin{equation}
    \mathbb{P}(\|U_j\|_{1,1} \le t \;\forall j) \ge 1 - \sum_{j=1}^L 2^{\xi_j} e^{-t^2/(2 \xi_j \sigma^2)} \ge 1 - e^{-t^2/(2 \xi_{\max} \sigma^2)} \sum_{j=1}^L 2^{\xi_j}.
\end{equation}
Setting the failure probability to $1/2$ gives
\begin{equation}
    t = \sigma \sqrt{2 \xi_{\max} \ln\left(2 \sum_{j=1}^L 2^{\xi_j}\right)}.
\end{equation}
Thus, with probability at least $1/2$, $\|U_j\|_{1} \le t$ holds for all layers simultaneously.

\subsubsection{Margin condition and noise scaling}
Using the perturbation bound derived for the PM and PTM representations, the output deviation under perturbation satisfies 
\begin{equation}
    \|f_{\vec{w}+\vec{u}}(\vec{x}) - f_{\vec{w}}(\vec{x})\|_{\infty} \le e\beta_{\rm P(T)M} \, \max_j \|U_j\|_{1,1}\leq e\beta_{\rm P(T)M} \,\sigma \sqrt{2 \xi_{\max} \ln\left(2 \sum_{j=1}^L 2^{\xi_j}\right)},
\end{equation}
where 
\begin{equation}
    \beta_{\rm PM} = \sum_{j=1}^L \prod_{\ell=j+1}^L \|W_\ell\|_{1,1}, \quad \beta_{\rm PTM} = \sqrt{d_\mathrm{out}^{(L)}}\sum_{j=1}^L\frac{1}{\sqrt{d_\mathrm{in}^{(j)}}}\prod_{\ell=j+1}^L \lVert W_\ell\rVert_{1,1}\,.
\end{equation}
We approximate $\beta_{\rm P(T)M}$ by $\tilde{\beta}_{\rm P(T)M}$ because $\beta_{\rm P(T)M}$ depends on the final weights $\vec{w}$ of the network (which are random variables drawn from the posterior), whereas the prior (and thus the standard deviation $  \sigma  $) must be chosen independently of the data and of the specific hypothesis $\boldsymbol{w}$. The PAC-Bayes margin condition in~\Cref{lemma:pac-bayes_margin} (perturbation $\le \gamma/4$) therefore imposes 
\begin{equation}
    \sigma \le \frac{\gamma (1-1/L)}{4e \tilde{\beta}_{\rm P(T)M} \sqrt{2 \xi_{\max} \ln\left(2 \sum_{j=1}^L 2^{\xi_j}\right)}}\,,
\end{equation}
and holds with probability at least $1/2$.

\subsubsection{KL divergence}

For our Gaussian prior and posterior distributions, the KL divergence reads
\begin{equation}
    \mathrm{KL}(\vec{w}+\vec{u}\|P) = \mathrm{KL}\bigl(\mathcal{N}(\boldsymbol{w},\sigma^2\mathbb{I}),\mathcal{N}(0,\sigma^2\mathbb{I})\bigr) = \frac{\|\boldsymbol{w}\|_2^2}{2 \sigma^2}\,,
\end{equation}
where $\vec{w}=\mathrm{vec}(\{W_j\}_{j=1}^L)$, and therefore
\begin{equation}
    \norm{\vec{w}}_2^2=\sum_{j=1}^L\norm{\mathrm{vec}(W_j)}_2^2=\sum_{j=1}^L\norm{W_j}_F^2.
\end{equation}
Substituting the upper bound on $\sigma$ yields
\begin{equation}
    \mathrm{KL}(\vec{w}+\vec{u}\|P) \le \frac{16e^2 \tilde{\beta}_{\rm P(T)M}^2 \xi_{\max} \ln\left(2 \sum_{j=1}^L 2^{\xi_j}\right)}{\gamma^2 (1-1/L)^2} \sum_{j=1}^L \|W_j\|_F^2.
\end{equation}

Finally,  we can apply \cref{lemma:pac-bayes_margin} and get that for any $\tilde{\beta}_{\rm P(T)M}$ with probability at least $1-\delta$ and for all $\boldsymbol{w}$ such that $|\beta_{\rm P(T)M}-\tilde{\beta}_{\rm P(T)M}|\leq \beta_{\rm P(T)M}/L$, we have
\begin{equation}\label{eq:before_union}
    L_0(f_{\vec{w}}^{\mathrm{P(T)M}}) \leq \hat{L}_{\gamma}(f_{\vec{w}}^{\mathrm{P(T)M}})+ \mathcal{O}\left(\sqrt{\frac{\beta_{\rm P(T)M}^2\:\xi_{\max} \ln\left(\sum_{j=1}^L2^{\xi_j}\right)\:\sum_{j=1}^L\norm{W_j}_F^2+\ln\left(\frac{N}{\delta}\right)}{\gamma^2 N}} \right)\,.
\end{equation}

\subsubsection{Union bound}
The bound in Eq.~\eqref{eq:before_union} holds for a fixed choice of $\tilde{\beta}_{\rm P(T)M}$. To obtain a guarantee that holds simultaneously for all admissible values of $\beta_{\rm P(T)M}$, we apply a union bound over the discretized set of possible approximations $\tilde{\beta}_{\rm P(T)M}$.

Recall that $\beta_{\rm P(T)M}$ is discretized through a covering net so that for every realization of $\beta_{\rm P(T)M}$ there exists a corresponding $\tilde{\beta}_{\rm P(T)M}$ satisfying $|\beta_{\rm P(T)M}-\tilde{\beta}_{\rm P(T)M}|\le \beta_{\rm P(T)M}/L$. Let $\beta_{\rm P(T)M}$ take values in the interval $[\beta_{\min},\beta_{\max}]$. The resulting covering net therefore has size
\begin{equation}
    \mathcal{N}_\beta
    \le
    \frac{L\,\beta_{\max}}{\beta_{\min}} .
\end{equation}
Applying a union bound over all elements of this covering net amounts to replacing $\delta$ in Eq.~\eqref{eq:before_union} by $\delta/\mathcal{N}_\beta$, which guarantees that the overall failure probability remains at most $\delta$.

Depending on the chosen parametrization, $\beta_{\min}$ and $\beta_{\max}$ take different values. This leads to different covering net sizes $\mathcal{N}_\beta$ and therefore to different generalization bounds, as will be shown in the following subsections.

\subsection{PAC-Bayesian generalization bound in the PM formalism}\label{a:pac-bayes_PM}
After deriving a unified bound for both the PM and PTM parametrizations, we now need to determine the corresponding values of $\beta_{\min}$ and $\beta_{\max}$ in each case. These quantities allow us to characterize the size of the covering net and thus obtain the final generalization error bound. The range of $\beta_{\rm PM}$ is established in the following lemma.

\begin{lemma}[Upper and lower bounds on $\beta_{\mathrm{PM}}$]\label{a-lemma:beta_bounds-PM}
Let
\begin{equation}
\beta_{\mathrm{PM}} \coloneqq \sum_{j=1}^L \left( \prod_{\ell=j+1}^L \norm{W_\ell}_{1,1} \right).
\end{equation}
Then $\beta_{\mathrm{PM}}$ satisfies the bounds 
\begin{equation}
    1 \;\leq\; \beta_{\mathrm{PM}}
    \;\leq\;M(n,L,\xi_{\max}),
\end{equation}
where 
\begin{equation}
M(n,L,\xi_{\max})
:=
\frac{\left(\frac{2\xi_{\max}^2+\xi_{\max}-2}{4^n}\right)^{L}-1}
{\left(\frac{2\xi_{\max}^2+\xi_{\max}-2}{4^n}\right)-1}.
\end{equation}
\end{lemma}

\begin{proof}
We first establish the lower bound. Since all norms $\norm{W_\ell}_{1,1}$ are non-negative, every term in the sum defining $\beta_{\mathrm{PM}}$ is non-negative. In particular, for $j = L$ the product is taken over the empty index set, and by convention the empty product equals $1$. Therefore,
\begin{equation}
\beta_{\mathrm{PM}} \geq 1.
\end{equation}

We now turn to the upper bound. Looking at the definition of $\beta_{\rm PM}$, each term in the sum is a product of non-negative quantities and is therefore monotone increasing in each individual norm $\norm{W_j}_{1,1}$. Consequently, $\beta_{\mathrm{PM}}$ is maximized when all norms simultaneously attain their maximal allowed value. Denoting by $C$ a uniform upper bound such that $\norm{W_j}_{1,1} \leq C$ for all layers $j$, we obtain
\begin{equation}
    \beta_{\mathrm{PM}}
    \leq \sum_{j=1}^L C^{L-j}
    = \frac{C^L - 1}{C - 1}.
\end{equation}
Taking $C=\frac{2\xi_{\max}^2+\xi_{\max}-2}{4^n}$, as established in \cref{a:preliminaries-param-PM}, to obtain a uniform bound on all layers, completes the proof.
\end{proof}

\cref{a-lemma:beta_bounds-PM} provides the physical bounds on $ \beta_{\rm PM} $, which we use to determine the covering net size required for the generalization error bound. However, note that $ \xi \ln(L 2^\xi) > 1 $. Therefore, if
\[
\frac{\beta^2 \sum_{j=1}^L \|W_j\|_F^2}{\gamma^2 N} \geq 1 ,
\]
the error term in Eq.~\eqref{eq:before_union} would exceed 1, which would render the bound trivial.

For this reason, the maximum value of $\beta$ used in the covering net size can be taken as the minimum between the physically allowed maximum value of $\beta$ and the value of $\beta$ that would make the bound trivial. Using this observation together with Eq.~\eqref{eq:before_union} and \cref{a-lemma:beta_bounds-PM}, we can now state the generalization error bound for the PM parameterization presented in the main text.

\begin{theorem}[PAC-Bayes generalization bounds for quantum models-- PM parametrization (restatement of \cref{th:pac-bayes-PM})]
      For any $L > 0$, let $f_{\vec{w}}^{\mathrm{PM}}: \mathcal{X} \to \mathbb{R}^K$ be an $L$-layer quantum model represented via the PM framework $\phi_j = \phi_{\mathrm{PM}}^{W_j}$ with $\xi_j\coloneqq \norm{W_j}_0$ and $\xi_{\max}=\max_j \xi_j$. Then, for any $\gamma > 0$, with probability at least $1-\delta$ over the training set of size $N$, for any $\vec{w}= \mathrm{vec}(\{W_j\}_{j=1}^L)$, we have
    \begin{equation}
        L_0(f_{\vec{w}}^{\mathrm{PM}}) \leq \hat{L}_{\gamma}(f_{\vec{w}}^{\mathrm{PM}})+ \mathcal{O}\left(\sqrt{\frac{\beta_{\rm PM}^2\:\xi_{\max} \ln\left(\sum_{j=1}^L2^{\xi_j}\right)\:\sum_{j=1}^L\norm{W_j}_F^2+\ln\left(\frac{N \;L\min\left\{\frac{\gamma \sqrt{N}}{\sqrt{\sum_{j=1}^L\norm{W_j}_F^2}},\:M(n,L,\xi_{\max})\right\}}{\delta}\right)}{\gamma^2 N}} \right)\,,
    \end{equation}
    where
    \begin{equation}\label{eq:beta-pm}
        \beta_{\rm PM}=\sum_{j=1}^L\left(\prod_{\ell=j+1}^L\norm{W_\ell}_{1,1}\right)\,,
    \end{equation}
    and
    \begin{equation}
        M(n,L,\xi_{\max})=\frac{\left(\frac{2\xi_{\max}^2+\xi_{\max}-2}{4^n}\right)^{L}-1}{\left(\frac{2\xi_{\max}^2+\xi_{\max}-2}{4^n}\right)-1}\,.
    \end{equation}
\end{theorem}
The logarithmic term in our bound may at first seem somewhat difficult to interpret, and the reader might be concerned that it could scale unfavorably with the number of layers $L$. However, we can study the asymptotic behavior of $M(n,L,\xi_{\max})$  to identify the regime in which its dependence on $L$ is at most logarithmic. For this analysis, we take$C=(2\xi_{\max}^2+\xi_{\max}-2)/4^n$ and examine the asymptotic behavior of this quantity as a function of $C$. If $C>1$, then
\begin{equation}
\frac{C^{L}-1}{C-1}
\sim \frac{C^{L}}{C-1}
=
\Theta(C^{L-1}),
\end{equation}
so $\beta_{\mathrm{PM}}$ grows exponentially with $L$. Consequently, the generalization bound contains a depth-dependent contribution that cannot be absorbed into logarithmic factors. In contrast, if $C\to 0$, then
\begin{equation}
\frac{C^{L}-1}{C-1}
=
1+\mathcal{O}(C),
\end{equation}
so that $\beta_{\mathrm{PM}} = 1+\mathcal{O}(C)$ and the dependence on $L$ disappears up to lower-order corrections. Moreover, taking the limit $C\to 1$ and applying L'Hôpital's rule yields
\begin{equation}
\lim_{C\to 1}\frac{C^{L}-1}{C-1}
=
L,
\end{equation}
so in this regime $\beta_{\mathrm{PM}} = \Theta(L)$ and its contribution to the generalization bound appears only logarithmically. Finally, from the definition
\begin{equation}
C = \frac{2\xi_{\max}^2+\xi_{\max}-2}{4^n},
\end{equation}
we observe that $C = o(1)$ whenever $\xi_{\max} = o(2^n)$. In this sparsity regime, $\beta_{\mathrm{PM}}$ remains at most linear in $L$, and therefore all depth-dependent contributions enter only through logarithmic factors.
\\

\paragraph{Comparison to uniform bounds.}
Naturally, perturbation bounds can also be used to derive uniform generalization bounds via Rademacher complexities and covering numbers, following a proof strategy that is well known in classical ML (compare, e.g., \cite{dudley_1999, Barlett2001}) and that was introduced to QML in \cite{Du2022, caro2022generalization}. Hence, for each of our PAC-Bayesian generalization bounds, we provide a comparison to a uniform generalization bound derived from the respective perturbation bound.
We note that this comparison requires us to prove new uniform generalization bounds, tailored to our chosen parameterizations by taking the specific perturbation bounds and parameter norm bounds into account. 

In the proof of \Cref{a-lemma:pert_bound_PM}, we have in particular shown that
\begin{equation}\label{eq:beginning-uniform-pm-bound}
    \norm{\phi^{W}_{\mathrm{PM}} - \phi^{W'}_{\mathrm{PM}} }_{1\to 1}
    \leq \norm{W - W'}_{1,1} \, .
\end{equation}
Consequently, using \cite[Proposition 3.41, Point 2]{Watrous_2018} and a telescopic sum, we see that
\begin{align}
    \norm{f_{\vec{w}}^{\mathrm{PM}} - f_{\vec{w}'}^{\mathrm{PM}}}_\infty
    &\leq \norm{\phi^{W_L}_{\mathrm{PM}}\dots(\phi^{W_1}_{\mathrm{PM}}(\rho(\vec{x})))) - \phi^{W'_L}_{\mathrm{PM}}\dots(\phi^{W'_1}_{\mathrm{PM}}(\rho(\vec{x}))))}_1\\
    \nonumber
    &\leq \sum_{j=1}^L \norm{\phi^{W_j}_{\mathrm{PM}} - \phi^{W'_j}_{\mathrm{PM}} }_{1\to 1}\cdot \underbrace{\norm{\rho(\vec{x})}_1}_{=1}\\
    \nonumber
    &\leq \sum_{j=1}^L\norm{W_j - W'_j}_{1,1} \, .
\end{align}
Next, using \Cref{a-lemma:W_11_bound-PM}, we know that for each layer $j$, the set of admissible matrices $W_j$ lies inside the $\norm{\cdot}_{1,1}$-ball of radius $R_j\leq (2\xi_j^2+\xi_j-2)/4^n\leq 3\xi_j^2/4^n$ around $0$ in the space of $4^n\times 4^n$ matrices, and admissible matrices have at most $\xi_i$ many non-zero complex entries.
We can now use standard covering number bounds for norm balls in real vector spaces \cite[Corollary 4.2.13]{vershynin2018highdimensional}, together with the fact that there are $\binom{16^n}{\xi_i}$ possible ways of placing the $\xi_i$ many non-zero entries among the overall $16^n$ entries, as well as approximate monotonicity of covering numbers (see, e.g., \cite[Section 4.2]{vershynin2018highdimensional}) to see that the relevant covering number is bounded by
\begin{align}
    \mathcal{N}(\{f_{\vec{w}}^{\mathrm{PM}}\}_{\vec{w}}, \norm{\cdot}_\infty, \varepsilon)
    &\leq \prod_{j=1}^L \begin{cases} 
        \binom{16^n}{\xi_j} \mathcal{N}(B^{2\xi_j}_{R_j, \norm{\cdot}_{1,1}}, \norm{\cdot}_{1,1}, \varepsilon/2L) &\textrm{if $\varepsilon/2L < R_j$}\\
        \nonumber
        1 &\textrm{else}
    \end{cases}\\
    &\leq \prod_{j=1}^L \begin{cases}
        \left(\frac{16^n e}{\xi_j}\right)^{\xi_j} \cdot\left(\frac{6LR_j}{\varepsilon}\right)^{2\xi_j} \quad &\textrm{if $\varepsilon/2L < R_j$}\\
        1 &\textrm{else} 
    \end{cases} \, .
\end{align}
Here, we use $B_{r,\norm{\cdot}}^{d}$ to denote the $\norm{\cdot}$-ball of radius $r$ around $0$ in $\mathbb{R}^d$.
Accordingly, the metric entropy can be bounded as
\begin{equation}
    \log \mathcal{N}(\{f_{\vec{w}}^{\mathrm{PM}}\}_{\vec{w}}, \norm{\cdot}_\infty, \varepsilon)
    \leq \mathcal{O}\left(  \sum_{j=1}^L 
    \begin{cases} 
        \xi_j \log\left(\frac{4^n L R_j}{\varepsilon}\right)\quad &\textrm{if $\varepsilon/2L < R_j$}\\
        0 &\textrm{else}
    \end{cases}\right) \, ,
\end{equation}
where the $\widetilde{\mathcal{O}}$ hides factors logarithmic in the leading order. 
The regime in which term $i$ contributes non-trivially to the bound is $\varepsilon < 2L R_i$. 

Combining Rademacher complexities \cite{Barlett2001} and Dudley's theorem \cite{dudley_1999} (see, e.g., 
Ref.\ \cite[Supplementary Note 3]{caro2022generalization} and 
Ref.\ \cite{Barlett2001} for details), this leads to the following uniform generalization bound: With high (constant) probability over the choice of a training data set of size $N$, we have 
\begin{align}
    L_0(f_{\vec{w}}^{\mathrm{PM}})
    &\leq \hat{L}_\gamma (f_{\vec{w}}^{\mathrm{PM}}) + \mathcal{O}\left( \sqrt{\frac{\sum_{j=1}^L \xi_j \log\left(4^n L R_j\right)\cdot \min\{2LR_j, 1/2\}}{\gamma^2 N}}\right) \\
    &\leq \hat{L}_\gamma (f_{\vec{w}}^{\mathrm{PM}}) + \mathcal{O}\left( \sqrt{\frac{\sum_{j=1}^L \xi_j \log\left( L \xi_j^2 \right)\cdot \min\{6L\xi_j^2 / 4^n, 1/2\}}{\gamma^2 N}}\right) \, ,\nonumber
\end{align}
where $\vec{w}$ is the 
parameter vector after training.

Let us now compare this uniform generalization bound with our PAC-Bayes generalization bound for the PM formalism, which (compare \cref{th:pac-bayes-PM}) is
\begin{equation}
    L_0(f_{\vec{w}}^{\mathrm{PM}}) \leq \hat{L}_{\gamma}(f_{\vec{w}}^{\mathrm{PM}})+ \widetilde{\mathcal{O}}\left(\sqrt{\frac{\beta_{\rm PM}^2\:\xi_{\max} \ln\left(\sum_{j=1}^L2^{\xi_j}\right)\:\sum_{j=1}^L\norm{W_j}_F^2}{\gamma^2 N}} \right)\,,
\end{equation}
with $\beta_{\rm PM}=\sum_{j=1}^L\left(\prod_{\ell=j+1}^L\norm{W_\ell}_{1,1}\right)$.
For simplicity, we make the comparison for the case $\xi_j=\xi_{\mathrm{max}}$ for all $j$. In this case, we see that our PAC-Bayesian bound improves upon the uniform bound in the regime 
\begin{equation}
    \beta_{\rm PM}^2\:\xi_{\max} \ln\left(L 2^{\xi_{\max}}\right)\:\sum_{j=1}^L\norm{W_j}_F^2
    \leq \mathrm{const}\cdot  L\xi_{\max} \log\left( L \xi_{\max}^2\right)\cdot \min\{6L\xi_{\max}^2 / 4^n, 1/2\} \, ,
\end{equation}
which we can rearrange to
\begin{equation}
    \beta_{\rm PM}^2\:\sum_{j=1}^L\norm{W_j}_F^2
    \leq \mathrm{const}\cdot  \frac{L \log\left( L \xi_{\max}^2\right)\cdot \min\{6L\xi_{\max}^2 / 4^n, 1/2\}}{\ln\left(L 2^{\xi_{\max}}\right)}\, .
\end{equation}
As discussed in the previous section, by substituting the explicit upper bound on $\norm{W_l}_{1,1}$ from \Cref{a-lemma:W_11_bound-PM} into \Cref{a-lemma:beta_bounds-PM}, we get
 $1\leq \beta_{\rm PM}\leq M(n,L,\xi_{\max})$, with
 \begin{equation}
 M(n,L,\xi_{\max})=\frac{\left(\frac{2\xi_{\max}^2+\xi_{\max}-2}{4^n}\right)^L-1}{\left(\frac{2\xi_{\max}^2+\xi_{\max}-2}{4^n}\right)-1}.
\end{equation}
 Therefore, a sufficient condition for the PAC-Bayesian bound to outperform the uniform bound is
\begin{equation}
    \sum_{j=1}^L\norm{W_j}_F^2
    \leq \mathrm{const}\cdot \frac{L \log\left( L \xi_{\max}^2\right)\cdot \min\{6L\xi_{\max}^2 / 4^n, 1/2\}}{M(n,L,\xi_{\max})^2 \ln\left(L 2^{\xi_{\max}}\right)} \, ,
\end{equation}
This condition gives a regime for parameters post-training in which our non-uniform PAC-Bayesian bound is better than the uniform bound for the same QML model class. 
Using that $\norm{W_j}_F^2 \leq \frac{\xi_{\max}^3}{4^{2n}}$ by \Cref{a-lemma:W_frob_bound-PM}, we can obtain a (more pessimistic) parameter-independent sufficient condition for the PAC-Bayesian bound to be superior, namely 
\begin{equation}\label{eq:PM_comparison}
    \frac{\xi_{\max}^3}{4^{2n}}
    \le
    \mathrm{const}\cdot
    \frac{
        \log\!\left(L \xi_{\max}^2\right)\,
        \min\{6L\xi_{\max}^2 / 4^n,\, 1/2\}
    }{
        M(n,L,\xi_{\max})^2\,
        \ln\!\left(L 2^{\xi_{\max}}\right)
    } .
\end{equation}
Recall that when
\begin{equation}
    \frac{2\xi_{\max}^2+\xi_{\max}-2}{4^n} > 1 ,
\end{equation}
the quantity $M(n,L,\xi_{\max})$ grows exponentially in $L$. Consequently, the right-hand side of~\eqref{eq:PM_comparison} decays exponentially with $L$, making this regime clearly unfavorable. We therefore restrict our analysis to
\begin{equation}
    0 < \frac{2\xi_{\max}^2+\xi_{\max}-2}{4^n} \le 1 .
\end{equation}
The case where this ratio is zero is excluded, since it corresponds to $\xi_{\max}=0$, i.e., a maximally depolarizing channel. In that situation our bound vanishes, while the uniform bound remains non-zero, so the comparison is not meaningful. When the ratio is strictly smaller than one, we have
\begin{equation}
    M(n,L,\xi_{\max})
    =
    1 + \mathcal{O}\!\left(
        \frac{2\xi_{\max}^2+\xi_{\max}-2}{4^n}
    \right),
\end{equation}
and in particular $M(n,L,\xi_{\max})=\mathcal{O}(1)$ whenever $\xi_{\max}=o(2^n)$. We now analyze the two regimes determined by the minimum in~\eqref{eq:PM_comparison}. If $6L\xi_{\max}^2 / 4^n \le 1/2$, which is equivalent to $\xi_{\max} \le \frac{2^n}{\sqrt{12L}}$, then $M(n,L,\xi_{\max})=\mathcal{O}(1)$ and the comparison reduces, up to constants, to
\begin{equation}
    \frac{\xi_{\max}^3}{4^{2n}}
    \lesssim
    \frac{
        \bigl(\log L + 2\log \xi_{\max}\bigr)
        \, L \xi_{\max}^2
    }{
        \bigl(\ln L + \xi_{\max}\ln 2\bigr)
        \, 4^n
    } .
\end{equation}
Ignoring logarithmic factors, this inequality requires $\xi_{\max} \lesssim \widetilde{\mathcal{O}}(\sqrt{L}\,2^n)$. However, this condition is weaker than the defining assumption of the present regime, namely $\xi_{\max} \lesssim 2^n/\sqrt{L}$, which is more restrictive for large $L$. Therefore, throughout this regime, whenever its defining condition holds, the comparison inequality is automatically satisfied up to logarithmic factors, and this entire region is favorable. On the other hand, when $6L\xi_{\max}^2 / 4^n \ge 1/2$, the minimum equals $1/2$ and~\eqref{eq:PM_comparison} becomes, up to constants,
\begin{equation}
    \frac{\xi_{\max}^3}{4^{2n}}
    \lesssim
    \frac{
        \log(L\xi_{\max}^2)
    }{
        M(n,L,\xi_{\max})^2\,
        \ln(L2^{\xi_{\max}})
    } .
\end{equation}
Assuming again that $\xi_{\max}=o(2^n)$ so that $M(n,L,\xi_{\max})=\mathcal{O}(1)$, this yields the favorable scaling $\xi_{\max}\le \widetilde{\mathcal{O}}(2^n)$. Finally, at the boundary case where $\frac{2\xi_{\max}^2+\xi_{\max}-2}{4^n}=1$, one has $M(n,L,\xi_{\max})=L$, so the right-hand side of~\eqref{eq:PM_comparison} acquires an additional factor $L^{-2}$, producing a polynomial suppression in $L$ and signaling the onset of the unfavorable regime. Combining the previous considerations, we conclude that a sufficient and essentially sharp condition ensuring that the PAC–Bayesian bound is superior, up to logarithmic factors, is $\xi_{\max}=o(2^n)$. In this regime $M(n,L,\xi_{\max})=\mathcal{O}(1)$ and the right-hand side of~\eqref{eq:PM_comparison} does not suffer from exponential or polynomial suppression in $L$. 
\\

\paragraph{PAC-Bayes generalization bound for dynamic PQCs.}
Dynamic parameterized quantum circuits
\cite{deshpande2024dynamicparameterizedquantumcircuits} generalize
standard PQCs by allowing mid-circuit measurements and classical feedforward. Their basic building block is a parameterized dynamic operation implementing a probabilistic unitary transformation. 
The associated 
quantum channel is
\begin{equation}
\mathcal{F}(\theta,\phi,\varphi)(\cdot)
=
\cos^2\!\left(\tfrac{\theta}{2}\right)(\cdot)
+
\sin^2\!\left(\tfrac{\theta}{2}\right)\,
\mathcal{U}(\phi,\varphi)(\cdot)\mathcal{U}^\dagger(\phi,\varphi),
\end{equation}
where
\begin{equation}
\mathcal{U}(\phi,\varphi)=
\begin{pmatrix}
\cos\varphi & -e^{-i\phi}\sin\varphi\\
e^{i\phi}\sin\varphi & \cos\varphi
\end{pmatrix}.
\end{equation}
We consider a 
single-layer architecture consisting of two consecutive dynamic operations per 
qubit, followed by a global unitary, which we treat as one effective layer. In this setting $L=1$ and $\beta=1$, and the leading contribution to the non-uniform PAC-Bayes bound is controlled by the Frobenius norm $\|W\|_F^2$, where $W=\chi-\mathbb{I}/4^n$ and $\chi$ is the PM of the induced channel. As shown in the previous subsection, $\|\chi\|_F^2$ is proportional to the purity of the unnormalized Choi operator of the channel, and therefore quantifies the deviation from the maximally mixed Choi state. Since
$\tr(\chi)=1$, we have
\begin{equation}
\|W\|_F^2=\|\chi\|_F^2-\frac{1}{4^n}.
\end{equation}
Bounding the generalization error thus reduces to analyzing the purity of the PM associated with the dynamic PQC architecture.

Unitary transformations preserve the purity of the Choi operator and, equivalently, the Frobenius norm of the PM. Consequently, the global unitary (as well as any intermediate unitaries) does not affect $\|\chi\|_F^2$, and only the local dynamic blocks contribute to the generalization bound.

Because the dynamic operations act locally on each qubit, the overall channel as well as the norm of the PM factorize as
\begin{equation}
\phi=\bigotimes_{i=1}^n \phi_i,\quad \|\chi\|_F^2=\prod_{i=1}^n \|\chi_i\|_F^2,
\end{equation}
where $\chi_i$ denotes the PM of the effective single-qubit channel
\begin{equation}
\phi_i=
\mathcal{F}(\theta^{(i)}_1,\phi^{(i)}_1,\varphi^{(i)}_1)\circ
\mathcal{F}(\theta^{(i)}_2,\phi^{(i)}_2,\varphi^{(i)}_2).
\end{equation}
This tensor-product structure has an important computational consequence: The PM of the full $n$-qubit channel can be evaluated efficiently by computing the PM of each effective single-qubit channel independently. As a result, both $\|\chi\|_F^2$ and the corresponding PAC-Bayes bound can be computed with a cost that scales linearly in $n$. In contrast, for a generic $n$-qubit channel without such a factorized structure, the PM acts on a $4^n$-dimensional space and its explicit construction requires resources that scale exponentially with $n$.

Each dynamic operation admits a Kraus representation with two operators, 
\begin{equation}
\widetilde{K}_0=\cos\!\left(\tfrac{\theta}{2}\right)\mathbb{I},
\qquad
\widetilde{K}_1=\sin\!\left(\tfrac{\theta}{2}\right)\mathcal{U}(\phi, \theta).
\end{equation}
The composition of two such operations yields four Kraus operators $\{K_i\}_{j=0}^3$. From this point on, we focus on a single-qubit channel and, for notational simplicity, drop the subscript $i$, writing $\chi$ for the PM of the effective single-qubit channel $\phi_i$.

For a single-qubit channel with Kraus operators $\{K_j\}$, the purity of the PM can be expressed as
\begin{equation}
\tr(\chi^2)=\frac{1}{4}\sum_{i,j}|\tr(K_i^\dagger K_j)|^2,
\end{equation}
where the prefactor corresponds to the single-qubit case.
In the present setting, the Kraus operators take the form
\begin{align}
    K_0&=\cos(\theta_1/2)\cos(\theta_2/2)\,\mathbb{I},\\
    K_1&=\sin(\theta_1/2)\cos(\theta_2/2)\,U(\phi_1,\varphi_1),\\
    K_2&=\cos(\theta_1/2)\sin(\theta_2/2)\,U(\phi_2,\varphi_2),\\
    K_3&=\sin(\theta_1/2)\sin(\theta_2/2)\,
    U(\phi_2,\varphi_2)U(\phi_1,\varphi_1).
\end{align}
Writing $K_i=c_i U_i$, with $U_i$ unitary, we obtain
\begin{equation}
    \tr(\chi^2)=\frac{1}{4}\sum_{i,j}c_i^2 c_j^2
    \,|\tr(U_i^\dagger U_j)|^2.
\end{equation}
The minimum of $\tr(\chi^2)$ is achieved when
$\tr(U_i^\dagger U_j)=0$ for all $i\neq j$, yielding
\begin{equation}
    \tr(\chi^2)=\sum_i c_i^4
    =\frac{1}{4}\left(1+\cos^2\theta_1\right)
     \left(1+\cos^2\theta_2\right).
\end{equation}
In particular, for $\theta_1=\theta_2=\pi/2$ one finds $\tr(\chi^2)=1/4$, corresponding to the single-qubit maximally depolarizing channel. An explicit realization is obtained by choosing $\phi_1=\varphi_1=\varphi_2=\pi/2$ and $\phi_2=0$, for which the unitaries $\{U_i\}$ form an orthogonal Pauli set.

We therefore conclude that, for this dynamic PQC architecture, tighter generalization guarantees are obtained when each effective single-qubit channel approaches the maximally depolarizing channel. From the Choi perspective, this corresponds to the regime in which the local Choi operators are close to the maximally mixed state, providing an explicit and physically implementable realization of the small-norm regime.

\subsection{PAC-Bayesian generalization bound in the PTM formalism}\label{a:pac_bayes-PTM}
We now determine the corresponding bounds on $\beta_{\rm PTM}$ in order to characterize the covering net size and obtain the final generalization error bound for the PTM parametrization.

\begin{lemma}[Upper and lower bounds on $\beta_{\mathrm{PTM}}$]\label{a-lemma:beta_bounds-PTM}
Let
\begin{equation}
\beta_\mathrm{PTM}
        = \sqrt{d_\mathrm{out}^{(L)}}\sum_{j=1}^L\frac{1}{\sqrt{d_\mathrm{in}^{(j)}}}\prod_{\ell=j+1}^L \lVert W_\ell\rVert_{1,1} \, .
\end{equation}
Then, $\beta_{\mathrm{PTM}}$ satisfies the bounds 
\begin{equation}
    \sqrt{\frac{d_{\rm out}^{(L)}}{d_{\rm in}^{(L)}}} \;\leq\; \beta_{\mathrm{PTM}}
    \;\leq\; F(C_\ell),
\end{equation}
where 
\begin{equation}
F(C_\ell):=\sqrt{d_\mathrm{out}^{(L)}}\sum_{j=1}^L\frac{1}{\sqrt{d_\mathrm{in}^{(j)}}}\prod_{\ell=j+1}^L C_\ell.
\end{equation}
and $C_\ell$ is an upper bound on $\|W_\ell\|_{1,1}$. 
\end{lemma}
The proof follows by applying the same reasoning as in the PM case established in \cref{a-lemma:beta_bounds-PM}.

As we did in \cref{a:pac-bayes_PM} for the PM parametrization we can now state the generalization error bound for the PTM.
\begin{theorem}[PAC-Bayes generalization bound for quantum models -- PTM parameterization (restatement of \cref{th:pac-bayes-PTM})]
      For any $L > 0$, let $f_{\vec{w}}^{\mathrm{PTM}}: \mathcal{X} \to \mathbb{R}^K$ be an $L$-layer quantum model represented via the PTM framework $\phi_j = \phi_{\mathrm{PTM}}^{W_j}$ with $\xi_j\coloneqq \norm{W_j}_0$ and $\xi_{\max}=\max_j \xi_j$. Then, for any $\gamma > 0$, with probability at least $1-\delta$ over the training set of size $N$, for any $\vec{w}= \mathrm{vec}(\{W_j\}_{j=1}^L)$, we have 

    \begin{equation}\label{eq:formal_PM}
        L_0(f_{\vec{w}}^{\mathrm{PTM}}) \leq \hat{L}_{\gamma}(f_{\vec{w}}^{\mathrm{PTM}})+ \mathcal{O}\left(\sqrt{\frac{\beta_{\rm PTM}^2\:\xi_{\max} \ln\left(\sum_{j=1}^L2^{\xi_j}\right)\:\sum_{j=1}^L\norm{W_j}_F^2+\ln\left(\frac{N \;L\min\left\{\frac{\gamma \sqrt{N}}{\sqrt{\sum_{j=1}^L\norm{W_j}_F^2}},\:F(C_\ell)\right\}}{\delta}\right)}{\gamma^2 N}} \right)\,,
    \end{equation}
    where
    \begin{align}
        \beta_{\rm PTM}& = \sqrt{d_\mathrm{out}^{(L)}}\sum_{j=1}^L\frac{1}{\sqrt{d_\mathrm{in}^{(j)}}}\prod_{\ell=j+1}^L \lVert W_\ell\rVert_{1,1} \,,
    \end{align}
    and 
    \begin{align}
        F(C_\ell)=\sqrt{d_\mathrm{out}^{(L)}}\sum_{j=1}^L\frac{1}{\sqrt{d_\mathrm{in}^{(j)}}}\prod_{\ell=j+1}^L C_\ell\,,
    \end{align}
    where $C_\ell$ is an upper bound on $\|W_\ell\|_{1,1}$.
\end{theorem}

We see that the upper bound expression in the PTM case is less compact than in the PM case. This is due to possibly different input and output dimensions at each layer. Consequently, the upper bounds for $\|W_\ell\|_{1,1}$ depend explicitly on the dimensions of layer $\ell$, preventing the use of a uniform upper bound across all layers, as was possible in the PM setting.

Moreover, since the layer dimensions vary, assuming a constant sparsity across all layers would constitute a rather crude approximation—again in contrast to the PM case, where such a simplification was more natural.
In a later example based on the QCNN architecture, where the number of qubits is halved at each layer, we will analyze how these upper bounds and the function $F(C_\ell)$ behave. As introduced in~\Cref{a:preliminaries-param-PTM}, the optimal choice of $C_\ell$ depends on the specific scenario under consideration. In general, one may take $C_\ell=\xi_\ell \sqrt{d_{\rm in}^{(\ell)}d_{\rm out}^{(\ell)}} $ or $C_\ell=\sqrt{\xi_\ell}\,d_{\rm in}^{(\ell)}$, with $\xi_\ell$ the sparsity of the $\ell$th layer. The first choice is preferable when $\xi_\ell \leq (d_{\rm in}^{(\ell)})^3/d_{\rm out}^{(\ell)}$ and the second otherwise. Additionally, if the layers correspond to unital channels, one may instead use $C_\ell=\xi_\ell \sqrt{d_{\rm out}^{(\ell)}/d_{\rm in}^{(\ell)}}$.

Because the expression for $\beta_{\rm PTM}$ is less transparent than in the PM case, we analyze its scaling in a regime of particular interest: circuits in which the number of qubits is halved at each layer and the channels are unital, so that the structural bound
\[
\|W_\ell\|_{1,1} \le \xi_\ell \sqrt{d_{\rm out}^{(\ell)} / d_{\rm in}^{(\ell)}}
\]
applies. These conditions arise naturally in architectures such as the QCNN \cite{cong2019quantum}.

Starting with $n_1 = n$ qubits in the first layer and imposing $n_{\ell+1} = n_\ell / 2$, the maximum depth is $L = \log_2 n$ (taking $n$ to be a power of two for simplicity). 
We fix a reference sparsity $\xi_L$ at the final layer $\ell = L$. Moving backward through the circuit, the Hilbert-space dimension grows exponentially, so we model the sparsity at layer $\ell$ as
\[
\xi_\ell = \xi_L \, 2^{\alpha (n_\ell - n_L)}= \xi_L \, 2^{\alpha n (2^{-(\ell-1)} - 2^{-(L-1)})}.
\]

This ansatz reflects the natural tendency of earlier (larger) layers to allow greater sparsity while introducing a single parameter $\alpha$ that controls the overall growth rate.

We now derive an upper bound on $\beta_{\rm PTM}$. Using the structural bound $\|W_\ell\|_{1,1} \leq \xi_\ell \sqrt{d_{\rm out}^{(\ell)} / d_{\rm in}^{(\ell)}}$ together with $n_\ell = n / 2^{\ell-1}$ and $n_{L+1} = n_L / 2$, we obtain
\begin{equation}\label{eq:beta_PTM_QCNN}
    \beta_{\rm PTM} \leq \sqrt{2^{n_{L+1}}} \sum_{i=1}^L \frac{1}{\sqrt{2^{n_i}}} \prod_{\ell=i+1}^L \xi_\ell \sqrt{\frac{2^{n_{\ell+1}}}{2^{n_\ell}}}.
\end{equation}
The product of the square-root factors simplifies to
\[
\prod_{\ell=i+1}^L \sqrt{\frac{2^{n_{\ell+1}}}{2^{n_l}}} = 2^{-\frac{1}{4} \sum_{\ell=i+1}^L n_\ell}, \qquad \sum_{\ell=i+1}^L n_\ell = 2n (2^{-i} - 2^{-L}),
\]
so that
\begin{equation}
    \beta_{\rm PTM} \leq 2^{\frac{3n}{2^L}} \sum_{i=1}^L 2^{-\frac{3}{2} n 2^{-i}} \prod_{\ell=i+1}^L \xi_\ell.
\end{equation}
Substituting the sparsity model yields
\[
\prod_{\ell=i+1}^L \xi_\ell = \xi_L^{L-i} \, 2^{-\frac{2\alpha n}{2^L}(L-2)} \, 2^{\frac{2\alpha n}{2^L} i} \, 2^{\frac{4\alpha n}{2^i}},
\]
and therefore
\begin{equation}
    \beta_{\rm PTM} \leq 2^{\frac{n}{2^L} (1 - 2\alpha (L-2))} \, \xi_L^L \, \sum_{i=1}^L 2^{\frac{(4\alpha - 3/2) n}{2^i}} \left( \frac{2^{\frac{2\alpha n}{2^L}}}{\xi_L} \right)^i.
\end{equation}
For the remainder of the discussion, we consider the regime $L = \log_2 n$, which implies $n / 2^L = 1$. The prefactor outside the sum is then $\mathcal{O}(n^{-2\alpha}) \cdot n^{\log_2 \xi_L}$. The behavior of $\beta_{\rm PTM}$ is consequently governed by the sum
\[
S = \sum_{i=1}^L 2^{\frac{(4\alpha - 3/2) n}{2^i}} \left( \frac{2^{\frac{2\alpha n}{2^L}}}{\xi_L} \right)^i.
\]

We now analyze the scaling of $S$ (and hence of $\beta_{\rm PTM}$) in two regimes: 

\begin{itemize}
    \item \textbf{Case $\boldsymbol{\alpha} \leq \boldsymbol{3/8}$}: Here $4\alpha - 3/2 \leq 0$, so $2^{(4\alpha - 3/2) n / 2^i} \leq 1$ for every $i$. Let $C = 2^{2\alpha n / 2^L} / \xi_L = 2^{2\alpha} / \xi_L$. Then
    \[
    S \leq \sum_{i=1}^L C^i = C \frac{C^L - 1}{C - 1}.
    \]
    In the regime of interest ($L = \log_2 n$) we have $C < 1$ whenever $\alpha < \frac{1}{2} \log_2 \xi_L$, in which case $S = \mathcal{O}(1)$. When $C \geq 1$ the sum is dominated by its last term, $S = \Theta(C^L) = \Theta(n^{2\alpha - \log_2 \xi_L})$. In both sub-cases we obtain
    \[
    \beta_{\rm PTM} \lesssim n^{-2\alpha + \log_2 \xi_L} \cdot \mathrm{poly}(n) \quad \text{or better},
    \]
    i.e., $\beta_{\rm PTM}$ is at most polynomial in $n$ (frequently $\mathcal{O}(1)$).

    \item \textbf{Case $\boldsymbol{\alpha} > \boldsymbol{3/8}$}: Now $4\alpha - 3/2 > 0$. The terms $2^{(4\alpha - 3/2) n / 2^i}$ become exponentially large for small $i$; the $i=1$ term alone is $2^{(4\alpha - 3/2) n / 2} = \exp(\Theta(n))$. We obtain
    \[
    S \leq 2^{\frac{(4\alpha - 3/2) n}{2}} C \frac{C^L - 1}{C - 1},
    \]
    so that $\beta_{\rm PTM} \lesssim \exp(\Theta(n)) \cdot \mathrm{poly}(n)$. This exponential growth in $n$ makes the PAC-Bayes complexity term vacuous for any polynomially large training set $N$.
\end{itemize}

The PTM of the first layer is a $2^n \times 2^{2n}$ matrix, hence at most $2^{3n}$ non-zero entries are possible: $\xi_1 \leq 2^{3n}$. From the model $\xi_1 = \xi_L \, 2^{\alpha (n - n_L)}$, we obtain
\[
\alpha \leq \frac{3n - \log_2 \xi_L}{n - 2}.
\]
For large $n$ this approaches $\alpha \leq 3$. The range $\alpha \leq 3/8$ is therefore significantly more restrictive than the maximum, but this bound assumes a maximally dense first-layer PTM, which is not necessarily the case in typical architectures.

In summary, the $\beta_{\rm PTM}$ exhibits controllable (at most polynomial) scaling with $n$ precisely when $\alpha \leq 3/8$. In this regime all exponential factors cancel or are suppressed, and $\beta_{\rm PTM}$ remains $\mathcal{O}(\mathrm{poly}(n))$ independently of the precise value of $\xi_L$. For $\alpha > 3/8$ an $\exp(\Theta(n))$ factor appears in $\beta_{\rm PTM}$, rendering the bound impractical. Consequently, the $\alpha$-region that guarantees controllable scaling of $\beta_{\rm PTM}$ (and therefore of the whole generalization bound) is
\begin{equation}\label{eq:betaPTM_regime}
    \alpha \leq \frac{3}{8}
\end{equation}
In this regime, just as oin the PM case for $\xi_{\max} = o(2^n)$, $\beta_{\rm PTM}$ remains at most polynomial in $n$. Therefore the logarithmic term in the bound has at most logarithmic dependence on $n$.
\\

\paragraph{Comparison to uniform bounds.}
Just as in the case of the PM, we can also use PTM perturbation bounds to derive uniform generalization bounds. When proving \Cref{a-lemma:pert_bound_PTM}, we have in particular shown that 
\begin{equation}
    \norm{\phi^{W}_{\mathrm{PTM}} - \phi^{W'}_{\mathrm{PTM}} }_{1\to 1}
    \leq \sqrt{\frac{d_{\mathrm{\mathrm{out}}}}{d_{\mathrm{in}}}}\norm{W - W'}_{1,1} \, .
\end{equation}
Following a derivation analogous to that in the PM case, we see that this perturbation bound leads to a metric entropy bound of
\begin{equation}
    \log \mathcal{N}(\{f_{\vec{w}}^{\mathrm{PTM}}\}_{\vec{w}}, \norm{\cdot}_\infty, \varepsilon)
    \leq \widetilde{\mathcal{O}}\left(  \sum_{j=1}^L 
    \begin{cases} 
        \xi_j \log\left(\frac{4^n L R_j \sqrt{d_{\mathrm{\mathrm{out}}}^{(j)}}}{\varepsilon \sqrt{d_{\mathrm{in}}^{(j)}}}\right)\quad &\textrm{if $\varepsilon \sqrt{d_{\mathrm{in}}^{(j)}}/2L\sqrt{d_{\mathrm{\mathrm{out}}}^{(j)}} < R_j$}\\
        0 &\textrm{else}
    \end{cases}\right) \, ,
\end{equation}
where, if we assume unitality (compare \cref{a-lemma:PTM_entry_bound}), the relevant radius is $R_j \leq \xi_j \sqrt{d_{\mathrm{\mathrm{out}}}^{(j)}/d_{\mathrm{in}}^{(j)}}$. Note that this bound on the radius is useful if $d_{\mathrm{\mathrm{out}}}^{(j)}<d_{\mathrm{in}}^{(j)}$, which is for instance the case in QCNNs.
Again, we can identify a regime in which summand $i$ contributes non-trivially, namely for $\varepsilon < 2L R_j \sqrt{d_{\mathrm{\mathrm{out}}}^{(j)} / d_{\mathrm{in}}^{(j)}}$.
Via the standard machinery, this yields the following uniform generalization bound: With high (constant) probability over the choice of a training data set of size $N$, we have
\begin{align}
    L_0(f_{\vec{w}}^{\mathrm{PTM}})
    &\leq \hat{L}_\gamma (f_{\vec{w}}^{\mathrm{PTM}}) + \mathcal{O}\left( \sqrt{\frac{\sum_{j=1}^L \xi_j \log\left(4^n L R_j\right)\cdot \min\{2LR_i\sqrt{d_{\mathrm{out}}^{(j)} / d_{\mathrm{in}}^{(j)}}, 1/2\}}{\gamma^2 N}}\right) \nonumber \\
    &\leq \hat{L}_\gamma (f_{\vec{w}}^{\mathrm{PTM}}) + \mathcal{O}\left( \sqrt{\frac{\sum_{j=1}^L \xi_j \log\left(4^n L \xi_j \sqrt{d_{\mathrm{out}}^{(j)} / d_{\mathrm{in}}^{(j)}}\right)\cdot \min\{2L\xi_j d_{\mathrm{out}}^{(j)} / d_{\mathrm{in}}^{(j)}, 1/2\}}{\gamma^2 N}}\right),
\end{align}
where $\vec{w}$ is the parameter vector after training. In what follows we focus exclusively on the main term of the complexity expression---i.e., the leading dependence on the sparsity parameters $\xi_i$ and the dimensionality factors---in exactly the same spirit as was done for the PM case.
Let us now compare this uniform generalization bound with the main complexity term of our PAC-Bayes generalization bound for the PTM formalism, which (compare \cref{th:pac-bayes-PTM}) is
\begin{equation}
    L_0(f_{\vec{w}}^{\mathrm{PTM}}) \leq \hat{L}_{\gamma}(f_{\vec{w}}^{\mathrm{PTM}})+ \widetilde{\mathcal{O}}\left(\sqrt{\frac{\beta_{\rm PTM}^2\:\xi_{\max} \ln\left(\sum_{j=1}^L2^{\xi_j}\right)\:\sum_{j=1}^L\norm{W_j}_F^2}{\gamma^2 N}} \right),
\end{equation}
where
\[
\beta_\mathrm{PTM}=\sqrt{d_\mathrm{out}^{(L)}}\sum_{j=1}^L\frac{1}{\sqrt{d_\mathrm{in}^{(j)}}}\prod_{\ell=j+1}^L \lVert W_\ell\rVert_{1,1}.
\]
To obtain a clean scaling analysis, we again consider the regime used in Eq.~\eqref{eq:beta_PTM_QCNN} for $\beta_{\rm PTM}$, in which the number of qubits is halved at each successive layer ($n_1=n$, $n_{\ell+1}=n_\ell/2$, $L=\log_2 n$) and the sparsity at layer $l$ is modeled as $\xi_\ell=\xi_L\,2^{\alpha(n_\ell-n_L)}$. As shown in \cref{eq:betaPTM_regime}, the PAC-Bayes complexity term exhibits controllable (at most polynomial) scaling with $n$ when $\alpha\leq 3/8$, with $\beta_{\rm PTM}=\mathcal{O}(\poly(n))$ in this regime; for $\alpha>3/8$, an $\exp(\Theta(n))$ factor renders the bound vacuous for polynomial $N$. We now compare the main complexity term of our PAC-Bayes bound with that of the uniform bound under the same halving-layer and sparsity-model assumptions. Our PAC-Bayes bound improves upon the uniform bound whenever
\begin{equation}
\beta_{\rm PTM}^2\:\xi_{\max} \ln\left(\sum_{j=1}^L2^{\xi_j}\right)\:\sum_{j=1}^L\norm{W_j}_F^2
\leq \mathrm{const}\cdot \sum_{j=1}^L \xi_j \log\left(4^n L \xi_j \sqrt{d_{\mathrm{out}}^{(j)} / d_{\mathrm{in}}^{(j)}}\right)\cdot \min\{2L\xi_j d_{\mathrm{out}}^{(j)} / d_{\mathrm{in}}^{(j)}, 1/2\} .
\end{equation}
Rearranging gives the sufficient condition
\begin{equation}
\beta_{\rm PTM}^2 \sum_{j=1}^L\norm{W_j}_F^2
\leq \mathrm{const}\cdot \frac{U}{\xi_{\max} \ln\left(\sum_{j=1}^L2^{\xi_j}\right)},
\end{equation}
where $U$ denotes the sum appearing on the right-hand side of the uniform bound (the ``uniform main term'').
Using the upper bound on $\beta_{\rm PTM}$ derived above and the structural Frobenius bound
\[
\sum_{j=1}^L\norm{W_j}_F^2\leq \sum_{j=1}^L\xi_j^2\frac{2^{n_{j+1}}}{2^{n_j}}=:F= \xi_L^2 2^{-\frac{4\alpha n}{2^L}}\sum_{j=1}^L 2^{\frac{n(4\alpha-1)}{2^j}},
\]
a pessimistic sufficient condition for superiority becomes
\[
F\leq \mathrm{const}\cdot\frac{U}{\beta_{\rm PTM}^2 \,\xi_{\max} \ln\left(\sum 2^{\xi_j}\right)}.
\]
In the controllable regime $\alpha\leq 3/8$ we have $\beta_{\rm PTM}\leq\mathcal{O}(\mathrm{poly}(n))$, $\xi_{\max}=\Theta(\xi_L 2^{\alpha n})$, and $\ln(\sum 2^{\xi_j})=\Theta(\xi_{\max})$. The condition therefore simplifies (up to polynomial factors in $n$) to
\[
F\cdot 2^{2\alpha n}\lesssim U.
\]
The scaling of $U = \sum_{j=1}^L \xi_j \log\left(4^n L R_j\right) \cdot \min\left\{2 L R_j \sqrt{d_{\mathrm{out}}^{(j)} / d_{\mathrm{in}}^{(j)}}, 1/2\right\}$ is derived layer-by-layer. For $\alpha \leq 3/8 < 1/2$, the min selects the left term: $U = 2 L \sum_{j=1}^L \xi_j^2 (d_{\mathrm{out}}^{(j)} / d_{\mathrm{in}}^{(j)}) \log\left(4^n L R_j\right)$. The log term is $\Theta(n)$; the sum is $\xi_L^2 2^{-4\alpha} \sum_j 2^{2^{L-j} (4\alpha - 2)}$. With $4\alpha - 2 < 0$, the sum is $\mathcal{O}(\log n)$, yielding $U = \mathcal{O}(n (\log n)^2)$.

\begin{itemize}
\item \textbf{Important regimes:}
\begin{itemize}
\item $  \alpha = 0  $ (rapid decay, $  \mathcal{O}(n \log n)  $ tight);
\item $  \alpha \approx 1/4  $ (moderate decay, $  \mathcal{O}(\log n)  $ layers contribute);
\item $  \alpha \approx 3/8  $ (faster decay, fewer layers).
\end{itemize}
We now estimate $  U  $ and $  F  $ layer-by-layer under the sparsity model (setting $  \xi_L=\Theta(1)  $ for the typical regime).
\item \textbf{Subcase $  \alpha\leq 1/4  $:} $  F=\mathcal{O}(L)=\mathcal{O}(\log n)  $ and $  U=\mathcal{O}(n(\log n)^2)  $ (the linear regime of the min dominates for almost all layers). The condition reduces to $  2^{2\alpha n}\lesssim\mathrm{poly}(n)  $. This holds for all large $  n  $ \emph{only if $  \alpha=0  $}. When $  \alpha=0  $ we recover $  \beta_{\rm PTM}=\mathcal{O}(n^{\log_2\xi_L})  $, and the condition holds whenever $  2^{\log_2\xi_L}<1  $, i.e., $  \xi_L<\sqrt{2}  $.
\item \textbf{Subcase $  1/4 < \alpha \leq 3/8  $:} Here $  4\alpha - 1 > 0  $, so $  F = \widetilde{\mathcal{O}}\left(2^{\frac{(4\alpha - 1)n}{2}}\right)  $ (dominated by early layers). For $  U  $, since $  4\alpha - 2 < 0  $, $  U = \mathcal{O}\left(n (\log n)^2\right)  $ (dominated by late layers). The condition simplifies to $  2^{n(4\alpha - 0.5)} \lesssim \poly(n)  $, which does not hold for large $  n  $ as the left side grows exponentially.
\end{itemize}
In summary, the PAC-Bayes bound improves upon the uniform bound in the regime $\alpha = 0$ (constant sparsity across layers) with $\xi_L < \sqrt{2}$. For $\alpha > 0$, even within $\alpha \leq 3/8$, exponential growth prevents the condition from holding asymptotically. 

In contrast to the PM case, where we identified regimes in which even a worst-case version of our PAC-Bayes bound offers a clear advantage over the uniform bound, our analysis for the PTM formalism does not reveal such favorable regimes under the considered sparsity and dimensionality assumptions. Ultimately, this outcome stems from the conservative nature of our derivations, which rely on worst-case scenario upper bounds for the parameter-dependent generalization error; our non-uniform PAC-Bayesian bound is most interesting when we take the parameter dependence seriously, which allows it to outperform parameter-independent uniform bounds.
\\

\paragraph{PAC-Bayes generalization bound for QCNNs.}
A widely studied architecture that precisely fits the regime analyzed above (qubit halving at each layer and unital quantum channels with $d_{\mathrm{in}} > d_{\mathrm{out}}$) is the quantum convolutional neural network (QCNN) introduced in \cite{cong2019quantum}. The general assumptions and sparsity model considered in the preceding sections already align with the structural properties of the QCNN. Here, however, we go further by specializing to its concrete architecture. The traditional QCNN building block is a two-qubit operation that consists of applying a two-qubit unitary and then tracing out the first qubit. In the most general scenario, one can consider an input state of dimension $d_{\mathrm{in}}$, applying a global unitary $U$ to the whole state, and then tracing out some of the qubits, leaving the untraced state of dimension $d_{\mathrm{out}}$. We can consider tracing out a subsystem of dimension $d_{\mathrm{env}}$ such that $d_{\mathrm{in}}=d_{\mathrm{out}}\cdot d_{\mathrm{env}}$. The quantum channel of this block on some input quantum state $\rho$ is described as
\begin{equation}
    \phi(\rho)=\sum_\ell (\langle \ell|\otimes\mathbb{I})\,U\,\rho\,U^\dagger\,(|\ell\rangle\otimes\mathbb{I}),
\end{equation}
which clearly shows that the Kraus operators are $K_\ell=(\langle \ell|\otimes\mathbb{I})U$ for $\ell=0,\dots,d_{\mathrm{env}}-1$. The Frobenius norm of the Pauli transfer matrix (PTM) associated with a quantum channel $\phi$ can be expressed in terms of its Kraus operators as
\begin{equation}
    \begin{split}
        \|R^\phi\|_F^2
        &= \tr(R^\phi {R^\phi}^\dagger) \\
        &= \frac{1}{d_{\mathrm{out}}\,d_{\mathrm{in}}}
        \sum_{A,B}
        \tr\!\bigl(\sigma_A\,\phi(\sigma_B)\bigr)\,
        \tr\!\bigl(\sigma_B\,\phi^\dagger(\sigma_A)\bigr) \\
        &= \frac{1}{d_{\mathrm{out}}\,d_{\mathrm{in}}}
        \sum_{A,B}
        \left(\sum_i \tr(\sigma_A K_i \sigma_B K_i^\dagger)\right)
        \left(\sum_j \tr(\sigma_B K_j^\dagger \sigma_A K_j)\right) \\
        &= \frac{1}{d_{\mathrm{in}}}
        \sum_{B}\sum_{i,j}
        d_{\mathrm{out}}\,
        \tr\!\left(
            K_i \sigma_B K_i^\dagger
            K_j \sigma_B K_j^\dagger
        \right) \\
        &= \sum_{i,j} \left|\tr(K_i^\dagger K_j)\right|^2 .
    \end{split}
\end{equation}
Here ${R^\phi}^\dagger$ denotes the Pauli transfer matrix of the adjoint channel $\phi^\dagger$, defined by $\tr[X\,\phi(Y)] = \tr[\phi^\dagger(X)\,Y]$. In passing from the third to the fourth line 
we have used the completeness of the Pauli basis on the output space,
\begin{equation}
\sum_A \tr(\sigma_A X)\tr(\sigma_A Y) = d_{\mathrm{out}}\,\tr(XY),
\end{equation}
while in the last line we have used the corresponding completeness relation on the input space,
\begin{equation}
\sum_B \tr(X\sigma_B Y\sigma_B) = d_{\mathrm{in}}\,\tr(X)\tr(Y).
\end{equation}
We can then compute
\begin{equation}
    \begin{split}
        \|R^\phi\|_F^2&=\sum_{i,j}\left|\tr\left(U^\dagger(|j\rangle\otimes\mathbb{I})(\langle i|\otimes\mathbb{I})\,U\right)\right|^2\\
        &=\sum_{i,j}\left|\tr\left(\langle i|j\rangle\otimes \mathbb{I}\right)\right|^2\\
        &=\sum_i d_{\mathrm{out}}^2\\
        &=d_{\mathrm{in}}\,d_{\mathrm{out}},
    \end{split}
\end{equation}
where we have used that that the index $i$ runs over
$d_{\mathrm{env}} = d_{\mathrm{in}}/d_{\mathrm{out}}$ values. Using the definition of $W$ one can easily show that
\begin{equation}
    \|W\|_F^2=\|R^\phi\|_F^2-\frac{d_{\mathrm{in}}}{d_{\mathrm{out}}}.
\end{equation}
We have thus shown that for a QCNN with only building blocks as described above, the Frobenius norm of $W$ does not depend on the learned parameters. The only dependency on the parameters will be hidden inside $\beta$.
This observation is consistent with the fact that this QCNN construction, based on a global unitary followed by a partial trace, cannot implement a maximally depolarizing channel for any choice of parameters. Although the partial trace reduces the effective system dimension and may decrease the magnitude of individual entries of the PTM $W$, the fixed value of the squared Frobenius norm, $\|W\|_F^2 = d_{\mathrm{in}} d_{\mathrm{out}} - d_{\mathrm{in}}/d_{\mathrm{out}}$, imposes a fundamental lower bound on other entry-wise norms. In particular, for any finite-dimensional matrix $W$, the general inequality
\begin{equation}
    \|W\|_{1,1} \ge \|W\|_F
\end{equation}
holds. As a consequence, the entry-wise $\ell_1$ norm of $W$ can never vanish in this architecture, independently of the choice of the unitary $U$. Therefore, this class of QCNNs cannot realize a completely depolarizing channel, but can at most approximate it in a limited way by redistributing the total weight of the PTM.

\subsection{PAC-Bayesian generalization bound for equivariant quantum models}
Finally, we bring our attention to the derivation of the PAC-Bayesian generalization bound for equivariant quantum models. We address the general case where the number of equivariant parameters, denoted by $\Xi_j = \sum_{\lambda} m_{j,\lambda}^2$, varies across layers $j=1, \dots, L$. We proceed by establishing a high-probability tail bound for the perturbations in the equivariant basis, determining the necessary noise scaling to satisfy the margin condition, and finally computing the generalization bound.

\subsubsection{Tail bound for Gaussian perturbation}

We assume the perturbation matrices $U_j=\{U_{j,\lambda}\}_{\lambda}$ consist of entries drawn from $\mathcal{N}(0, \sigma^2)$. Our goal is to bound the equivariant 1-norm, $\|U_j\|_{eq,1} = \sum_{\lambda} d_{\lambda} \|U_{j,\lambda}\|_1$, simultaneously for all layers.

First, we utilize the relationship between the equivariant 1-norm and the standard $l_1$ norm of the vectorized parameters. With $d_{max} = \max_{\lambda} d_{\lambda}$, we have $\|U_j\|_{eq,1} \le d_{max} \|u_j\|_1$, where $u_j$ is the vector of all parameters in layer $j$. Using the tail bound for the $l_1$ norm of a Gaussian vector with $\Xi_j$ parameters, we obtain the bound for a single layer:
\begin{equation}
    P(\|U_{j}\|_{eq,1} \ge t) \le 2^{\Xi_{j}} e^{\frac{-t^{2}}{2d_{max}^{2}\Xi_{j}\sigma^{2}}}
\end{equation}
To ensure the bound holds for all layers, we apply a union bound:
\begin{equation}
    P(\exists j : \|U_{j}\|_{eq,1} > t) \le \sum_{j=1}^{L} 2^{\Xi_{j}} e^{\frac{-t^{2}}{2d_{max}^{2}\Xi_{j}\sigma^{2}}}
\end{equation}
To simplify this sum while ensuring the bound holds for the worst-case layer (the one with the slowest exponential decay), we bound the decay term using $\Xi_{max} = \max_j \Xi_j$:
\begin{equation}
    e^{\frac{-t^{2}}{2d_{max}^{2}\Xi_{j}\sigma^{2}}} \le e^{\frac{-t^{2}}{2d_{max}^{2}\Xi_{max}\sigma^{2}}}
\end{equation}
Substituting this into the union bound yields:
\begin{equation}
    P(\exists j : \|U_{j}\|_{eq,1} > t) \le e^{\frac{-t^{2}}{2d_{max}^{2}\Xi_{max}\sigma^{2}}} \sum_{j=1}^{L} 2^{\Xi_{j}}
\end{equation}
We require the failure probability to be at most $1/2$. Setting the right-hand side to $1/2$ and solving for $t$, we obtain the high-probability upper bound on the perturbation norms:
\begin{equation}
    t = \sigma d_{max} \sqrt{2\Xi_{max} \ln\left(2\sum_{j=1}^{L} 2^{\Xi_{j}}\right)}\,.
\end{equation}
Thus, with probability at least $1/2$, $\|U_j\|_{eq,1} \le t$ holds for all layers simultaneously.

\subsubsection{Margin condition and noise scaling}

We now connect this tail bound to the output stability of the model. Using the equivariant perturbation bound (Lemma 5) and the derived tail bound $t$, the change in the model output is bounded by:
\begin{equation}
    \|f_{\vec{w}+\vec{u}}(\vec{x}) - f_{\vec{w}}(\vec{x})\|_{\infty} \le e \beta_{\rm eq} \max_j \|U_j\|_{eq,1} \le e \beta_{\rm eq} \sigma d_{max} \sqrt{2\Xi_{max} \ln\left(2\sum_{j=1}^{L} 2^{\Xi_{j}}\right)}
\end{equation}
Here, the complexity term is defined as $\beta_{\rm eq} = \sum_{i=1}^{L} \prod_{l=i+1}^{L} \|W_{l}\|_{eq,1}$. Accounting for the discretization of $\beta_{\rm eq}$ via $\tilde{\beta}_{\rm eq}$ (where $|\beta_{\rm eq} - \tilde{\beta}_{\rm eq}| \le \beta_{\rm eq}/L$) and enforcing the PAC-Bayes margin condition (perturbation $\le \gamma/4$), we constrain the noise variance $\sigma$:
\begin{equation}
    \sigma \le \frac{\gamma(1-1/L)}{4e d_{max} \tilde{\beta} \sqrt{2\Xi_{max} \ln\left(2\sum_{j=1}^{L} 2^{\Xi_{j}}\right)}}
\end{equation}

\subsubsection{KL divergence}

As in the PM and PTM cases, but now using the equivariant Frobenius norm, the KL divergence for our Gaussian prior and posterior is given by:
\begin{equation}
    \textrm{KL}(\vec{w}+\vec{u} || P) = \frac{1}{2\sigma^{2}} \sum_{j=1}^{L} \|W_{j}\|_{eq,F}^{2}
\end{equation}
Substituting the upper bound for $\sigma$ derived above, we obtain the bound on the KL term:
\begin{equation}
    \textrm{KL}(\vec{w}+\vec{u} || P) \le \frac{16e^{2} d_{max}^{2} \tilde{\beta}_{\rm eq}^{2} \Xi_{max} \ln\left(2\sum_{j=1}^{L} 2^{\Xi_{j}}\right)}{\gamma^{2}(1-1/L)^{2}} \sum_{j=1}^{L} \|W_{j}\|_{eq,F}^{2}
\end{equation}

Finally,  we can apply \cref{lemma:pac-bayes_margin} and get that for any $\tilde{\beta}_{\rm eq}$ with probability at least $1-\delta$ and for all $\boldsymbol{w}$ such that $|\beta_{\rm eq}-\tilde{\beta}_{\rm eq}|\leq \beta_{\rm eq}/L$, we have
\begin{equation}\label{eq:before_union_equivariant}
    L_0(f_{\vec{w}}^{\mathrm{eq}}) \leq \hat{L}_{\gamma}(f_{\vec{w}}^{\mathrm{eq}})+ \mathcal{O}\left(\sqrt{\frac{\beta_{\rm eq}^2\:\xi_{\max} \ln\left(\sum_{j=1}^L2^{\xi_j}\right)\:\sum_{j=1}^L\norm{W_j}_F^2+\ln\left(\frac{N}{\delta}\right)}{\gamma^2 N}} \right)\,.
\end{equation}

\subsubsection{Equivariant generalization bound}
As discussed previously, the covering net construction and the corresponding union bound follow the same reasoning as in the PM and PTM cases. In particular, the final bound is obtained by dividing $\delta$ by the covering net size. It therefore remains to determine the minimum and maximum values of $\beta_{\rm eq}$ in order to characterize this covering. We now establish upper and lower bounds for $\beta_{\rm eq}$ in the following lemma.

\begin{lemma}[Upper and lower bounds on $\beta_{\mathrm{eq}}$]
Assume a uniform parameter count $\Xi_l = \Xi$ and a uniform spectral bound $\eta_l = \eta$ across all layers. The complexity term $\beta_{\mathrm{eq}}$ satisfies
\begin{equation} \label{eq:betaeq_max}
    1\leq \beta_{\mathrm{eq}} \leq M(\eta,|G|,\Xi,L)\,,
\end{equation}
where 
\begin{equation}
    M(\eta,|G|,\Xi,L)=\beta_{\mathrm{eq}} \leq \frac{(\eta \sqrt{|G|\Xi})^L - 1}{\eta \sqrt{|G|\Xi} - 1}.
\end{equation}
\end{lemma}

\begin{proof}
By definition, $\beta_{\mathrm{eq}}$ is given by
\begin{equation}
    \beta_{\mathrm{eq}} = \sum_{i=1}^L \prod_{l=i+1}^{L} \norm{W_{l}}_{\mathrm{eq},1}.
\end{equation}
The lower bound follows in the same way as in \cref{a-lemma:beta_bounds-PM}. For the upper bound, applying the uniform upper bound $\norm{W_{l}}_{\mathrm{eq},1} \leq \eta \sqrt{|G| \Xi}$, we obtain a geometric series:
\begin{equation}
    \beta_{\mathrm{eq}} \leq \sum_{i=1}^L \prod_{l=i+1}^{L} \left( \eta \sqrt{|G| \Xi} \right) = \sum_{i=1}^L \left( \eta \sqrt{|G| \Xi} \right)^{L-i}.
\end{equation}
Re-indexing the sum by letting $j = L-i$, we have:
\begin{equation}
    \beta_{\mathrm{eq}} \leq \sum_{j=0}^{L-1} \left( \eta \sqrt{|G| \Xi} \right)^j = \frac{(\eta \sqrt{|G|\Xi})^L - 1}{\eta \sqrt{|G|\Xi} - 1},
\end{equation}
yielding the final result.
\end{proof}
Thus, the generalization error bound for equivariant quantum models can be stated in the following theorem.

\begin{theorem}[PAC-Bayes generalization bound for equivariant quantum models (restatement of \cref{th:pac-bayes-equivariant})]
For any depth $L > 0$, let $f_{\vec{w}}^{\mathrm{eq}}: \mathcal{X} \to \mathbb{R}^K$ be an $L$-layer equivariant quantum model parameterized by $w = \{W_{j,\lambda}\}_{j,\lambda}$. Let $\Xi_j = \sum_\lambda m_{j,\lambda}^2$ be the number of equivariant parameters per layer, where $\Xi_{\max} = \max_j \Xi_j$, and let $d_{\max} = \max_\lambda d_\lambda$ be the maximum dimension of the irreducible representations.
Then, for any margin $\gamma > 0$, with probability at least $1-\delta$ over the training set of size $N$, for any learned parameters $\vec{w}=\mathrm{vec}(\{W_j\}_{j=1}^L)$, we have
\begin{equation}
    L_0(f_{\vec{w}}^{\mathrm{eq}}) \leq \hat{L}_{\gamma}(f_{\vec{w}}^{\mathrm{eq}}) + \mathcal{O}\left(\sqrt{\frac{d_{\max}^2 \, \Xi_{\max} \, \ln(\sum^L_{j=1} 2^{\Xi_j}) \, \beta_{\mathrm{eq}}^2 \sum_{j=1}^L \|W_j\|_{\mathrm{eq},F}^2+\ln\left(\frac{N \;L\min\left\{\frac{\gamma \sqrt{N}}{\sqrt{\sum_{j=1}^L\norm{W_j}_{\mathrm{eq},F}^2}},\:M(\eta,|G|,\Xi)\right\}}{\delta}\right)}{\gamma^2 N}} \right)\,,
\end{equation}
where $\beta_{\mathrm{eq}}$ is defined as
\begin{equation}
    \beta_{\mathrm{eq}} = \sum_{j=1}^L \left( \prod_{\ell=j+1}^L \|W_l\|_{\mathrm{eq},1} \right)\,,
\end{equation}
the layer-wise equivariant norms 
are defined in terms of the irrep blocks $W_{j,\lambda}$ as
\begin{equation}
    \|W_j\|_{\mathrm{eq},1} = \sum_{\lambda} d_\lambda \|W_{j,\lambda}\|_1 \quad \text{and} \quad \|W_j\|_{\mathrm{eq},F}^2 = \sum_{\lambda} d_\lambda \|W_{j,\lambda}\|_F^2 \,,
\end{equation}
and 
\begin{equation}
    M(\eta,|G|,\Xi,L)=\frac{(\eta \sqrt{|G|\Xi})^L - 1}{\eta \sqrt{|G|\Xi} - 1}.
\end{equation}
\end{theorem}

\section{Extending to data re-uploading quantum models}\label{a:data-re-uploading}

The results in \Cref{a:perturbation-bounds,a:generalization-bounds} are phrased for encoding-first quantum models, in which the input vector $\vec{x}$ is initially encoded into a state $\rho(\vec{x})$, but then all consecutive layers are independent of $\vec{x}$. In this appendix, we argue that the results straightforwardly extend to data re-uploading quantum models~\cite{perez2020data}, which alternate between input-dependent and parameter-dependent layers. 

\begin{corollary}[PAC-Bayes generalization bounds for data re-uploading quantum models (restatement of \Cref{cor:reuploading}]\label{a-cor:reuploading} 
The PAC-Bayes generalization bounds established in \Cref{th:pac-bayes-PM} and \Cref{th:pac-bayes-PTM}  also apply to data re-uploading quantum models of the form
\begin{equation}
   \rho_{\mathrm{\mathrm{out}}}(\vec{x}) = \left(\phi_L \circ \mathcal{E}_L^{\vec{x}} \circ \phi_{L-1} \circ \mathcal{E}_{L-1}^{\vec{x}} \circ \cdots \circ \phi_1 \circ \mathcal{E}_1^{\vec{x}}\right)(\rho_0)\,,
\end{equation}
where each $\phi_j=\phi_{\mathrm{P(T)M}}^{W_j}$ is a parameter-dependent quantum channel, and each $\mathcal{E}_j^{\vec{x}}$ is an input-dependent quantum channel implementing the data re-uploading step.
\end{corollary}

\begin{proof}
    The idea of the proof is to refer back to the arguments used in~\Cref{a-lemma:pert_bound_PM,a-lemma:pert_bound_PTM}, and show that the perturbation bounds remain invariant even when quantum data-encoding channels are inserted at each layer.
    In both the PM and PTM settings, the proof relies on the recursive definitions of the intermediate states. Since the subsequent analysis is identical in the two cases, we present a unified notation by writing $\mathrm{P(T)M}$ to denote either model. We thus redefine the states to incorporate the encoding channels
    \begin{align}
        \rho_{\leq j}^{W+U}(\vec{x})&\coloneqq \left(\phi^{W_j+U_j}_{\mathrm{P(T)M}}\circ\mathcal{E}_j^{\vec{x}}\circ\dots\circ\phi^{W_1+U_1}_{\mathrm{P(T)M}}\circ \mathcal{E}_1^{\vec{x}}\right)\bigl(\rho_0\bigr), \\
    \rho_{\leq j}^{W}(\vec{x})&
    \coloneqq 
    \left(\phi^{W_j}_{\mathrm{P(T)M}}\circ\mathcal{E}_j^{\vec{x}}\circ\dots\circ\phi^{W_1}_{\mathrm{P(T)M}}\circ \mathcal{E}_1^{\vec{x}}\right)\bigl(\rho_0\bigr).
    \end{align}
    We define
    \begin{equation}
        \Delta_j=\norm{\rho_{\leq j}^{W+U}(\vec{x})-\rho_{\leq j}^{W}(\vec{x})}_1.
    \end{equation}
    A similar inductive argument applies, using
    \begin{equation}
        \begin{split}
            \Delta_{j+1}&=\norm{\left(\phi^{W_{j+1}+U_{j+1}}_{\mathrm{P(T)M}}\circ\mathcal{E}_{j+1}^{\vec{x}}\right)\left(\rho_{\leq j}^{W+U}(\vec{x})\right)-\left(\phi^{W_{j+1}}_{\mathrm{P(T)M}}\circ\mathcal{E}_{j+1}^{\vec{x}}\right)\left(\rho_{\leq j}^{W}(\vec{x})\right)}_1\\
            &=\norm{\left(\phi^{W_{j+1}+U_{j+1}}_{\mathrm{P(T)M}}\circ\mathcal{E}_{j+1}^{\vec{x}}\right)\left(\rho_{\leq j}^{W+U}(\vec{x})-\rho_{\leq j}^{W}(\vec{x})\right)+\left((\phi^{W_{j+1}+U_{j+1}}_{\mathrm{P(T)M}}-\phi^{W_{j+1}}_{\mathrm{P(T)M}})\circ\mathcal{E}_{j+1}^{\vec{x}}\right)\left(\rho_{\leq j}^{W}(\vec{x})\right)}_1\\
            &\leq CC_{1,1}\left(\phi^{W_{j+1}+U_{j+1}}_{\mathrm{P(T)M}}\right)\norm{\mathcal{E}_{j+1}^{\vec{x}}\left(\rho_{\leq j}^{W+U}(\vec{x})-\rho_{\leq j}^{W}(\vec{x})\right)}_1+\norm{\phi^{W_{j+1}+U_{j+1}}_{\mathrm{P(T)M}}-\phi^{W_{j+1}}_{\mathrm{P(T)M}}}_{r,1}\norm{\mathcal{E}_{j+1}^{\vec{x}}\left(\rho_{\leq j}^{W}(\vec{x})\right)}_r
        \end{split}
    \end{equation}
    Since $\mathcal{E}_{j+1}^{\vec{x}}$ is a quantum channel, it is contractive in the $1$-norm. Using sub-multiplicativity, we obtain
    \begin{equation}
        \Delta_{j+1}\leq CC_{1,1}\left(\phi^{W_{j+1}+U_{j+1}}_{\mathrm{P(T)M}}\right)\underbrace{\norm{\rho_{\leq j}^{W+U}(\vec{x})-\rho_{\leq j}^{W}(\vec{x})}_1}_{=\Delta_j}+\norm{\phi^{W_{j+1}+U_{j+1}}_{\mathrm{P(T)M}}-\phi^{W_{j+1}}_{\mathrm{P(T)M}}}_{r,1}\norm{\mathcal{E}_{j+1}^{\vec{x}}}_{r,r}\norm{\rho_{\leq j}^{W}(\vec{x})}_r
    \end{equation}
    Finally, by setting $r=1$, we obtain $\norm{\mathcal{E}_{j+1}^{\vec{x}}}_{r,r}=1$ for both the PM and PTM representations. Consequently, from this stage onwards the argument is identical in both cases and coincides with the standard proof of the perturbation bounds.
\end{proof}
While the analysis is carried out for the PM and PTM representations, the underlying argument remains unchanged and extends directly to the equivariant setting.

\section{PAC-Bayes bounds from the Kraus and Stinespring representations}\label{a:kraus}

We have described three parameterizations of quantum models and corresponding perturbation bounds that lead to meaningful PAC-Bayesian generalization bounds. The attentive reader may have wondered why we do not use more standard mathematical representation of quantum channels, such as the Kraus or Stinespring representations. In this appendix, we demonstrate that applying our recipe to the Kraus and Stinespring representations does not lead to useful generalization bounds. This highlights the importance of a well-chosen parameterization, not from the perspective of actually implementing a quantum model, but for its mathematical analysis.

\paragraph{Kraus parameterization.}
A natural way of parameterizing a quantum channel $\phi_j$ acting at layer $j$ is via its Kraus representation,
\begin{equation}
\phi_j(\rho)=\sum_{i=1}^{r_j} K_i^{(j)} \rho K_i^{(j)\dagger},
\end{equation}
where $r_j$ denotes the Kraus rank, which can be taken to satisfy $r_j\leq d^{(j)}_{\mathrm{in}} d^{(j)}_{\mathrm{out}}$. For the sake of the analysis, we assume $r_j=d^2$, where $d=2^n$, for all layers.

We model perturbations at layer $j$ by assuming that each Kraus operator is additively perturbed. Concretely, the perturbed channel is defined as
\begin{equation}
\tilde{\phi}_j(\rho)
=
\sum_{i=1}^{r_j}
\bigl(K_i^{(j)} + \tilde K_i^{(j)}\bigr)
\, \rho \,
\bigl(K_i^{(j)} + \tilde K_i^{(j)}\bigr)^\dagger,
\end{equation}
where $\tilde K_i^{(j)}$ denotes a perturbation of the $i$-th Kraus operator at layer $j$.
Applying the telescoping argument used throughout this work, one can derive the following perturbation bound
\begin{equation}
\bigl\|\tilde{\phi}_L(\dots(\tilde{\phi}_1(\rho))) - \phi_L(\dots(\phi_1(\rho)))\bigr\|_1
\;\leq\;
\sum_{j=1}^L\sum_{i=1}^{r_j}
3\,\|K_i^{(j)}\|_\infty\,\|\tilde K_i^{(j)}\|_\infty .
\end{equation}

To place this within our PAC-Bayesian framework, we collect all Kraus operators into a parameter vector
$\vec w=\mathrm{vec}(\{K_i^{(j)}\}_{i,j})$
and endow it with the Frobenius norm
\begin{equation}
\|\vec w\|_F^2=\sum_{j=1}^L\sum_{i=1}^{r_j}\|K_i^{(j)}\|_F^2.
\end{equation}
However, the complete positivity and trace-preservation constraints imply
$\sum_i K_i^{(j)\dagger}K_i^{(j)}=\mathbb{I}$, from which it follows that
$\|\vec W\|_F^2=\sum_{j=1}^L\mathrm{tr}(\mathbb{I})=Ld$,
independently of the actual learned parameters.
As a consequence, the resulting generalization bound becomes essentially parameter-blind.
In particular, one obtains a bound of the form
\begin{equation}
L_0(f_W) \leq \hat{L}_\gamma(f_W) + \widetilde{\mathcal{O}}\Bigg(
\sqrt{\frac{L d^2 \ln(L d^3) \Big(\sum_{j=1}^L\sum_{i=1}^{d^2} \|K_i^{(j)}\|_\infty\Big)^2}{\gamma^2 N}}
\Bigg).
\end{equation}
Noting that the Kraus operator norms satisfy $1 \le \sum_{i=1}^{d^2} \|K_i\|_\infty \le d$, 
even in the best possible case $\sum_{i=1}^{d^2} \|K_i\|_\infty = 1$, 
this bound is at best comparable to  uniform generalization error bounds in the literature~\cite{caro2022generalization}.

\paragraph{Stinespring parameterization.}
An alternative standard parameterization of quantum channels is given by the Stinespring dilation. 
Each layer $\phi_j$ can be represented via an isometry
\begin{equation}
V^{(j)}=\sum_{i=1}^{r_j} K_i^{(j)}\otimes\lvert i\rangle
\end{equation}
acting on an enlarged Hilbert space, such that
\begin{equation}
\phi_j(\rho)
=
\mathrm{tr}_{\mathrm{env}}\!\left(V^{(j)} \rho V^{(j)\dagger}\right).
\end{equation}
We consider again $r_j=d^2$ for all layers.

We model perturbations at layer $j$ by assuming that the Stinespring isometry is additively perturbed. Concretely, the perturbed channel is defined as
\begin{equation}
\tilde{\phi}_j(\rho)
=
\mathrm{tr}_{\mathrm{env}}\!\left(
\bigl(V^{(j)}+\tilde V^{(j)}\bigr)
\, \rho \,
\bigl(V^{(j)}+\tilde V^{(j)}\bigr)^\dagger
\right),
\end{equation}
where $\tilde V^{(j)}$ denotes a perturbation of the Stinespring isometry at layer $j$.

Assuming $\|\tilde V^{(j)}\|_\infty \le 1$, an argument analogous to the Kraus case yields the perturbation bound
\begin{equation}
\bigl\|\tilde{\phi}_L(\dots(\tilde{\phi}_1(\rho))) - \phi_L(\dots(\phi_1(\rho)))\bigr\|_1
\;\leq\;
3\sum_{j=1}^L \|\tilde V^{(j)}\|_\infty .
\end{equation}

In this case, the parameter vector is
$\vec w=\mathrm{vec}(\{V^{(j)}\}_{j=1}^L)$
and when computing its Frobenius norm, one again finds that
\begin{equation}
\|\vec w\|_F^2=\sum_{j=1}^L\|V^{(j)}\|_F^2=Ld,
\end{equation}
independently of the specific channel.
Consequently, the induced PAC-Bayesian generalization bound is entirely uniform, taking the form
\begin{equation}
L_0(f_W) \leq \hat{L}_\gamma(f_W) + \widetilde{\mathcal{O}}\Bigg(
\sqrt{\frac{L^3 d^2 \ln(L d^3)}{\gamma^2 N}}
\Bigg),
\end{equation}
scaling only with the depth and system dimension and not with any learned structure of the model.
This mirrors the situation encountered with the Kraus parameterization and further illustrates that, while the Stinespring representation is mathematically elegant and physically meaningful, it might be ill-suited for deriving parameter-dependent generalization bounds.

\section{Alternative convention for the Pauli Transfer Matrix and its impact on the PAC-Bayes bound}

The choice of convention for the Pauli Transfer Matrix (PTM) directly affects the weight-norm and channel-norm contributions appearing in the PAC-Bayes generalization bound of~\Cref{th:pac-bayes-PTM}. In this section we contrast the convention used in the main text with a closely related alternative and report a numerically intriguing—though not yet theoretically understood—effect that appears when combining elements of both.

In the main text the PTM is defined such that the channel acts on a density operator $\rho$ as
\begin{align}
    \phi_{\mathrm{PTM}}^{W}(\rho) 
    = \frac{1}{\sqrt{d_{\mathrm{in}} d_{\mathrm{out}}}} 
    \sum_{A,B} R(A,B)\, \mathrm{Tr}(\sigma_B \rho)\, \sigma_A ,
\end{align}
with coefficients
\begin{align}
    R(A,B) = \delta_{A,0}\delta_{B,0}\sqrt{\frac{d_{\mathrm{in}}}{d_{\mathrm{out}}}} + W(A,B), 
    \qquad
    R(A,B) = \frac{1}{\sqrt{d_{\mathrm{in}} d_{\mathrm{out}}}}
    \mathrm{Tr}\!\left( \sigma_A\, \phi_{\mathrm{PTM}}^{W}(\sigma_B) \right).
\end{align}
Under this convention the Frobenius norm of the weight matrix satisfies
\begin{align}
    \|W\|_F^2 = \mathrm{Tr}(J_\phi^2) - \frac{d_{\mathrm{in}}}{d_{\mathrm{out}}}.
\end{align}

An alternative, equally valid convention absorbs the dimensional prefactor into the matrix entries:
\begin{align}
    \phi_{\mathrm{PTM}}^{W}(\rho) 
    = \sum_{A,B} \tilde{R}(A,B)\, \mathrm{Tr}(\sigma_B \rho)\, \sigma_A ,
\end{align}
with
\begin{align}
    \tilde{R}(A,B) 
    = \delta_{A,0}\delta_{B,0}\frac{1}{d_{\mathrm{out}}} + \tilde{W}(A,B),
    \qquad
    \tilde{R}(A,B) 
    = \frac{1}{d_{\mathrm{in}} d_{\mathrm{out}}}
    \mathrm{Tr}\!\left( \sigma_A\, \phi_{\mathrm{PTM}}^{W}(\sigma_B) \right).
\end{align}
Within this convention the squared Frobenius norm becomes
\begin{align}
    \|\tilde{W}\|_F^2 
    = \frac{\mathrm{Tr}(J_\phi^2)}{d_{\mathrm{in}} d_{\mathrm{out}}} 
    - \frac{1}{d_{\mathrm{out}}^2}.
\end{align}
When $d_{\mathrm{in}} = d_{\mathrm{out}} = 4^n$, this convention yields a Frobenius norm equivalent to that obtained in the PM representation. However, the dimensional factors that appear dividing the Frobenius norm in this alternative convention reappear multiplicatively in the factor $\beta_{\mathrm{PTM}}$, which constitutes the other non-uniform contribution in the generalization error bound.

While the alternative convention therefore improves the Frobenius-norm contribution to the bound, this advantage is offset by the corresponding increase in the $\beta_{\mathrm{PTM}}$ term, where the dimensional factors reappear. This observation motivated us to explore a hypothetical ``best-of-both-worlds'' scenario: we retain the $\beta_{\mathrm{PTM}}$ factor from the main-text convention while recomputing the Frobenius norms entering the KL-divergence term using the rescaled weights $\tilde{W}$ associated with the alternative convention.

Although we have no proof of a generalization bound that uses only the more favorable contributions from the two different normalization conventions, each convention is individually physically valid. Figure~\ref{fig:hybrid} illustrates the outcome of this hybrid evaluation for our QCNN experiment. The plotted density corresponds to the same models reported in the main text, but with each $\|W_j\|_F^2$ recomputed using the alternative normalization implied by $\tilde{R}$. As can be seen, this modification leads to a substantially smaller value of the complexity term while preserving essentially the same degree of correlation with the observed generalization error, namely $r = 0.46$.

\begin{figure}[h]
    \centering    \includegraphics[width=0.6\textwidth]{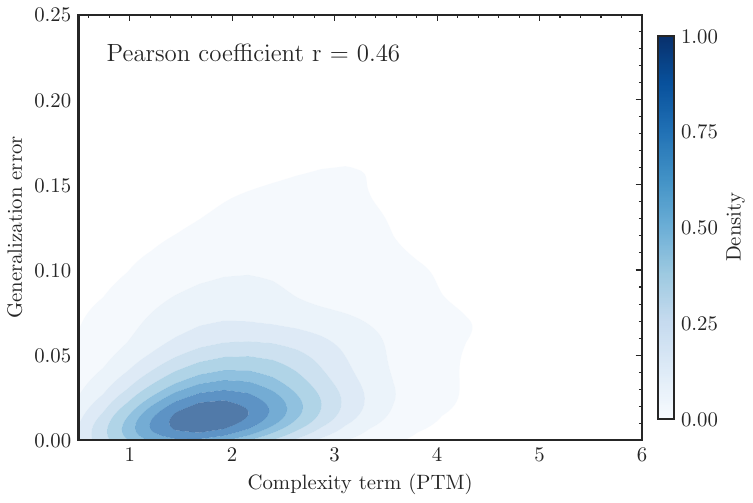}
    \caption{\textbf{Hybrid evaluation of the PAC-Bayes bound on the QCNN experiment.} Density corresponds to the same models as in \Cref{fig:numerics_results}(b) in the main text, but with each $\|W_j\|_F^2$ replaced by the dimensionally rescaled value suggested by the alternative convention $\|\tilde{W}_j\|_F^2$, while $\beta_{\mathrm{PTM}}$ is kept unchanged. The complexity term becomes significantly smaller, yet the correlation with the observed generalization error remains $r = 0.46$.}
    \label{fig:hybrid}
\end{figure}

At present we do not have a theoretical explanation for why this particular combination preserves the empirical correlation while substantially tightening the complexity term. It remains unclear whether a principled reweighting or renormalization of this kind can be formally justified, or whether the observed improvement is specific to the architecture, dataset, or training procedure considered here. We therefore report this result simply as a numerical observation that may merit further investigation when studying how PTM conventions and normalization choices influence PAC-Bayes generalization bounds for quantum models.

\end{document}

%% file: commands.tex
\newcommand{\Id}{\ensuremath{\mathbb{I}}}

\newcommand{\uu}{\ensuremath{\vec{u}}}
\newcommand{\ww}{\ensuremath{\vec{w}}}
\newcommand{\xx}{\ensuremath{\vec{x}}}

\renewcommand{\norm}[1]{\left\lVert#1\right\rVert}

\DeclareMathOperator{\poly}{poly}

\renewcommand{\exp}{\ensuremath{\mathrm{exp}}}

\providecommand{\hL}{\ensuremath{\hat{L}}}

\providecommand{\to}{\ensuremath{\Tilde{o}}}

\providecommand{\tr}{\ensuremath{\Tilde{r}}}

\providecommand{\tt}{\ensuremath{\Tilde{t}}}

\providecommand{\calC}{\ensuremath{\mathcal{C}}}
\providecommand{\calD}{\ensuremath{\mathcal{D}}}

\providecommand{\calH}{\ensuremath{\mathcal{H}}}

\providecommand{\calX}{\ensuremath{\mathcal{X}}}

\providecommand{\bbN}{\ensuremath{\mathbb{N}}}

\providecommand{\bbP}{\ensuremath{\mathbb{P}}}

\providecommand{\bbR}{\ensuremath{\mathbb{R}}}
